\documentclass[%
 reprint,
 superscriptaddress,
showpacs,
 amsmath,amssymb,
 aps,
 pra,
]{revtex4-1}

\usepackage{graphicx}
\usepackage{dcolumn}
\usepackage{bm}

\usepackage[T1]{fontenc}
\usepackage[usenames,dvipsnames,svgnames]{xcolor}
\usepackage{xspace}
\usepackage{amsthm}

\makeatletter
\newtheorem*{rep@theorem}{\rep@title}
\newcommand{\newreptheorem}[2]{%
\newenvironment{rep#1}[1]{%
 \def\rep@title{#2 \ref{##1}}%
 \begin{rep@theorem}}%
 {\end{rep@theorem}}}
\makeatother

\newtheorem{theorem}{Theorem}
\newtheorem{lemma}[theorem]{Lemma}
\newtheorem{corollary}[theorem]{Corollary}
\newreptheorem{corollary}{Corollary}
\newtheorem{proposition}[theorem]{Proposition}

\theoremstyle{remark}

\newtheorem{remark}[theorem]{Remark}

\newcommand{\xqedhere}[2]{%
  \rlap{\hbox to#1{\hfil\llap{\ensuremath{#2}}}}}

\renewcommand{\Re}{\mathfrak{Re}}
\renewcommand{\Im}{\mathfrak{Im}}
\renewcommand{\vec}{\operatorname{vec}}

\newcommand{\Z}{\ensuremath{\mathbb Z}}
\newcommand{\R}{\ensuremath{\mathbb R}}
\newcommand{\C}{\ensuremath{\mathbb C}}

\usepackage{bbm}
\newcommand{\unity}{\mathbbmss{1}}

\newcommand{\ket}[1]{\ensuremath{| #1 \rangle}{}}

\newcommand{\ketbra}[2]{\ensuremath{| #1 \rangle \langle #2 |}{}}

\newcommand{\expt}[1]{\ensuremath{\langle #1 \rangle}{}}

\newcommand{\adr}{\operatorname{ad}}
\newcommand{\Adr}{\operatorname{Ad}}

\renewcommand{\vec}{\operatorname{vec}{}}

\newcommand{\Alt}{\mathrm{Alt}}
\newcommand{\Sym}{\mathrm{Sym}}
\newcommand{\uu}{\mathfrak{u}}
\newcommand{\su}{\mathfrak{su}}
\newcommand{\SU}{\mathrm{SU}}
\newcommand{\SO}{\mathrm{SO}}

\newcommand{\SP}{\mathrm{Sp}}
\newcommand{\U}{\mathrm{U}}
\newcommand{\PSU}{\mathrm{PSU}}

\newcommand{\so}{\mathfrak{so}}
\newcommand{\spp}{\mathfrak{sp}}

\newcommand{\fa}{\mathfrak{a}}
\newcommand{\fb}{\mathfrak{b}}

\newcommand{\fc}{\mathfrak{c}}
\newcommand{\fe}{\mathfrak{e}}
\newcommand{\ff}{\mathfrak{f}}
\newcommand{\fg}{\mathfrak{g}}
\newcommand{\fk}{\mathfrak{k}}
\newcommand{\fh}{\mathfrak{h}}
\newcommand{\fm}{\mathfrak{m}}

\newcommand{\fq}{\mathfrak{q}}
\newcommand{\fs}{\mathfrak{s}}
\newcommand{\fii}{\mathfrak{i}}
\newcommand{\fu}{\mathfrak{u}}

\newcommand{\bG}{\mathbf{G}}

\newcommand{\mF}{\mathcal{F}}

\newcommand{\cent}{\mathfrak{cent}}
\newcommand{\comm}{\mathfrak{comm}}
\newcommand{\centre}{\mathfrak{center}}

\newcommand{\TU}{\mathcal{U}}
\newcommand{\TUN}{\mathcal{V}}
\newcommand{\NN}{\mathcal{N}}
\newcommand{\CC}{\mathcal{C}_d}
\newcommand{\CCC}{\mathcal{C}'_d}
\newcommand{\Tm}{r}
\newcommand{\TM}{\hat{r}}
\newcommand{\DD}{\mathcal{D}}
\newcommand{\g}{\tilde{g}}
\newcommand{\oI}{{\ell_k^\unity}}
\newcommand{\oX}{{\ell_k^\mathrm{X}}}
\newcommand{\oY}{{\ell_k^\mathrm{Y}}}
\newcommand{\oZ}{{\ell_k^\mathrm{Z}}}
\newcommand{\hI}{h^\unity}
\newcommand{\hX}{h^\mathrm{X}}
\newcommand{\hY}{h^\mathrm{Y}}
\newcommand{\hZ}{h^\mathrm{Z}}

\renewcommand{\mod}{\operatorname{mod}{\;}}

\newcommand{\abs}[1]{\ensuremath{\vert #1 \vert}}

\newcommand{\tr}{\operatorname{tr}}

\newcommand{\reach}{\mathfrak{reach}{}}

\newcommand{\eR}{\mathcal{R}}
\newcommand{\hsodd}{h_{\mathrm{o}}}
\newcommand{\hseven}{h_{\mathrm{e}}}
\newcommand{\QFn}{\mathcal{QF}_n}
\newcommand{\Pdn}{\mathfrak{P}(d,n)}

\newcommand{\fop}{f}
\newcommand{\XX}{{\sf XX}}

\usepackage[bookmarks=false,pdfstartview={FitH}]{hyperref}

\makeatother

\begin{document}


\title{A Dynamic Systems Approach to Fermions and Their Relation to Spins}

\author{Zolt{\'a}n Zimbor{\'a}s}
\email{zimboras@gmail.com}
\affiliation{Institute for Scientific Interchange Foundation, 
Via Alassio 11/c, 10126 Torino, Italy}
\affiliation{Department of Theoretical Physics, University of the Basque Country UPV/EHU,
P.O. Box 644, E-48080 Bilbao, Spain}

\author{Robert Zeier}
\email{robert.zeier@ch.tum.de}
\affiliation{Department Chemie, Technische Universit{\"a}t M{\"u}nchen,
Lichtenbergstrasse 4, 85747 Garching, Germany}

\author{Michael Keyl} 
\email{michael.keyl@tum.de}
\affiliation{Institute for Scientific Interchange Foundation, 
Via Alassio 11/c, 10126 Torino, Italy}
\affiliation{Zentrum Mathematik, M5, Technische Universit{\"a}t M{\"u}nchen,
Boltzmannstrasse 3, 85748 Garching, Germany}

\author{Thomas \surname{Schulte-Herbr{\"u}ggen}}
\email{tosh@ch.tum.de}
\affiliation{Department Chemie, Technische Universit{\"a}t M{\"u}nchen,
Lichtenbergstrasse 4, 85747 Garching, Germany}

\date{December 20, 2013}


\begin{abstract}
Dynamic properties of fermionic systems, like controllability, reachability, 
and simulability, are investigated in a general Lie-theoretical frame for quantum systems theory. 
Observing the  parity superselection rule, we treat
the fully controllable and quasifree cases, 
as well as various translation-invariant and particle-number conserving
cases.
We determine the respective dynamic system Lie algebras to express 
reachable sets of pure (and mixed) states by explicit orbit manifolds.  
\end{abstract}

\maketitle

\section{Introduction}

Over the last decade, there has been a considerable experimental 
progress in achieving coherent control of ultra-cold
gases including fermionic systems \cite{GMEHB02, GRJ03, JB+03, ZSS+03, BDZ08, LBL12}.
This is also of great interest in view of quantum simulation (e.g., \cite{JC03})
of quantum phase transitions \cite{Sachdev99, QPT10}, 
pairing phenomena \cite{Pairing},
and in particular, 
for understanding phases in Hubbard models \cite{Esslinger10}. --- Even 
earlier on, the simulation of fermionic systems on quantum computers
had been in focus \cite{AL97, BK02}. For either case, there are interesting
algebraic aspects going beyond the standard textbook approach \cite{AlicFann01},
some of which can be found in \cite{Vlasov03, kraus-pra71, kraus09, SRL12}.
Here we set out for a unified picture of quantum systems theory in a Lie-algebraic 
frame following the lines of \cite{ZS11} to pave the way for optimal-control methods
to be applied to fermionic systems.

It is generally
recognized that optimal control algorithms are
key tools needed for further advances in experimentally 
exploiting these quantum systems for simulation as well as for
computation \cite{DowMil03, dAll08, WisMil09, DYNAMO}. 
In the implementation
of these algorithms it is crucial to know before-hand 
to which extent the system can be controlled. The usual 
scenario (in coherent control) is that 
we are given a \emph{drift Hamiltonian} and a set of 
\emph{control Hamiltonians} with tunable strengths. 
The achievable operations will be characterized by the
\emph{system Lie algebra}, while the 
reachable sets of states are given by the respective \emph{pure state orbits}.
Dynamic Lie algebras and reachability questions
have been intensively studied in the literature for qudit 
systems \cite{ZS11,SchiSOLea02b,SchiSOLea02,AA03}. 
However, in the case of fermions these 
questions have to be reconsidered mainly due to the
presence of the \emph{parity superselection rule}.
Hence in a broader sense the present work on fermions 
can be envisaged also as a step towards 
quantum control theory for quantum simulation in the presence of 
superselection rules.

Apart from discussing the implications of
the parity superselection rule in the theory 
of dynamic Lie algebras and of pure-state orbits,
we will also treat the case when one imposes translation-invariance or
particle-number conservation.
Moreover,
the experimentally relevant case of  
quasifree fermions (with and
without translation invariance) is discussed
in detail. 
Since we interrelate fermionic systems
with the Lie-theoretical framework of quantum-dynamical systems, at times we will be
somewhat more explicit and put known results into a new frame.
The main results extend from general fermionic systems to the action of 
Hamiltonians with and without restrictions like quadratic interactions, translation invariance, reflection symmetry,
or particle-number conservation.

The paper itself is structured as follows:
In order to set a unified frame, we resume some basic concepts of 
Hamiltonian controllability of qudit systems in Sec.~\ref{quditcontrol}, since for comparison these 
concepts will subsequently be translated to their fermionic counterparts, starting
with the discussion of general fermionic systems in  Sec.~\ref{FQS}.

Then the new results are presented in the following six sections:
In Sec.~\ref{fully} we obtain the dynamic 
system algebra for \emph{general fermionic systems} respecting the parity 
superselection rule (see Theorem~\ref{general}
in Subsection~\ref{DSA}).
An explicit example  for a set of Hamiltonians that provides full controllability over the fermionic system 
is discussed in Subsection~\ref{ExAndDisc}.
Some general results on the controllability of fermionic and spin systems,
such as Theorem~\ref{thm:double-centralizer}, are relegated to Appendix~\ref{DoubleCentralizers}.
Following the same line, in Sec.~\ref{QUASI} we wrap up some 
known results on \emph{quasifree} fermionic systems 
in a general Lie-theoretic frame by streamlining the derivation
for the respective system algebra in Proposition~\ref{quasifree} of Sec.~\ref{QUASI}.
Corollary~\ref{cor_tensor_three} provides a most general
controllability condition of quasifree fermionic systems
building on the tensor-square representation used in \cite{ZS11}.
Furthermore, we develop methods for restricting the set of possible system algebras by 
analyzing their rank, see Theorem~\ref{quad_ex_u} 
as well as Appendices~\ref{appl_quasi} and \ref{proof_u_d}.
The structure and orbits of \emph{pure states} in quasifree fermionic systems are analyzed in
Sec.~\ref{pure-quasifree} leading to a complete characterization 
of pure-state controllability (Theorem~\ref{controllability_quasifree_pure}).
Sections~\ref{sec:TI} and \ref{TIQuasi} are devoted to \emph{translation-invariant systems}.
For \emph{spin chains} 
we give in Theorem~\ref{translation_spins} the first full
characterization of the corresponding system algebras 
and strengthen in Theorem~\ref{Thm:SpinShortRange} earlier results on short-range 
interactions in  \cite{kraus-pra71}.
The system algebras for {\em general translation-invariant fermionic chains}
are given in Theorem~\ref{translationinvariant} of Sec.~\ref{sec-transl-fermionic}.
We also identify translation-invariant fermionic Hamiltonians of bounded interaction length 
which cannot be generated
from nearest-neighbor ones (see Theorem~\ref{thm:main_nn} of Sec.~\ref{ferm_nn}).
The particular case of \emph{quadratic} interactions  (see Sec.~\ref{transl_quadratic}) is settled in
Theorem~\ref{thm:qf-trans-inv}. Corollary~\ref{Rsym} considers systems which additionally carry a
twisted reflection symmetry (or equivalently have no imaginary hopping terms)
as discussed in \cite{kraus-pra71}.
Furthermore, we provide a \emph{complete} classification of
all pure quasifree state orbits in Theorem~\ref{orbit_pure_translation} of Sec.~\ref{orbits-q-trans}.
This leads to Theorem~\ref{thm:gap_close} of Sec.~\ref{applManyBody} presenting
a bound on the scaling of the gap for a class of quadratic Hamiltonians
which are translation-invariant.
Section~\ref{sec:PC} deals with \emph{fermionic systems}  conserving
the \emph{number of particles}. 
Their system
algebras in the general case  as well as in the quasifree case 
are derived in Proposition~\ref{prop:particle_number} and
Proposition~\ref{quasifree_particle}, respectively.
Furthermore, a necessary and sufficient condition for quasifree pure-state controllability
in the particle-number conserving setting is provided by Theorem~\ref{thm:qf_pc_pure_cont}.

In Sec.~\ref{conclusion}, we summarize the main results 
as given in Theorem~\ref{general}, Corollary~\ref{cor_tensor_three},
as well as in
Theorems~\ref{controllability_quasifree_pure}, \ref{translation_spins}, \ref{Thm:SpinShortRange}, 
\ref{translationinvariant},
\ref{thm:main_nn}, \ref{thm:qf-trans-inv}, \ref{orbit_pure_translation}, \ref{thm:gap_close}, 
and \ref{thm:qf_pc_pure_cont}.
We conclude  leaving a number of 
details and proofs
to the Appendices in order to streamline the presentation.

\section{Basic Quantum Systems Theory of $N$-Level Systems}\label{quditcontrol}
As a starting point, consider the controlled 
Schr{\"o}dinger (or Liouville) equation
\begin{equation}\label{eqn:bilinear_contr}
       \dot{\rho}(t) = - [iH_u,\rho(t)] := -(iH_u \rho(t) - \rho(t) iH_u)
\end{equation}
driven by the Hamiltonian $H_u:= H_0 + \sum_{j=1}^m u_j(t) H_j$
and fulfilling the initial condition  $\rho_0:=\rho(0)$.
Here the \emph{drift term} $H_0$ 
describes the evolution of the unperturbed system, 
while the \emph{control terms} $\{H_j\}$ 
represent coherent manipulations from outside. 
Equation~\eqref{eqn:bilinear_contr} defines a
{\em bilinear control system} $\Sigma$  \cite{Elliott09}, as  
it is linear both 
in the \emph{density operator} $\rho(t)$ and in the \emph{control amplitudes} 
$u_j(t)\in\mathbb R$.
 
For a $N$-level system, the natural representation as hermitian operators over
$\mathbb C^N$ relates the Hamiltonians as generators of unitary time evolutions
to the Lie algebra $\fu(N)$ of skew-hermitian operators that 
generate the unitary group $\U(N)$ of propagators. 
Let $L:= \{iH_1, iH_2, \dots, iH_m\}$ be a subset of Hamiltonians seen as Lie-algebra elements. 
Then the smallest subalgebra (with respect to the
commutator $[A,B]:=AB-BA$)
of $\fu(N)$ containing $L$ is called the \emph{Lie closure} of
$L$ written as $\expt{iH_1, iH_2, \dots, iH_m}_{\rm Lie}$. Moreover,
for any element
$iH\in \langle iH_1, \ldots, iH_m \rangle_{\rm Lie}$, there exist \emph{control amplitudes} 
	\mbox{$u_j(t)\in\mathbb R$} with $j\in\{1,\ldots,m\}$
such that
\begin{equation}\label{contr}
\exp(-iH) = \mathcal{T} \int_{t=0}^{1} \exp\left[\sum_{j=1}^m -iu_j(t) H_j\right]\, \mathrm{d}t,
\end{equation}
where $\mathcal{T}$ denotes time-ordering.

Now taking the Lie closure over the system Hamiltonian and all control Hamiltonians of a
bilinear control system $(\Sigma)$
defines the 
{\em dynamic system Lie algebra} (or system algebra for short)
\begin{equation}
\fg_\Sigma := \expt{ iH_0, iH_j\,|\,j=1,2,\dots,m}_{\rm Lie}\;.
\end{equation}
It is the key to characterize the differential geometry of a dynamic system in terms of its 
complete set of Hamiltonian directions forming the tangent space to the time evolutions.
For instance, the condition for \emph{full controllability} of bilinear systems can readily be adopted from
classical systems \cite{SJ72, JS72, Bro72, Jurdjevic97} to the quantum realm such as to take
the form of
\begin{equation}
        \langle iH_0, iH_j \,|\,j=1,2,\dots,m\rangle_{\rm Lie} = \fu(N)\
\end{equation}
saying that a $N$-level quantum system is fully controllable if and only if its system algebra
is the full unitary algebra, which we will relax to $\su(N)$ in a moment.
This notion of controllability is also intuitive (recalling that the exponential map is surjective for 
compact connected Lie groups), as it requires that all Hamiltonian directions can be generated.

So in fully controllable closed systems, to every initial state $\rho_0$ the \emph{reachable set}  
is the entire unitary orbit 
$
\reach_{\rm full}(\rho_0):=\{U\rho_0 U^\dagger\;|\; U\in \U(N)\}
$.
With density operators being hermitian,
this means any final state $\rho(t)$ can be reached from any initial state $\rho_0$
as long as both of them share the same spectrum of eigenvalues (including multiplicities).
Thus the reachable set of $\rho_0$ equals the \emph{isospectral set} of $\rho_0$.

\begin{remark}
Interestingly, this notion is stronger than the requirement that from any given (normalized)
\emph{pure} state one can reach any other (normalized) \emph{pure} state, since it is well known
 \cite{AA03,SchiSOLea02b,SchiSOLea02} that for $N$ being even, all rank-one projectors are
already on the {\em unitary symplectic orbit}
\begin{equation}
\begin{split}
\reach(\ketbra{\psi_0}{\psi_0}) &= \{K\ketbra{\psi_0}{\psi_0} K^\dagger\,|\, K\in \SP(N/2)\}\\
	&=\{U\ketbra{\psi_0}{\psi_0} U^\dagger\,|\, U\in \SU(N)\}
\end{split}
\end{equation}
and $\SP(N/2)$ is a proper subgroup of $\SU(N)$.
\end{remark}

In general, the \emph{reachable set} to an initial state $\rho_0$ of  a dynamic system $(\Sigma$) with 
\emph{system algebra} $\fg_\Sigma$ is given by the orbit of the dynamic (sub)group 
$\bG_\Sigma:=\exp(\fg_\Sigma)\subseteq U(N)$ generated by the system algebra 
\begin{equation}\label{eqn:reachSigma}
\reach_{\Sigma}(\rho_0):=\{G\rho_0 G^\dagger\;|\; G\in \bG_\Sigma=\exp(\fg_\Sigma)\}.
\end{equation}

Thus the system algebra $\fg_\Sigma$ can be envisaged as the \emph{fingerprint} 
encoding all the dynamic
properties of a dynamic system $\Sigma$. Via the respective reachable sets (see, e.g., \cite{ZS11}) 
it is easy to see
that a coherently controlled dynamic system $\Sigma_A$ {\em can simulate the 
dynamics of another
system} $\Sigma_B$ if and only if the system algebra $\fg_{\Sigma_A}$ of the 
simulating system
$\Sigma_A$ encompasses the system algebra $\fg_{\Sigma_B}$ of the simulated 
system $\Sigma_B$,
\begin{equation}\label{eqn:simulability}
\fg_{\Sigma_A} \supseteq \fg_{\Sigma_B}.
\end{equation}
In \cite{ZS11}, we have analyzed the possibility of quantum simulation with respect to the 
dynamic degrees of freedom and have given a number of illustrating worked examples.

Next we describe dynamic symmetries of bilinear control systems
whose Hamiltonians are given by $\fm:=\{iH_\nu\}=\{iH_0, iH_1, \dots, iH_m\}$. 
The \emph{symmetry operators} $s$  are collected in the {\em centralizer} of $\fm$ in $\fu(N)$:
\begin{equation}
\cent(\fm):= \big\{s\in\fu(N)\,|\, [s, i H] = 0 \;\; \forall\, iH \in \fm \big\}.
\end{equation}
More generally, let $S'$ denote the {\em commutant} 
of a set $S$ of matrices, i.e.,  
the set of all complex matrices which commute simultaneously
with all matrices in $S$.
By Jacobi's identity $\big[[a,b],c\big]+\big[[b,c],a\big]+\big[[c,a],b\big]=0$
one gets two properties of the centralizer pertinent for our context: 
First, an element $s$ that commutes with all Hamiltonians $a,b \in \fm$ also commutes 
with their Lie closure $\fg_\Sigma:=\expt{\fm}_{\rm Lie}$ 
(i.e.~$\cent(\fm) \equiv \cent(\fg_\Sigma)$), as 
$[s,a] = 0$ and $[s,b] = 0$ imply $\big[s,[a,b]] = 0$.
Second, for any $u \in\fu(N)$, 
$[s_1,u] = 0$ and $[s_2,u] = 0$ imply $\big[[s_1,s_2],u\big] = 0$,
so the centralizer  
forms itself a Lie subalgebra to $\mathfrak{u}(N)$ 
consisting of all symmetry operators.
By construction, $\cent(\fg_\Sigma)$ is also a 
\emph{normal subalgebra} or an \emph{ideal} of $\fu(N)$ observing
$[\cent(\fg_\Sigma),\fu(N)]\subseteq\cent(\fg_\Sigma)$.

Likewise one can describe the symmetries to a given set $\rho_\Sigma$ of states 
by its centralizer
\begin{align}
\cent(\rho_\Sigma) & :=\{ s\in\fu(N)\,|\, [s,\rho]=0\;\;\forall \rho\in\rho_\Sigma\}\\
                   & \phantom{:}= \cent(\langle \rho_\Sigma \rangle_{\R}), \nonumber
\end{align} 
where $\langle \,\cdot\, \rangle_{\R}$ denotes the real span.
Clearly, $\cent(\rho_\Sigma)\subseteq\fu(N)$ 
generates the \emph{stabilizer group} to the state space
$\rho_\Sigma$ of the control system~($\Sigma$).

Since in the absence of other symmetries 
the identity is the only and trivial symmetry of both any state space $\rho_\Sigma$ as well as
any set of Hamiltonians and their respective system algebra $\fg_\Sigma$, one has
$\cent(\fg_\Sigma)=\cent(\rho_\Sigma)
=\{i\,\lambda \unity_N\,|\, \lambda\in\mathbb R\}=:\fu(1)$.
So there is always a trivial stabilizer group $\U(1):=\{e^{i\phi}\unity_N\,|\,\phi\in\mathbb R\}$.
This explains why the time evolutions generated by two
Hamiltonians $H_1$ and $H_2$ coincide for the set of all density operators 
if (and without other symmetries only if) $H_1-H_2=\lambda \unity$.
As  is well known, by the same argument, in time evolutions 
\begin{equation}
	\rho(t):=U(t)\rho U^{\dagger}(t) = \Adr_{U(t)}(\rho_0)\;
\end{equation} 
following from Eq.~\eqref{eqn:bilinear_contr},
one may take $U(t):=\exp(-itH)$ equally well from $\U(N)$ or $\SU(N)$.
Thus henceforth we will only consider special unitaries (of determinant $+1$)
generated by traceless Hamiltonians $iH_\nu\in\su(N)$, since
for any Hamiltonian 
$\tilde H$ there exists an equivalent unique traceless Hamiltonian
$H:=\tilde H-\tfrac{1}{N}\tr (\tilde H)\unity_N$ generating a time evolution coinciding 
with the one of 
$\tilde H$ \footnote{
More precisely, all unitary conjugations of 
type $\Adr_U$ are elements of the 
projective special unitary group
$\PSU(N) = \U(N)/\U(1) \simeq \SU(N)/{\mathbb Z}(N)$,
where the centers of $\U(N)$ and $\SU(N)$
are respectively given by $\U(1)$ and 
${\mathbb Z}(N):=\{e^{ir}\unity_N\,|\,r\in \R \text{ with } rN \mod 2\pi=0\}$. 
Moreover, recall $\Adr_{\exp(-itH)}=e^{-it\adr_H}$, where  
$\adr_H:=[H,\cdot]$ can be represented
as \emph{commutator superoperator} 
$\adr_H = \unity_N \otimes H - H^t\otimes\unity_N$.
Now, for any  $H_1-H_2=\lambda \unity_N$, one immediately obtains $\adr_{H_1}=\adr_{H_2}$, 
which also elucidates that the generators of the projective 
unitaries are given by
$
\{i\adr_H\,|\,iH\in\fu(N)\}
$.}.

However, the above simple arguments are in fact much stronger, 
e.g., one readily gets the following statement:
\begin{lemma}\label{Lem:JointComm}
Consider a bilinear control system with system algebra $\fg_\Sigma$  
on a state space $\rho_\Sigma$. Let $iH_1\in \fg_\Sigma$ and $iH_2 \in \uu(N)$
while assuming that $[H_1,\langle\rho_\Sigma\rangle_{\R}]
\subseteq i\langle\rho_\Sigma\rangle_{\R}$ for all
$iH_1 \in \fg_\Sigma$, i.e., operations generated by $\fg_\Sigma$ map the set 
$\langle\rho_\Sigma\rangle_{\R}$ 
 into itself.
Then the condition
\begin{equation}\label{lem_cond}
e^{-iH_1t}\rho e^{iH_1t}=e^{-i(H_1+H_2)t}\rho e^{i(H_1+H_2)t}\;
\forall t\in \R, \rho\in\rho_\Sigma
\end{equation}
is equivalent to $iH_2 \in \cent(\rho_\Sigma)$.
\end{lemma}
\begin{proof}
Using the formula $e^{tA} B e^{-tA}=\exp[\adr_{tA}(B)]=\sum_{k=0}^{\infty} t^k/k! \adr^k_{A}(B)$ we 
show that Eq.~\eqref{lem_cond}
is equivalent to condition (a): $\adr^k_{H_1}(\rho)=\adr^k_{H_1+H_2}(\rho)$
for all non-negative integer $k$ and all $\rho\in\langle\rho_\Sigma\rangle_{\R}$.
Moreover, (a) implies condition (b): 
$(\adr_{H_2} \circ \adr^k_{H_1})(\rho) = 0$ for all non-negative integer $k$ and all 
$\rho\in\langle\rho_\Sigma\rangle_{\R}$, as
$[H_1,\adr^{k-1}_{H_1}(\rho)]=[H_1+H_2,\adr^{k-1}_{H_1+H_2}(\rho)]
=[H_1+H_2,\adr^{k-1}_{H_1}(\rho)]$. Also, (a) follows from (b) due to
$\adr^k_{H_1}(\rho)=[H_1+H_2,\adr^{k-1}_{H_1}(\rho)]=[H_1+H_2,[H_1+H_2,\adr^{k-2}_{H_1}(\rho)]]
=\cdots=\adr^{k}_{H_1+H_2}(\rho)$.
Applying $[H_1,\langle\rho_\Sigma\rangle_{\R}]{\subseteq} i\langle\rho_\Sigma\rangle_{\R}$ to 
(b) completes the proof.
\end{proof}
Therefore, let us consider
a pair of Hamiltonians $iH_1,iH_3\in\fg_\Sigma$ 
(fulfilling the conditions of Lemma~\ref{Lem:JointComm})
as \emph{equivalent} on the state space $\rho_{\Sigma}$, 
if their difference $iH_2:=i(H_1-H_3)$ 
falls into the centralizer $\cent(\rho_\Sigma)$.

\section{Fermionic Quantum Systems\label{FQS}}

In this section, we fix our notation by 
recalling basic notions for fermionic systems.
In the first subsection,
we discuss the Fock space and different operators acting on it as
given by
the creation and annihilation operators as well as the Majorana operators.
We point out how the Lie algebra $\uu(2^d)$ of skew-hermitian matrices
can be embedded as a real subspace in the set of the complex operators acting on the Fock space.
In the second subsection, we focus on the parity superselection rule and how it structures 
a fermionic system.

\subsection{The Fock Space and Majorana Monomials\label{fock_majorana}}
The complex Hilbert space of
a $d$-mode fermionic system with one-particle
subspace $\C^d$ is the 
\emph{Fock space}
\begin{equation*}
\mathcal{F}(\C^{d}):=
\bigoplus_{i=0}^{d} 
\left(\bigwedge^i \C^{d}\right)
=\C \oplus \C^{d} \oplus
 \wedge^2 \C^{d} \oplus \cdot\cdot  \oplus  (\wedge^d \C^{d}).
\end{equation*}
Given an orthonormal basis $\{e_i\}_{i=1}^d$ of $\C^d$, the
\emph{Fock vacuum} $\Omega:= 1\; (= 1\oplus 0  \oplus \cdots \oplus 0)$
and the vectors of the form
$e_{i_1} \wedge e_{i_2}\wedge \cdots \wedge e_{i_k}$ 
(with $i_1<i_2 < \cdots < i_{k}$ and $1 \leq k \le d$)
form an orthonormal basis of $\mathcal{F}(\C^{d})$.
Note that $\mathcal{F}(\C^{d})$ is a $2^d$-dimensional Hilbert space 
isomorphic to $\otimes_{i=1}^d\C^2\;
(\cong \C^{2^d})$.

The fermionic \emph{creation} and \emph{annihilation operators},
$\fop^\dagger_p$ and $\fop_p$ act on the Fock space in the following way:
$\fop^\dagger_p \Omega = e_p$, $\fop_p \Omega =0$, 
$\fop^\dagger_p e_q = e_p \wedge e_q$, and $\fop_p e_q = \delta_{pq}$;
while in the general case of $1\leq \ell \leq d$, their action is given 
by 
$\fop^\dagger_p (e_{q_1} \wedge e_{q_2}\wedge \cdots \wedge e_{q_\ell})=
(e_p \wedge e_{q_1} \wedge e_{q_2}\wedge \cdots \wedge e_{q_\ell})$ and
$\fop_p (e_{q_1} \wedge e_{q_2}\wedge \cdots \wedge e_{q_\ell})= 
\sum_{k=1}^{n} (-1)^k\delta_{pq_k} \, e_{q_1} 
\wedge \cdots \wedge e_{q_{(k-1)}} \wedge 
e_{q_{(k+1)}} \wedge \cdots \wedge e_{q_\ell}$.
By their definition, these operators 
satisfy the fermionic \emph{canonical anticommutation relations}
\begin{equation*}
\{\fop^{\dagger}_p, \fop^{\dagger}_q\} = \{\fop_p, \fop_q\} =0\; \text{ and }\;
\{\fop^{\dagger}_p, \fop_q\} = \delta_{pq} \unity,
\end{equation*}
where $\{A,B\}:=AB+BA$ denotes the anticommutator.
Moreover, every linear operator
acting on  $\mathcal{F}(\C^{d})$ can be 
be written as a complex polynomial in the creation and annihilation operators. 

Another set of polynomial generators 
acting on the Fock space is given by the $2d$ hermitian \emph{Majorana operators}
$m_{2p-1}:=\fop_p + \fop^\dagger_p$
and
$m_{2p}:=i(\fop_p-\fop^\dagger_p)$, which satisfy
the relations ($k,\ell\in\{1,\ldots,2d\}$)
\begin{equation*}
\{m_k, m_\ell\} = 2\delta_{k\ell} \unity.
\end{equation*}
A product $m_{q_1}m_{q_2}\cdots m_{q_k}$ of $k\geq 0$ Majorana operators is
called a \emph{Majorana monomial}.
The ordered Majorana monomials with $q_1 < q_2 < \cdots <q_k$
form a linearly independent basis of the complex operators 
acting on $\mathcal{F}(\C^{d})$. 
Each Majorana monomial acting on $d$-mode fermionic system
can be identified with a complex operator acting on a chain of $d$ qubits
via the \emph{Jordan-Wigner transformation} \cite{JW28,Boerner67,Miller72,SW86}
which is induced by
\begin{align*}
m_{2p-1}&\mapsto \underbrace{\mathrm{Z}\otimes\cdots\otimes\mathrm{Z}}_{p-1}\otimes
\mathrm{X}\otimes\underbrace{\mathrm{I}\otimes\cdots \otimes\mathrm{I}}_{d-p} 
\intertext{and}
m_{2p}&\mapsto
\underbrace{\mathrm{Z}\otimes\cdots\otimes\mathrm{Z}}_{p-1}\otimes\mathrm{Y}\otimes
\underbrace{\mathrm{I}\otimes\cdots \otimes\mathrm{I}}_{d-p} \, ,
\end{align*}
where the following notation for the Pauli matrices 
$\mathrm{X}:=
\left(
\begin{smallmatrix}
0 & \phantom{-}1\\
1 & \phantom{-}0
\end{smallmatrix}
\right)
$,
$\mathrm{Y}:=
\left(
\begin{smallmatrix}
0 & -i\\
i & \phantom{-}0
\end{smallmatrix}
\right)
$, and
$\mathrm{Z}:=
\left(
\begin{smallmatrix}
1 & \phantom{-}0\\
0 & -1
\end{smallmatrix}
\right)
$ is used.

Now we highlight the real subspace contained
in the set of complex operators acting on the Fock space $\mathcal{F}(\C^{d})$
which consists of all skew-hermitian operators
and which forms the real Lie algebra $\uu(2^d)$ closed
under the commutator $[A,B]=AB-BA$ and real-linear combinations.
More precisely, $\uu(2^d)$ is generated by all operators
\begin{equation}\label{eqn:Lv1}
L(M):=-\frac{1}{2} w(M) M\;,
\end{equation}
where $M$ denotes any ordered Majorana monomial
and 
\begin{equation}\label{eqn:Lv2}
w(M):=
\begin{cases}
\phantom{-}i & \text{if }\; [\deg(M) \mod 8] \in \{0,1\}, \\
\phantom{-}1 & \text{if }\; [\deg(M) \mod 8] \in \{2,3\}, \\
-i & \text{if }\; [\deg(M) \mod 8] \in \{4,5\}, \\
-1 & \text{if }\; [\deg(M) \mod 8] \in \{6,7\}. \\
\end{cases}
\end{equation}
Similarly, one obtains a basis of $\su(2^d)$ by excluding $-\tfrac{i}{2}\unity$.

\subsection{Parity Superselection Rule\label{parity}}

An additional fundamental ingredient in describing fermionic
systems is the \emph{parity superselection rule}.
Superselection rules were originally introduced
by  Wick, Wightman, and Wigner \cite{WWW52} 
(see also \cite{Wightman95,Earman08}). These rules,
in the finite-dimensional
definition of Piron \cite{Piron76},
describe the existence of non-trivial
observables that commute with \emph{all 
physical observables}.
The existence of
such a commuting observable in turn implies that a 
superposition of pure states from different blocks of 
a block-diagonal decomposition w.r.t.\
the eigenspaces of this observable
are equivalent to an incoherent classical mixture.

The \emph{parity superselection rule} identifies among 
the operators acting 
on $\mathcal{F}(\C^{d})$ the \emph{physical observables} $\mathbb H_F$
as those that do commute with the parity operator 
\begin{equation}\label{P_def}
P:=i^d\prod_{k=1}^{2d}m_k \,,
\end{equation}
where the adjoint action of $P$ on a Majorana monomial
is given as
$P m_{k_1}m_{k_2} \cdots m_{k_\ell}P^{-1}=(-1)^{\ell} m_{k_1}m_{k_2} \cdots m_{k_\ell}$.
These physical operators are also exactly the ones that
can be written as a sum of products of an \emph{even} number 
of Majorana operators (as $P$ contains all Majorana operators
whereof there exist an even number).
They are therefore 
 denoted as \emph{even operators} for short.
If the parity is the \emph{only} non-trivial symmetry, we obtain
$
{\mathbb H}_F' = \expt{\unity, P}
$, where the bracket stands for the complex-linear span.

Now we will discuss why
the set of all physical fermionic states $\rho_F$ consists similarly of all density operators
that commute with~$P$, notably
$
\rho_F'= \expt{\unity, P}.
$
As we will show,
the parity super\-selection rule induces a decomposition into a direct sum of two \emph{irreducible}
state-space components exploiting
$
{\mathbb H}_F'  \cap \rho_F'\, = \expt{\unity, P}.
$
Recall that $P^2=\unity$ and the eigenspaces to the eigenvalues $+1$ and $-1$ are
indeed of equal dimension, as there are exactly $2^{2d-1}$ 
even operators which map the vacuum state $\Omega$ into the $+1$ eigenspace of $P$.
Note that 
$P e_{q_1} \wedge e_{q_2}\wedge \cdots \wedge e_{q_\ell}=(-1)^{\ell}e_{q_1} \wedge e_{q_2}\wedge 
\cdots \wedge e_{q_\ell}$.
Thus the Fock space can be split up as a direct sum of 
two equal-dimensional eigenspaces of $P$, called the 
\emph{positive} and \emph{negative parity subspaces} (for clarity observe \footnote{We use this
notation in contrast to the notation of even and odd subspaces (which
is also used in the literature) in order to avoid any confusion 
with the even operators.}):
\begin{align*}
\mathcal{F}(\C^{d})
= \left[\bigoplus_{\text{$i$ even}}
\left(\bigwedge^i \C^{d}\right)\right]
\bigoplus 
\left[\bigoplus_{\text{$i$ odd}}
\left(\bigwedge^i \C^{d}\right)\right].
\end{align*}

Now we may write $P^2=\unity=P_++P_-$ with
the orthogonal projections $P_+:=\frac{1}{2}(\unity+P)$ and 
$P_-:=\frac{1}{2}(\unity-P)$ projecting onto the respective subspaces. 
Any physical observable (i.e.~even operator) $A$ has a 
block-diagonal structure with 
respect to the above splitting, i.e.\ $A=P_+AP_+ 
+ P_- A P_-$. This follows, as 
the requirement $[A,P]=\frac{1}{2}[A,P_+]=-\frac{1}{2}[A,P_-]=0$ 
enforces $P_+AP_- = P_-AP_+=0$ for 
any operator
$A=P_+AP_+ + P_+AP_- + P_-AP_+
+ P_- A P_-$.
We obtain
\begin{align}\label{evenstate}
\nonumber
\text{Tr}(\rho A) &= \text{Tr}(\rho P_+ A P_+ + 
\rho P_- A P_-)\\ &= \text{Tr}[(P_+ \rho P_+ 
+ P_- \rho P_-) A].
\end{align}
Hence \emph{physical} observables cannot distinguish between the density operator $\rho$
and its block-diagonal projection to $P_ +\rho P_+ + P_- \rho P_-$ (which is 
always an even density operator). 
In this sense, a physical linear combination (a formal
superposition) of pure states from the positive and negative parity subspaces  
is equivalent to an incoherent classical mixture. 
Equation~\eqref{evenstate} also shows that without loss of generality
we can restrict ourselves to even density operators and 
regard only those as physical.

Finally, we would like to recall three further aspects of the 
parity superselection rule. First, without the parity superselection rule,
two noncommuting observables acting 
on two different and spatially-separated regions would 
exist which would allow for a violation of locality (e.g., by
instantaneous signaling between the regions). 
Second, the parity superselection rule, of course,  
does not apply if one uses a spin system to
simulate a fermionic system 
via the Jordan-Wigner transformation.
This system respects locality, since the Majorana 
operators $m_{k}$ 
are---in this case---localized on 
the \emph{first} $[(k{+}1)\, \mathrm{div}\, 2]$ spins;
two non-commuting Majorana operators are therefore not acting 
on spatially-separated 
regions. Third, the parity superselection rule also affects the concept
of entanglement as has been pointed out and
studied in detail in \cite{BCW07,KKSCH}.

\section{Fully Controllable Fermionic Systems\label{fully}}

Here we derive a general controllability result for fermions obeying the parity superselection 
rule. We illustrate that full controllability for a 
fermionic system can be achieved with
quadratic Hamiltonians and a single fourth-order interaction term.
For example, in a system with $d$ modes, the complete fermionic dynamical algebra
 $\mathcal{L}_d \cong \su(2^{d-1}) \oplus \su(2^{d-1})$ (see Theorem~\ref{general}) can 
be generated by 
a quartic interaction between the first two modes $ih_{\rm{int}}=i(2f^{\dagger}_1f^{\phantom\dagger}_1-\unity)(2f^{\dagger}_2f^{\phantom\dagger}_2-\unity)=-im_1m_2m_3m_4$
combined with three quadratic Hamiltonians which are:
the nearest-neighbor hopping term
\begin{align*}
ih_{\rm{h}}&= -2i\sum_{p=1}^{d-1} f^{\dagger}_pf^{\phantom\dagger}_{p+1} {+} f^{\dagger}_{p+1}f^{\phantom\dagger}_{p}\\
&=\sum_{p=1}^{d-1}{-} m_{2p-1} m_{2p+2} {+} m_{2p} m_{2p+1},
\end{align*}
the on-site potential of the first site $ih_0{=}i(2f^{\dagger}_1f^{\phantom\dagger}_1{-}\unity)=m_1m_2$, and a
pairing-hopping term between the first two modes $ih_{12}=i(f^{\phantom\dagger}_1f^{\phantom\dagger}_2-f^{\dagger}_1f^{\dagger}_2) -i(f^{\dagger}_1f^{\phantom\dagger}_2-f^{\phantom\dagger}_1f^{\dagger}_2)=m_2m_3$ (see Proposition~\ref{prop_gen_ex}).
Finally, we provide a general discussion about 
when the commutant 
of a system algebra determines the algebra itself.

\subsection{System Algebra\label{DSA}}

In the case of qubit systems mentioned in Sec.~\ref{quditcontrol},  
two Hamiltonians generate equivalent time evolutions 
if and only if they differ by a multiple of the identity. 
This condition can readily be modified for the fermionic case 
such as to match the parity-superselection rule as well.

\begin{corollary}\label{equal_time_evol}
Let $H_1$  and $H_2$ be two 
physical fermionic Hamiltonians, i.e., even hermitian 
operators acting on $\mathcal{F}(\C^{d})$. 
Then by Lemma~\ref{Lem:JointComm} the equality
\begin{align*}
e^{-iH_1t}\rho e^{iH_1t}=e^{-iH_2t} \rho e^{iH_2t}
\end{align*} 
holds for all even (physical) density operators $\rho_F$ with 
$\rho_F'=\expt{\unity,P}$
in the sense that
$H_1$ and $H_2$ generate the same time-evolution, 
if and only if 
$
H_2-H_1=\lambda \unity + \mu P
	   =(\lambda+\mu)P_+ + (\lambda-\mu) P_-$ with $\lambda,\mu\in \mathbb R$.
\end{corollary}

This also implies that for any physical 
fermionic Hamiltonian $H$, there exists a unique 
Hamiltonian 
\begin{align}
\tilde H:=H - \tfrac{\tr(P_+HP_+)}{\dim P_+^{\phantom{\dagger}} }\;P_+ 
	- \tfrac{\tr(P_-HP_-)}{\dim P_-^{\phantom{\dagger}} }\;P_-
\end{align}
that is traceless on \emph{both} the positive and the negative
parity subspaces, i.e.,
\begin{equation}\label{pos-neg-traceless}
\tr(P_+\tilde H P_+)=\tr(P_-\tilde H P_-)=0\, ,
\end{equation} 
and moreover, $\tilde H$ and $H$ are equivalent and generate the same time evolution.
If necessary, we can restrict ourselves to the 
set of Hamiltonians  satisfying  Eq.~\eqref{pos-neg-traceless}. 
These elements decompose as 
$H=H_+ \oplus H_-$, where $H_+$ and $H_-$ are generic traceless
hermitian operators each acting on a
$2^{d-1}$-dimensional Hilbert space.
We explicitly define the linear space
$\mathbb{F}_{d}$
of physical fermionic Hamiltonians
as generated by the basis of all even Majorana monomials 
without the operators $\unity$ and 
$P$, ensuring that $\mathbb{F}_{d}$ is traceless both on
$H_+$ and $H_-$.---We summarize our exposition on fully controllable fermionic systems
in the following result:
\begin{theorem}\label{general}
The 
Lie algebra corresponding to the physical fermionic 
(and hermitian) Hamiltonians $\mathbb{F}_{d}$ is  
\begin{equation}
\mathcal{L}_d:=\su(2^{d-1})\oplus \su(2^{d-1}).
\end{equation}
The most general set of unitary transformations 
generated by $\mathcal{L}_d$
is given as the block-diagonal
decomposition $\SU(2^{d-1})\oplus \SU(2^{d-1})$. 
Hence 
a set $\{H_0, H_1, H_2, \ldots, H_m \}$ 
of hermitian Hamiltonians defines a 
fully controllable fermionic system if and only if
\begin{align}\label{f-controllability}
\langle iH_0, iH_1, \ldots, iH_m \rangle_{\rm Lie} =\su(2^{d-1})\oplus 
\su(2^{d-1})\;.
\end{align}
\end{theorem}

\begin{remark}\label{remark5}
For Lie algebras, $\fk_1 + \fk_2$ will denote only
an abstract direct sum  without referring to any concrete realization.
We reserve the notation $\fk_1 \oplus \fk_2$
to specify
a direct sum of Lie algebras which is
(up to a change of basis) represented in a block-diagonal form
$
\left(
\begin{smallmatrix}
\fk_1\\
& \fk_2
\end{smallmatrix}
\right)
$.
\end{remark}

\begin{proof}
It follows from Sec.~\ref{FQS} that $\mathbb{F}_{d}$ commutes with
$P$ and that the matrix representation of $\mathbb{F}_{d}$
splits into two blocks of dimension $2^{d-1}$ corresponding to the $+$ and $-$
eigenspaces of $P$. As the center of $\mathbb{F}_{d}$ is given by 
$\mathbb{F}_{d}' \cap \mathbb{F}_{d} = \expt{\unity,P} \cap \mathbb{F}_d = \{0\}$, 
the Lie algebra $\mathbb{F}_d$ is semisimple.
As there are exactly $2^{2d-1}-2$ linear-independent operators 
in $\mathbb{F}_d$, the system algebra could
be $\su(2^{d-1})\oplus \su(2^{d-1})$. And indeed, all other system algebras are ruled out as 
the subalgebras acting on each of the two matrix blocks would
have a smaller Lie-algebra dimension than $\su(2^{d-1})$.
\end{proof}

\subsection{Examples and Discussion\label{ExAndDisc}}

We start out with an example realizing a
fully controllable fermionic system by adding only
one quartic operator to the set of quadratic Hamiltonians
which will be discussed in Section~\ref{QUASI} below
(cf.\ Theorem~\ref{quad_ex_so}):
\begin{proposition}\label{prop_gen_ex}
Consider a fermionic quantum system with
$d > 2$ modes. The system algebra $\mathcal{L}_d=\su(2^{d-1})\oplus 
\su(2^{d-1})$ of a fully controllable 
fermionic system can be generated using the operators
$w_1:=L(v_1)$, $w_2:=L(v_2)$, $w_3:=L(v_3)$, and $w_4:=L(v_4)$ with
the map $L$ as defined in Eqs.~\eqref{eqn:Lv1} and \eqref{eqn:Lv2}, where
\begin{subequations}\label{v_ops}
\begin{gather}
v_1 := \sum_{p=1}^{d-1} - m_{2p-1} m_{2p+2} + m_{2p} m_{2p+1},\\
v_2 := m_1 m_2,\,
v_3 := m_2 m_3,\, 
v_4 := m_1 m_2 m_3 m_4.
\end{gather}
\end{subequations}
\end{proposition}
\begin{proof}
It follows from the independent Theorem~\ref{quad_ex_so} (see Sec.~\ref{QUASI} below)
that $w_1$, $w_2$, and $w_3$ generate all quadratic Majorana monomials $m_p m_q$.
Consider an even Majorana monomial
$s_1:=L(\prod_{i\in \mathcal{I}} m_i)$  of degree $2d'$, where $s_2$ is defined
using the ordered index set $\mathcal{I}$, and a  quadratic operator $s_2:=L(m_p m_q)$
with $p \in \mathcal{I}$ and $q \notin \mathcal{I}$.
We can change any index $p$ of $s_1$ into $q$ of using
$L(\prod_{k\in (\mathcal{I}\setminus \{p\} ) \cup \{q\}} m_k)=\pm[s_1,s_2]$.
Therefore, we get from 
$w_4$ and the quadratic operators
all Majorana monomials of degree four.

Using the quartic Majorana monomials
 we can increase the degree of the monomials in steps of two: 
Consider the operators
$s_3:=L(\prod_{i\in \mathcal{I}} m_i)$ and $s_4:=L(\prod_{j\in \mathcal{J}} m_j)$
which are defined using the
ordered index sets $\mathcal{I}$ and $\mathcal{J}$
and have degrees $2d''<2(d-1)$ and $4$, respectively.
Assuming that $\abs{\mathcal{I} \cap \mathcal{J}}=1$, we can generate
an operator $L(\prod_{k\in \mathcal{K}} m_k)=\pm[s_3,s_4]$ 
of degree $\abs{\mathcal{K}}=2(d''+1)<2d$
where the corresponding ordered index set is given by $\mathcal{K}:=
(\mathcal{I} \cup \mathcal{J})\setminus(\mathcal{I} \cap \mathcal{J})$.
By induction, we can now generate all even Majorana monomials except $L(\prod_{q=1}^{2d} m_q)$.
Note that $L(\prod_{q=1}^{2d} m_q)$ cannot be obtained
as $\mathcal{I} \cap \mathcal{J} \nsubseteq \mathcal{K}$ holds by construction.
Thus, we get all elements of $\mathcal{L}_d$ (see Subsection~\ref{DSA})
and the proposition follows.
\end{proof}

The proof also implies that all the operators generated
commute with 
$\prod_{q=1}^{2d} m_q=P/i^d$ [cf.~Eq.~\eqref{P_def}] (and the identity operator $\unity$). 
In addition,
all operators commuting simultaneously with all elements of  $\mathcal{L}_d$ 
can be written as a complex-linear combination
of $\unity$ and $P$. We thus obtain a partial characterization of full controllability
in fermionic systems:
\begin{lemma}\label{nec_fer}
Consider a fermionic quantum system with
$d \geq 2$ modes. A necessary condition
for full controllability of a given set of hermitian Hamiltonians
$H_v$
is that $\{i H_v\}'=  \langle \unity,P \rangle$.
\end{lemma}

One can expect that the condition of Lemma~\ref{nec_fer}
is not sufficient under any reasonable assumption by
applying counterexamples from spin systems in \cite{ZS11}.
These counterexamples 
could be lifted to fermionic systems by providing the explicit form 
of the embeddings from $\su(2^{d-1})$ to the first and second 
component of the direct sum $\mathcal{L}_d=\su(2^{d-1})\oplus 
\su(2^{d-1})$.

We guide the discussion in a different direction
by emphasizing that the property $\{i H_v\}' = \langle \unity,P \rangle$
does not determine the system algebra uniquely.
We define the centralizer of a set $B\subseteq\su(k)$ in $\su(k)$ (e.g.\ $k=2^d$) as
$$\cent_{\su(k)}(B):=\{ g\in\su(k) \, | \, [g, b]=0 \text{ for all } b \in B \}.$$
We consider the algebras $\mathcal{L}_d=\su(2^{d-1}) \oplus \su(2^{d-1})$
and $\mathrm{s}[\uu(2^{d-1}) \oplus \uu(2^{d-1})]$, where 
the latter algebra is isomorphic to
$\su(2^{d-1}) + \su(2^{d-1}) + \uu(1)$ and contains the additional
(non-physical) generator $L(\prod_{q=1}^{2d} m_q)$. Note 
that $\cent_{\su(k)}(\mathcal{L}_d)=
\cent_{\su(k)}(\mathrm{s}[\uu(2^{d-1}) \oplus \uu(2^{d-1})])
=L(\prod_{q=1}^{2d} m_q )$, i.e.,
the centralizers of both algebras are equal.
However  $\cent_{\su(k)}[L( \prod_{q=1}^{2d} m_q )]=\mathrm{s}[\uu(2^{d-1}) \oplus \uu(2^{d-1})]
\neq \su(2^{d-1})\oplus \su(2^{d-1})$. 
In particular,  we have $\mathcal{L}_d\neq\cent_{\su(k)}(\cent_{\su(k)}(\mathcal{L}_d))$,
and $\mathcal{L}_d$ does not fulfill
the double-centralizer property.
A more general incarnation of this effect in line with a discussion
of double centralizers is given in Appendix~\ref{DoubleCentralizers}.
It leads in the case of irreducible subalgebras to the following maximality result:
\begin{corollary}\label{cor_double}
Let $\fg$ denote an irreducible subalgebra of $\su(k)$, i.e.\ 
$\cent_{\su(k)}(\fg)=\{0\}$. Then one finds that
$\cent_{\su(k)}(\cent_{\su(k)}(\fg))=\fg$ if and only if $\fg=\su(k)$.
\end{corollary}

To sum up, the symmetry properties of a Lie algebra $\fg\subseteq \su(k)$,
as given by its commutant w.r.t.\ a representation of $\fg$,
do \emph{not} determine the Lie algebra $\fg$ uniquely. Yet the commutant allows us to infer a 
\emph{unique maximal Lie algebra}
contained in $\su(k)$, which is (up to an identity matrix) equal to the double commutant of $\fg$,
but in general not to $\fg$ itself.

\section{Quasifree Fermions\label{QUASI}}

Here we present the dynamic system algebras for fermions with
quadratic Hamiltonians. For illustration, also the relation to spin chains 
is worked out in detail. In this context, we show by free fermionic techniques
that a Heisenberg-XX Hamiltonian of Eq.~\eqref{xx} combined with the one-site term
$ih_{\rm{0}}=i
\mathrm{Z}\otimes \mathrm{I} \otimes \cdots \otimes \mathrm{I}=m_1m_2$  and the
two-site interaction
$ih_{12}=i
\mathrm{X}\otimes \mathrm{X} \otimes\mathrm{I} \otimes \cdots \otimes \mathrm{I}=m_2m_3$
gives rise to the system algebra $\so(2d)$ (see Theorem~\ref{quad_ex_so}), while
the first two operators generate only the subalgebra $\uu(d)$
(see Theorem~\ref{quad_ex_u}). Further results along this line are presented in 
Appendix~\ref{appl_quasi}.

Finally, we arrive at a very useful general result:
In order to decide if a set of operators generates the full quadratic algebra
for $d$ modes,
 we characterize quadratic operators 
by a real skew-symmetric matrix $T$ whose entries are given via
$-\tfrac{1}{2}\sum_{k,\ell}^{2d} T_{k\ell}\, m_k m_\ell$
(see Eq.~\eqref{stand_Maj}).
Adapting our tensor-square criterion for full controllability from spin systems \cite{ZS11}
to quasifree fermionic systems, a set of operators $T_\nu$ generates the full quadratic
algebra $\so(2d)$ if and only if the joint commutant of the operators
$T_\nu \otimes \unity_{2d} + \unity_{2d} \otimes T_\nu$
has dimension three (see Corollary~\ref{cor_tensor_three}).

\subsection{Quadratic Hamiltonians}

A general quadratic Hamiltonian of a fermionic system  
can be written as (cf.\ \cite{LSM61,Berezin,BR86,kraus-pra71,ZKZZ10}) 
\begin{equation}\label{Hqfree}
H=\sum_{p,q=1}^d A_{pq} (f_p^{\dagger}f_q{-}\delta_{pq} \tfrac{\unity}{2}) +
\frac{1}{2}B_{pq} f_{p}^{\dagger}f_{q}^{\dagger}-
\frac{1}{2}B_{pq}^{*} f_p f_{q},
\end{equation}
where the coupling coefficients $A_{pq}$ and $B_{pq}$ are complex entries of 
the $d\times d$-matrices $A$ and $B$, respectively.
The canonical anticommutation relations
and the hermiticity of $H$ require that $A$ is hermitian and $B$ is (complex) skew-symmetric.
The terms corresponding to the non-zero matrix entries
of $A$ and $B$ are usually referred to as \emph{hopping} 
and \emph{pairing} terms, respectively.
Related parameterizations for quadratic Hamiltonians are discussed 
in Appendix~\ref{ParaQuadratic}.

In the Majorana monomial basis, the
quadratic Hamiltonian $H$ can be rewritten such that
\begin{equation}\label{stand_Maj}
-i H= \sum_{k,\ell=1}^{2d} T_{k\ell}\, \left[ -\frac{1}{2}  m_k m_\ell \right]
\end{equation}
with
\begin{align*}
T=\frac{1}{2}& \left[ 
\Re(A) \otimes
\begin{pmatrix}
0 & 1\\
-1 & 0
\end{pmatrix}
+
\Re(B) \otimes
\begin{pmatrix}
0 & -1\\
1 & 0
\end{pmatrix}\right. \\
&\, \left.
+
\Im(A) \otimes
\begin{pmatrix}
-1 & 0\\
0 & -1
\end{pmatrix}
+
\Im(B) \otimes
\begin{pmatrix}
-1 & 0\\
0 & 1
\end{pmatrix}
\right].
\end{align*}
The properties of $A$ and $B$ directly imply that the matrix $T$ 
is real and skew-symmetric. Using the formula
\begin{subequations}\label{eq:quad-majorana}
\begin{gather}
[m_p m_q, m_r m_s]= - 4(\delta_{ps} \delta_{qr} \unity - \delta_{qs} 
\delta_{pr} \unity) \nonumber \\
+2(\delta_{ps}m_qm_r{-}\delta_{pr}m_qm_s{+}\delta_{qr}m_p m_s{-}\delta_{qs}m_pm_r)\\
=\delta_{ps} (m_q m_r - m_r m_q) - \delta_{pr} (m_q m_s - m_s m_q) \nonumber \\ +
\delta_{qr} (m_p m_s - m_s m_p) - \delta_{qs} (m_p m_r - m_r m_p) \label{eq:quad-majorana_b}
\end{gather}
\end{subequations}
one can easily verify that the space of quadratic Hamiltonians is closed under the 
commutator. To sum up, we have established the well-known Lie homomorphism from
the system algebra generated by  
a set of quadratic Hamiltonians (whose
control functions are given by the matrix entries of $A$ and $B$)
onto the system algebra $\so(2d)$ represented by the entries of $T$
(cf.\ pp.~183-184 of \cite{SW86}):
\begin{proposition}\label{quasifree}
The maximal system algebra for a system of quasifree
fermions with $d$ modes is given by $\so(2d)$.  
\end{proposition}
\begin{proof}
Let the map $h$ transform the Majorana monomial $-\tfrac{1}{2}(m_p m_q {-} m_q m_p)$ into
the skew-symmetric matrix $e_{pq}{-}e_{qp}$ where $e_{pq}$ has the matrix entries
$[e_{pq}]_{uv}:=\delta_{pu} \delta_{qv}$. We show that  $h$ 
is a Lie-homomorphism assuming $p\neq q$ and $r\neq s$ in the following, while the case of
$p=q$ or $r=s$ holds trivially.
Note that  $\tfrac{1}{2}(m_p m_q - m_q m_p)=m_p m_q$. It follows from
Eq.~\eqref{eq:quad-majorana_b}
that $h([-\tfrac{1}{2} (m_p m_q - m_q m_p),-\tfrac{1}{2} (m_r m_s {-} m_s m_r)])
=[(e_{pq} - e_{qp}),(e_{rs} - e_{sr})]=
[h(-\tfrac{1}{2}(m_p m_q - m_q m_p)),h(-\tfrac{1}{2}(m_r m_s - m_s m_r))]$.
\end{proof}

\subsection{Examples and Explicit Realizations}

We start by showing that the full system algebra $\so(2d)$
of quasifree fermions can be generated using
only three quadratic operators, namely
$w_1=L(v_1)$, $w_2=L(v_2)$, and $w_3=L(v_3) $ from Eq.~\eqref{v_ops}
where $v_1 = \sum_{p=1}^{d-1} - m_{2p-1} m_{2p+2} + m_{2p} m_{2p+1}$,
$v_2 = m_1 m_2$, and  $v_3 = m_2 m_3$. The Jordan-Wigner transformation maps
these generators respectively to the Heisenberg-\XX term
\begin{equation}\label{xx}
iH_{\XX}=-\frac{i}{2}\sum_{p=1}^{d-1}\left( \mathrm{X}_{p}\mathrm{X}_{p+1} + 
\mathrm{Y}_{p}\mathrm{Y}_{p+1}\right), 
\end{equation}
$-\frac{i}{2}\mathrm{Z}_1$, and $-\frac{i}{2}\mathrm{X}_1\mathrm{X}_2$, where
operators as (e.g.) $\mathrm{Z}_1$ are defined
as $\mathrm{Z}\otimes \mathrm{I} \otimes \cdots \otimes \mathrm{I}$.

\begin{lemma}\label{quad_ex_lem}
Consider a fermionic quantum system with $d\geq 2$ modes.
The system algebras $\fk_1$ and $\fk_2$ generated by the set of Lie generators
$\{w_1,w_2\}$ and $\{w_1,w_2,w_3\}$ contain
the elements $L(a_p)$ with $a_p:=m_{2p-1} m_{2p}$ for all $p\in \{1,\ldots,d\}$
as well as $L(b_p)$ with
$b_p:=- m_{2p-1} m_{2p+2} + m_{2p} m_{2p+1}$
and $L(c_p)$ with
$c_p:= m_{2p-1} m_{2p+1} + m_{2p} m_{2p+2}$
for all $p\in \{1,\ldots,d{-}1\}$.
\end{lemma}

Note that the elements $L(a_p)$, $L(b_p)$, and $L(c_p)$ are mapped by 
the Jordan-Wigner transformation to the spin operators
$-i\mathrm{Z}_p /2$, 
$-i(\mathrm{X}_{p}\mathrm{X}_{p+1} + \mathrm{Y}_{p}\mathrm{Y}_{p+1})/2$, and 
$-i(\mathrm{X}_{p}\mathrm{Y}_{p+1} - \mathrm{Y}_{p}\mathrm{X}_{p+1})/2$, respectively.

\begin{proof}
We compute the commutators
$w_4:= -L(c_1)=[w_2,w_1]$, 
$w_5:= L(b_1)=[w_4,w_2]$, and 
$w_6:= L(a_2)=[w_5,w_4]-w_2$. 
We can now reduce the problem from $d$ to $d-1$ by subtracting $w_5$ from $w_1$. 
The cases of $d\in\{2,3,4\}$ can be verified directly and the proof is
completed by induction.
\end{proof}

This proof also yields an explicit realization for the algebra $\so(2d)$
while providing a more direct line of reasoning as compared to our proof
of Theorem~32 in \cite{ZS11}.
\begin{theorem}\label{quad_ex_so}
Consider a fermionic quantum system with $d\geq 2$ modes.
The system Lie algebra $\fk_2$ generated by $\{w_1,w_2,w_3\}$
is given by $\so(2d)$.
\end{theorem}

\begin{proof}
The cases of $d\in\{2,3,4\}$ can be verified directly.
We build on Lemma~\ref{quad_ex_lem}
and remark that $\fk_2\subseteq\so(2d)$ as it is generated only by 
quadratic operators 
(see Proposition~\ref{quasifree}).
We compute in the Jordan-Wigner picture
$w_{7}:=-i(\mathrm{Y}_1\mathrm{Y}_2-\mathrm{Y}_2\mathrm{Y}_3)/2 =[w_3,[w_3,w_1]]$,
and
$w_{8}:=-i\mathrm{X}_2\mathrm{X}_3/2=L(b_2)-(w_5-w_3-w_{7})$.
This shows by induction that $\so(2d) \supseteq \fk_2 \supsetneq \uu(1)+\so(2d-2)$.
As $\uu(1)+\so(2d-2)$ is a maximal subalgebra of $\so(2d)$ (see p.~219 of \cite{BS49}
or Sec.~8.4 of \cite{GG78}), 
one obtains that $\fk_2 = \so(2d)$.
Alternatively, one can explicitly show that $\fk_2$ consists of all
quadratic Majorana operators, which combined with Proposition~\ref{quasifree}
would also complete the proof.
\end{proof}

Note that the generators $w_1$, $w_2$, and $w_3$ can be described using
the Hamiltonian of Eq.~\eqref{Hqfree} while keeping the control functions
given by the matrix entries
$A_{pq}$ and $B_{pq}$ in the real range, see
Appendix~\ref{ParaQuadratic} for details.
This also provides a simplified approach to
Theorem~32 in \cite{ZS11}, where only the real case was considered:

\begin{corollary}[see Theorem~32 in \cite{ZS11}]
Consider a control system given by the Hamiltonian components
of Eq.~\eqref{Hqfree}. The control functions
are specified by the matrix entries $A_{pq}$ and $B_{pq}$ 
which are assumed to be real.
The resulting system algebra is $\so(2d)$.
\end{corollary} 

The relations between quasifree fermions and spin systems will be analyzed
in Appendix~\ref{appl_quasi}. --- Next we treat the case of the 
algebra $\uu(d)$. 

\begin{theorem}[see Lem.~36 in \cite{ZS11}]\label{quad_ex_u}
Consider a fermionic quantum system with $d\geq 2$ modes.
The system Lie algebra $\fk_1$ generated by $\{w_1,w_2\}$
is given by $\uu(d)$.
\end{theorem}

Here we just sketch ideas for the proof of Theorem~\ref{quad_ex_u}
while leaving the full details to Appendix~\ref{proof_u_d}.
Our methods exploit the detailed structure of the appearing
Majorana operators while being more explicit than in \cite{ZS11}
and avoiding obstacles of the spin picture. 
Building on the notation of Lemma~\ref{quad_ex_lem}, we show that the elements
$L(a_p)$ with $1\leq p \leq d$ together with
the elements $L(b_p^{(i)})$ with
$b_p^{(i)}:=-m_{2p-1} m_{2p+2i} + m_{2p} m_{2p+2i-1}$ and
$L(c_p^{(i)})$ with $c_p^{(i)}:=m_{2p-1} m_{2p+2i-1} + m_{2p} m_{2p+2i}$ where
$p,i\geq 1$ and $p+i\leq d$ form a basis of $\fk_1$. One obtains that $\dim(\fk_1)=d+(d-1)d=d^2$.
Furthermore, the elements $L(a_p)$ form
a maximal abelian subalgebra and the rank of $\fk_1$ is equal to $d$ \footnote{The rank of a Lie algebra 
is defined as the dimension of its maximal abelian subalgebras.}.
We limit the possible cases further by showing
that $\fk_1$ is a direct sum of a simple and a one-dimensional 
Lie algebra. A complete enumeration of all possible cases completes the proof.

\begin{remark}
A spin chain equivalent to the fermionic system in Theorem~\ref{quad_ex_u} was also considered in \cite{BP09}, 
where it was shown how to swap pairs of fermions using the given Hamiltonians. 
As a consequence of Theorem~\ref{quad_ex_u}, the Lie algebra 
in the spin chain of \cite{BP09} can be identified as $\uu(d)$.
Clearly, its size grows only linearly with the number of modes~$d$. However, the addition of 
controlled-Z gates, as discussed in \cite{BP09}, already allows for scalable quantum computation.
\end{remark}

\subsection{Tensor-Square Criterion}

Consider a control system of quasifree fermions which is
represented by matrices $T_\nu$ in the form of Eq.~\eqref{stand_Maj}.
For more than two modes (i.e.\ $d\geq 3$), we 
can efficiently decide if the system algebra is equal to $\so(2d)$.
Recall that the alternating square $\Alt^2(\phi)$ and the symmetric square 
$\Sym^2(\phi)$ of a representation $\phi$ are defined
as restrictions to the alternating and symmetric subspace 
of the tensor square $\phi^{\otimes 2}=\phi \otimes \unity_{\dim(\phi)}
+ \unity_{\dim(\phi)} \otimes \phi$.

\begin{theorem}
Assume that $\fk$ is a subalgebra of $\so(2d)$ with $d\geq 3$ and denote
by $\Phi$ the standard representation of $\so(2d)$.
Then, the following are equivalent:\\
(1) $\fk = \so(2d)$.\\
(2) The restriction of $\Alt^2 \Phi$ to the subalgebra
$\fk$ is irreducible and the restriction
of  $\Sym^2 \Phi$ to $\fk$ splits into two irreducible components.
Each irreducible component occurs only once.\\
(3) The commutant of all complex matrices commuting with
the tensor square $(\Phi|_{\fk})^{\otimes 2}$ of $\fk$ has dimension three.
\end{theorem}
\begin{proof}
Assuming (1), condition (2) follows from the formulas for
the alternating and symmetric square of $\so(2d)$ with $d\geq 3$
given in its standard representation $\phi_{(1,0,\ldots,0)}$
[where $(1,0,\ldots,0)$ denotes the corresponding highest weight]:
The alternating square is given as
$\Alt^2\phi_{(1,0,0)}=\phi_{(0,1,1)}$ for $\so(6)$ 
and
$\Alt^2\phi_{(1,0,0,\ldots,0)}=\phi_{(0,1,0,\ldots,0)}$ for $\so(2d)$ in the case of $d>3$
(cf.\ Table~6 in \cite{Dynkin57} or Table~X in \cite{ZS11}).
The symmetric square
$\Sym^2\phi_{(1,0,\ldots,0)}=\phi_{(2,0,\ldots,0)} \oplus \phi_{(0,0,\ldots,0)}$
for $\so(2d)$ and $d \geq 3$ can be computed using Example~19.21 of Ref.~\cite{FH91}.
We verify the dimension of the commutant and show that (3) is a consequence of (2) 
by applying Proposition~\ref{reprtheory} which says that the dimension of the 
commutant of a representation $\phi$ is given by $\sum_i m_i^2$ where the $m_i$ are 
the multiplicities of the irreducible components of $\phi$.
For the rest of the proof we assume that condition (3) holds. We remark 
that the representation $\Phi|_{\fk}$ is irreducible as otherwise 
the dimension of the commutant would be larger than three. 
Thus, we obtain that $\fk$ is semisimple.
The dimension of the 
commutant allows only two possibilities: one of the restrictions 
$(\Alt^2 \Phi)|_{\fk}$ or $(\Sym^2 \Phi)|_{\fk}$ to the subalgebra $\fk$ has to be irreducible.
We emphasize that $\fk$ is given in an orthogonal representation
(i.e.\ a representation of real type) of even dimension, as
$\fk$ is given in an irreducible representation obtained by restricting the standard 
representation of $\so(2d)$.
Therefore, we can use the list of all irreducible representations
which are orthogonal or symplectic (i.e.\ of quaternionic type) 
and whose alternating or symmetric square is irreducible
(Theorem~4.5 as well as Tables~7a and 7b of Ref.~\cite{Dynkin57}): 
(a) for $\su(2)$ the alternating square of the symplectic representation $\phi=(1)$ of dimension two, 
(b) for $\so(3)\equiv\su(2)$ the alternating square of the orthogonal representation 
$\phi=(2)$ of dimension three,
(c) for $\so(2\ell+1)$ with $\ell>1$ the alternating square of the orthogonal representation 
$\phi=(1,0,\ldots,0)$ of dimension $2\ell+1$,
(d) for $\so(2\ell)$ with $\ell\geq 3$ the alternating square of the orthogonal representation 
$\phi=(1,0,\ldots,0)$ of dimension $2\ell$, and
(e) for $\spp(\ell)$ with $\ell\geq 1$
the symmetric square of the symplectic representation $\phi=(1,0,\ldots,0)$ of dimension
$2\ell$. Only possibility (d) fulfills all conditions which proves (1).
\end{proof}

Describing the matrices in the tensor square more explicitly 
along the lines of Ref.~\cite{ZS11},
we present a necessary and sufficient condition for
full controllability in systems of quasifree fermions.

\begin{corollary}\label{cor_tensor_three}
Consider a set of matrices $\{T_\nu\,|\,\nu\in\{0; 1,\ldots, m\}\}$
as given by Eq.~\eqref{stand_Maj} generating the system algebra 
$\fk \subseteq \so(2d)$ with $d\geq 3$. We obtain $\fk=\so(2d)$
if and only if the joint commutant of $\{ T_\nu\otimes \unity_{2d} + \unity_{2d} \otimes 
T_\nu \,|\,\nu\in\{0; 1,\ldots, m\}\}$
has dimension three. \qed
\end{corollary}

 Along the lines of Eq.~\eqref{stand_Maj},
one can apply Corollary~\ref{cor_tensor_three} to the 
matrices $T$
corresponding to the generators of $\so(2d)$ of Theorem~\ref{quad_ex_so}.
For $d\geq 3$ one can verify that the commutant of the tensor square
has dimension three. But for $d=2$ one computes a dimension of four
as $\so(4)=\su(2)+\su(2)$ is not simple.

For illustration, note that two elements in the commutant are trivial, to wit the identity and 
the generator for
the {\sc swap}-operation
between the two tensor copies. The third element does not yet occur in the unitary 
case described in \cite{ZS11}: it is the projector $P_S$ onto the totally anti-symmetric state.
To see this, recall that Ref.~\cite{Obata58} implies that if the Hamiltonians
$\{iH_\nu\,|\,\nu\in\{0;1,\dots, m\}\}$ generate a system algebra 
of orthogonal type, then there is an operator $S \in SL(N)$ satisfying
\begin{equation}
S H^t_\nu + H_\nu S = 0
\end{equation}
\emph{jointly} for all $\nu\in\{0;1,\ldots, m\}$ as in \cite{ZS11}. Using Kronecker products and writing
$\ket S:= \vec(S)$ \cite{HJ1}, one sees that $\ket{S}$ is in the
intersection of all the kernels of the tensor squares, so
\begin{eqnarray}
(H_\nu\otimes\unity+\unity\otimes H_\nu)\ket{S} &=& \ket 0 \nonumber \\  
\Leftrightarrow (H_\nu\otimes\unity+\unity\otimes H_\nu)\ketbra{S}{S}&=& 0_N  \nonumber \\ 
\Leftrightarrow \ketbra{S}{S}(H_\nu\otimes\unity+\unity\otimes H_\nu)&=& 0_N 
\end{eqnarray}
and thus $P_S:=\ketbra{S}{S} \in (H_\nu\otimes\unity+\unity\otimes H_\nu)'$ holds jointly for 
all $\nu\in\{0;1,\ldots, m\}$; $0_N$ denotes the zero matrix of degree $N$.

\section{Pure-State Controllability for Quasifree Systems\label{pure-quasifree}}

In this section, we present a straightforward  criterion for {\em pure-state controllability} of quasifree fermionic
systems with $d$ modes. A fermionic state is called quasifree if Majorana operators of odd degree
map it to zero and even-degree ones map it to states which factorize into the Wick expansion form (see below).
We obtain that quadratic Hamiltonians act transitively on pure quasifree states, i.e.,
every pure quasifree state can be transformed into any other pure quasifree state using only 
quadratic Hamiltonians (see Theorem~\ref{pure_quasi_uu}). 

In particular, an algebra isomorphic to $\uu(d)$
is left invariant by quadratic Hamiltonians and pure quasifree states can be related to
an homogeneous space of type 
$\SU(2d)/\U(d)$. At first glance, this might suggest that for full pure-state controllability the system algebra 
has to be isomorphic to $\so(2d)$. However, the central result of this section shows
that this is in general not necessary:
a quasifree fermionic system (with $d>4$ or $d=3$) is fully pure-state controllable 
iff its system algebra is isomorphic to $\so(2d)$ or $\so(2d-1)$, see Theorem~\ref{controllability_quasifree_pure}.

\subsection{Quasifree States\label{subsec:quasifree_state}}

A fermionic state $\rho$ on $\mathcal{F}(\mathbb{C}^d)$ is called
{\it quasifree} or {\it Gaussian} if it vanishes on odd monomials of  Majorana 
operators and factorizes 
on even monomials into the {\em Wick expansion form}
\begin{align*}
\tr(\rho m_{k_{1}} \dots m_{k_{2d}}) = 
\sum\limits_{\pi}{\mathrm{sgn}}(\pi) 
\prod\limits_{p=1}^{d}
\tr( \rho m_{k_{\pi(2p-1)}}m_{k_{\pi(2p)}}). 
\end{align*}
Here the sum runs over all pairings of 
$[1,\ldots , 2d]$, i.e., over all
permutations $\pi$ of $[1,\ldots , 2d]$
 satisfying $\pi(2q-1)< \pi(2q)$ and $\pi(2q-1) < \pi(2q+1)$ for all $q$.
The \emph{covariance matrix} of $\rho$ is defined as the $2d \times 2d$ 
skew-symmetric matrix with real entries 
\begin{equation} \label{eq:cov-mat}
G_{pq}^{\rho}=i [\rm{Tr}(\rho m_p m_q) - \delta_{pq}].
\end{equation}
Due to the Wick expansion property, 
a quasifree state is uniquely characterized by its covariance
matrix. (General references for this section include \cite{Araki70,kraus09,Bach94,BM62, dMCT13}.)
The following proposition 
 resumes a known result on these covariance matrices 
(see, e.g., Lemma 2.1 and Theorem 2.3 in \cite{Bach94}),
which will be useful in the later development:

\begin{proposition} \label{prop:cov-sing-val}
The singular values of the covariance matrix of a $d-$mode fermionic state
must lie between $0$ and $1$. Conversely, for 
any $2d \times 2d$ skew-symmetric 
matrix $G^\rho$ with singular values between $0$ and $1$ there exist
a quasifree state that has $G^\rho$ as a covariance matrix.
\end{proposition}

\subsection{Orbits and Stabilizers of Quasifree States under the Action of Quadratic
Hamiltonians}\label{SubSec:PureOrbits1}

The action of the time-evolution unitaries generated by
quadratic Hamiltonians on quasifree states can be described by the following
proposition (see Lemma 2.6 in \cite{Bach94}):

\begin{proposition} \label{prop:qf-dyn}
Consider a  quasifree state $\rho_{a}$ corresponding to the  (skew-symmetric) covariance matrix $G^a$. 
The quadratic Hamiltonian
\begin{equation*}
H=  i \, \sum_{p,q=1}^{2d} T_{pq}\, (-\tfrac{1}{2}  m_p m_q)
\end{equation*}
is defined using
the skew-symmetric matrix $T$ and generates the time-evolution of $\rho_{a}$.
The time-evolved state  (at unit time),
$\rho_{b}= e^{-iH} \rho_{a} e^{iH}$ is again a quasifree state
with a (skew-symmetric) covariance matrix
$G^b= O_TG^aO^t_T$, where $O_T:=e^{-iT} \in \SO(2d)$.
\end{proposition}

Any skew-symmetric matrix $G$ can be brought into its
canonical form 
\begin{equation*}
O_G G O^t_G  =
\begin{pmatrix}
0 & \nu_{1} &  & & && \\
-\nu_{1} & 0 & &  &&\\
& & 0 & \nu_2&  &&\\
& & -\nu_2  & 0 & &&\\
&&&& \ddots && \\
&&&&& 0 & \nu_N  \\
&&&&&-\nu_N & 0 
\end{pmatrix}
\end{equation*}
using a (not necessarily unique) 
element $O_G \in \SO(2d)$
where $\{ \nu_i \}_{i=1}^d$ denotes the singular values of $G$. This means
that a quasifree state can be
reached from another one by the action of quadratic Hamiltonians if  
their covariance matrices share the same singular
values (including multiplicities).
Let us now recall another result related to
the singular values of the covariance matrices of pure quasifree states (Theorem~6.2 in
\cite{Bach94}, and Lemma~1 in \cite{dMCT13}):

\begin{proposition} \label{prop:qf-stab}
A quasifree state $\rho$ is pure iff the following (equivalent) conditions
hold for its covariance matrix $G^{\rho}$:\\
(a) The rows (and columns) of $G^{\rho}$ are real 
unit vectors which are pairwise orthogonal to each other.\\
(b) The singular values of $G^{\rho}$ are all $1$.
\end{proposition}
Applying this result together with Proposition~\ref{prop:qf-dyn},
we obtain the following theorem:

\begin{theorem}\label{pure_quasi_uu}
The set of quadratic Hamiltonians act transitively on pure quasifree states,
and the corresponding stabilizer algebras are isomorphic to $\uu(d)$. 
\end{theorem}

\begin{proof}
We have already shown that the singular values of the covariance matrices 
(with multiplicities) form a separating set of invariants for the orbits
generated by quadratic Hamiltonians over the set of quasifree states.
This means, according to Proposition~\ref{prop:qf-stab},
that the pure quasifree states form a single orbit.

As the set of quadratic Hamiltonians generate a transitive action over  
the pure quasifree states, the corresponding stabilizer subalgebras
are isomorphic to each other. 
Consider a quadratic Hamiltonian $H$ with the coefficient matrices $A$ and $B$ as given
in Eq.~\eqref{Hqfree}
and the Fock state $\rho_\Omega$, which is the projection onto the Fock vacuum
vector $\Omega$. 
The state $\rho_\Omega$ is left invariant under the time 
evolution generated by $H$ ($\rho_\Omega=e^{-iHt} \rho_\Omega e^{iHt}$) iff 
$\Omega$ is an eigenvector of $H$. 
We obtain that $H \Omega=$
\begin{align*}
 &\left[\sum_{p,q=1}^d A_{pq}(f^{\dagger}_ pf^{\phantom\dagger}_q{-}\delta_{pq}\tfrac{\unity}{2}) 
+ \frac{1}{2}B_{pq} f_{p}^{\dagger}f_{q}^{\dagger}-
\frac{1}{2}B_{pq}^{*} f_p f_{q}\right] \Omega=\\ & 
-\sum_{p=1}^d \frac{1}{2} A_{pp} \Omega + \sum_{ p < q} B_{pq} f_{p}^{\dagger}f_{q}^{\dagger} \Omega.
\end{align*}
By noting that $\Omega$ and $ f_{p}^{\dagger}f_{q}^{\dagger} \Omega$ (with $p < q $) are linearly 
independent vectors, we can conclude that a quadratic  Hamiltonian $H$ leaves the Fock 
vacuum invariant iff
$H=\sum_{p,q=1}^d A_{pq}(f^{\dagger}_ pf^{\phantom\dagger}_q{-}\delta_{pq}\tfrac{\unity}{2})$. 
In Theorem~\ref{quasifree_particle} of Sec.~\ref{sec:PC} we will show that these
operators form a Lie algebra isomorphic to $\uu(d)$. 
\end{proof}

\begin{corollary}\label{cor:pure_quasi_uu}
Theorem~\ref{pure_quasi_uu} identifies the space of pure
quasifree states with the quotient space $\SO(2d)/\U(d)$.
\end{corollary}

\subsection{Conditions for Quasifree Pure-State Controllability\label{sec:qf_pure_cont}} 

According to Theorem~\ref{pure_quasi_uu},
a set of quasifree control Hamiltonians 
$\{ H_1, \ldots, H_m \}$ allows for quasifree 
pure-state controllability, if the corresponding
Hamiltonians generate 
the full quasifree system algebra, i.e.\ if
$ \langle iH_1, \ldots, iH_m \rangle_{\text{Lie}} \cong \so(2d)$.
It is natural to ask whether this condition is also a necessary. 
Remarkably, it turns out that this is not the case, which is
shown by the following proposition:

\begin{lemma} \label{prop:so(2d-1)}
Consider a quasifree fermionic system with $d>1$ modes.
Let $K$ be the subgroup of $\SO(2d)$ which 
is isomorphic to $\SO(2d-1)$ and 
stabilizes the first coordinate; the 
corresponding
Lie algebra is denoted by $\fk$. Then\\
(a) via its adjoint action the group $K$ acts transitively 
on the set of all skew-symmetric covariance matrices 
of pure quasifree states
(whose singular values are all  $1$);\\
(b) the quasifree system is  pure-state controllable 
if its system algebra is conjugate under $\SO(2d)$ to  $\fk$.
\end{lemma}

\begin{proof}
We prove the statement (a)
by showing that all pure quasifree states can be transformed 
under $K$-conjugation
to the same pure state.
We employ an induction on $d$.
The base case $d=2$
can be directly verified.
It follows from Proposition~\ref{prop:qf-stab}(b) that the
skew-symmetric covariance matrix of a pure quasifree state
can be written as
\begin{equation*}
G^{\rho}=
\begin{pmatrix}
0 &  v_1^t\\
 - v_1 & A_1
\end{pmatrix} ,
\end{equation*}
where $v_1$ denotes a normalized $(2d{-}1)$-dimensional vector and $A_1$
denotes a $(2d{-}1) \times (2d{-}1)$-dimensional skew-symmetric matrix.
We consider the action of a general orthogonal
transformation $1 \oplus O_1$ with $O_1 \in \SO(d-1)$:
\begin{equation*}
\begin{pmatrix}
1 &   \\
  & O_1
\end{pmatrix}
\begin{pmatrix}
0 &  v_1^t\\
 - v_1 & A_1
\end{pmatrix}
\begin{pmatrix}
1 &   \\
  & O_1^t
\end{pmatrix} =
\begin{pmatrix}
0 &  v_1^t O_1^t\\
 - O_1 v_1 & O_1A_1O_1^t 
\end{pmatrix}.
\end{equation*}
Since any $(2d{-}1)$-dimensional vector $v_1$ with unit length
can be transformed by an orthogonal transformation to
$(1,0,0,\ldots, 0)$, we can choose $O_1$ such that $v_1^t O_1^t = (1,0,0,\ldots,0)$.
We have $(O_1A_1O_1^t)_{11} =0$
as  the transformed matrix is skew-symmetric.
Again by Proposition~\ref{prop:qf-stab}(b) we obtain the transformed matrix as
\begin{equation*}
\begin{pmatrix}
0   & 1  \\
-1  & 0 \\
    &  & 0    & v_2^t\\
    &  & -v_2 & A_2
\end{pmatrix},
\end{equation*}
where $v_2$ is a $2d-3$ dimensional unit real vector and $A_2$
is a $(2d-3) \times (2d-3)$ skew-symmetric matrix. 
Now the proof of (a) follows using the induction hypothesis.
The statement (b) is a consequence of (a).
\end{proof}

We relate Lemma~\ref{prop:so(2d-1)} to what is known about 
transitive actions on the coset space $\SO(2d)/\U(d)$.
Only Lie groups isomorphic to $\SO(2d-1)$ and $\SO(2d)$ can act transitively
(i.e.~in a pure-state controllable manner) on 
the homogeneous space $\SO(2d)/\U(d)$
assuming $d\geq 3$. The case $d\geq 4$  is discussed in \cite{Kerr96}.
For $d=3$ we have $\SO(6)\cong \SU(4)$ and
$\SU(4)/\U(3)=\mathrm{CP}^3$ (where $\mathrm{CP}^3$ denotes the
complex projective space in four dimensions), and it is known that
only subgroups of $\SU(4)$ isomorphic to
$\SU(4)$ or $\SP(2)\cong \SO(5)$ can act transitively
on $\mathrm{CP}^3$ (see p.~168 of \cite{Oni66} or p.~68 of \cite{DiHeGAMM08}; 
refer also to \cite{KDH12}). 

In most cases the $\so(k-1)$-subalgebras of $\so(k)$ are conjugate to each other.
More precisely, Lemma~7 of \cite{MS43} states
that for $3 \leq k\notin \{4,8\}$ all subalgebras of $\so(k)$ whose dimension
is equal to $(k-1)(k-2)/2$ are conjugate to each other under the  action of the group $\SO(k)$. 
In particular, it follows in these cases that all subalgebras of $\so(k)$ with dimension
$(k-1)(k-2)/2$ are isomorphic to $\so(k-1)$. 
Interestingly, the last statement holds also for $k\in\{4,8\}$
(see Lemma~3 of \cite{MS43}); however 
not all of these subalgebras of $\so(k)$ are conjugate.
We obtain the following theorem providing a necessary and sufficient
condition for full quasifree pure-state controllability in the case of  $d>4$ or $d=3$ modes:

\begin{theorem}\label{controllability_quasifree_pure}
A quasifree fermionic 
system with $d> 4$ or $d=3$ modes is fully pure-state controllable 
iff 
its system algebra is isomorphic to $\so(2d)$ or
$\so(2d-1)$.
\end{theorem}

\begin{proof}
\noindent ``$\Rightarrow$'': Note that Theorem~\ref{pure_quasi_uu} identifies the space of 
pure quasifree states with 
the homogeneous space $\SO(2d)/\U(d)$. Assuming $d\geq 3$,
we summarized above that
a group acting transitively on this homogeneous space is isomorphic either to
$\SO(2d)$ or $\SO(2d-1)$.
Thus only the full quasifree system algebra $\so(2d)$ 
or a system algebra isomorphic to $\so(2d-1)$ can generate a transitive action on the 
space of pure quasifree states.\\
\noindent ``$\Leftarrow$'':  As discussed, all $\so(2d-1)$-subalgebras are conjugate 
to each other for $d>4$ and $d=3$.  
Lemma~\ref{prop:so(2d-1)}(b) then implies that any set of Hamiltonians 
generating a system algebra isomorphic to $\so(2d-1)$
will allow for full quasifree pure-state controllability.
\end{proof}

Note that the cases $d=2$ and $d=4$
are well-known pathological exceptions. 
The algebra $\so(4)$ 
breaks up into a direct sum of two 
$\so(3)$-algebras which hence cannot be conjugate
to each other.
For $d=4$, there are three classes of non-conjugate
subalgebras of type $\so(7)$ in $\so(8)$ where two classes are given
by irreducible embeddings and the third one is conjugate to the reducible 
standard embedding fixing the first coordinate \footnote{For 
details refer to the discussions
on the pages 57--58 of \cite{Oni94},
on the pages 234--235 of \cite{Min06}, or on the pages 418--419 of \cite{deGraaf11}.
In addition, this information can also be inferred from the tables on page~260
of \cite{McKay81}.}.

On a more general level, Theorem~\ref{controllability_quasifree_pure}
can be seen as a fermionic variant of the pure-state controllability criterion for spin systems
 \cite{AA03,SchiSOLea02b,SchiSOLea02}. We note here
that the result for spin systems has been recently generalized from the
transitivity over a set of one-dimensional projections (i.e.\ pure states) to the transitivity over 
a set of projections of arbitrary fixed rank (i.e., over Grassmannian  spaces) \cite{KDH12}. 
We will use exactly this generalization in Section~\ref{subsec:part-numb-pure-quasifree} in order 
to find a necessary and sufficient pure-state controllability condition for 
particle-conserving quasifree systems.

\section{Translation-Invariant Systems\label{sec:TI}}

We study system algebras generated by translation-invariant
Hamiltonians of the type which arises approximately in experimental settings of, e.g., optical lattices. 
As the naturally occurring interactions  
are usually short-ranged, we pay particular attention to the case of 
Hamiltonians with restricted interaction length.
For example, consider a $d$-site fermionic chain with Hamiltonians which are translation-invariant 
and are composed of nearest-neighbor (plus on-site) terms. All elements in its dynamic algebra
can be written as linear combinations of six types of terms:
the chemical potential 
\begin{align}
\label{Eq_lin_1}
h_0&:= \sum_{n=1}^{d} \left(f^{\dagger}_n f^{\phantom\dagger}_n -\tfrac{1}{2}\unity\right),
\intertext{the real and complex hopping Hamiltonians}
\label{Eq_lin_2}
h_{\mathrm{rh}}&:= \sum_{n=1}^{d} (f^{\dagger}_n f^{\phantom\dagger}_{n+1} 
+ f^{\dagger}_{n+1} f^{\phantom\dagger}_{n})\; \text{ and }\\ 
\label{Eq_lin_3}
h_{\mathrm{ch}}&:= \sum_{n=1}^{d} i(f^{\dagger}_n f^{\phantom\dagger}_{n+1} 
- f^{\dagger}_{n+1} f^{\phantom\dagger}_{n}),
\intertext{the real and complex pairing terms}
\label{Eq_lin_4}
h_{\mathrm{rp}}&:= \sum_{n=1}^{d} (f^{\dagger}_n f^{\dagger}_{n+1} 
+ f^{\phantom\dagger}_{n+1} f^{\phantom\dagger}_{n})\; \text{ and }\\
\label{Eq_lin_5}
h_{\mathrm{cp}}&:=\sum_{n=1}^{d} i(f^{\dagger}_n f^{\dagger}_{n+1} 
- f^{\phantom\dagger}_{n+1} f^{\phantom\dagger}_{n}),
\intertext{as well as a local density-density-type interaction}
\label{Eq_lin_6}
h_{\mathrm{int}}&:=  
\sum_{n=1}^{d}\left(f^{\dagger}_n f^{\phantom\dagger}_n f^{\dagger}_{n+1} f^{\phantom\dagger}_{n+1}  
- \tfrac{1}{4}\unity\right).
\end{align}
The corresponding dynamic system algebras (given in Table~\ref{tab:trans-inv}) were computed with
the help of the computer algebra system {\sc magma} \cite{MAGMA} for up to six modes
while distinguishing nearest-neighbor interactions from arbitrary translation-invariant ones.

\begin{table}[tb]
\caption{\label{tab:trans-inv} System algebras of translation-invariant fermionic systems with $d$ modes for
(a) nearest-neighbor interactions only and (b) arbitrary translation-invariant interactions}
\begin{tabular}{l@{\hspace{5mm}}l@{\hspace{15mm}}r}
\hline\hline\\[-2mm]
Case & $d$
& System algebra\\[1mm] 
\hline \\[-2mm]
(a) & 1 & --\\[1mm]
& 2 & $\sum_{i=1}^2 \uu(1)$\\[1mm]
& 3 & $\sum_{i=1}^2\su(2)+\sum_{i=1}^3\uu(1)$
 \\[1mm]
& 4 & $\sum_{i=1}^5 \su(2) + \sum_{i=1}^4 \uu(1)$ \\[1mm]
& 5 & $\sum_{i=1}^2\su(4) + \sum_{i=1}^{8} \su(3) + \sum_{i=1}^3 \uu(1)$ \\[1mm]
& 6 & $\sum_{i=1}^4\su(6) + \sum_{i=1}^8 \su(5) + \sum_{i=1}^3 \uu(1)$ \\[2mm] \hline \\[-2mm]
(b) & 1 & --\\[1mm]
& 2 & $\sum_{i=1}^2 \uu(1)$\\[1mm]
& 3 & $\sum_{i=1}^2\su(2)+\sum_{i=1}^4\uu(1)$\\[1mm]
& 4 & $\sum_{i=1}^8 \su(2) + \sum_{i=1}^6 \uu(1)$  \\[1mm]
& 5 & $\sum_{i=1}^2\su(4) + \sum_{i=1}^8 \su(3) + \sum_{i=1}^8 \uu(1)$ \\[1mm]
& 6 & $\sum_{i=1}^4\su(6) + \sum_{i=1}^8 \su(5) + \sum_{i=1}^{10} \uu(1)$ \\[1mm]
\hline\hline
\end{tabular}
\end{table}

In this context, two sets of natural questions arise: 
(a) How does the dimension of these dynamic system algebras scale with the number of modes? 
(b) How do the system algebras generated by the nearest-neighbor terms differ from the general
translation-invariant ones? 
Can one characterize those elements that are translation-invariant 
yet not generated by nearest-neighbor Hamiltonians? Are there, for example, next-nearest-neighbor 
interactions of this type?
In this section, we will answer these questions partially. We determine the system algebra 
for general translation-invariant fermionic Hamiltonians, and conclude that its dimension scales 
exponentially with the the number of modes.
We also provide translation-invariant fermionic Hamiltonians of bounded interaction length
which cannot be generated by nearest-neighbor ones.

The structure of this section is the following:
As the structure of system algebras for translation-invariant systems has only been studied
sparsely even for simple scenarios of spin models, we start by examining this case first. 
In Secs.~\ref{transl_spin} and
\ref{shortrange}, we determine the system algebras of all translation invariant spin-chain Hamiltonians 
with $L$ qubits. In particular, we simplify and generalize results of \cite{kraus-pra71} concerning 
finite-ranged interactions. Finally, we present the corresponding results for the fermionic case in
Sec.~\ref{sec-transl-fermionic} and \ref{ferm_nn}.

\begin{table*}[bt]
\caption{\label{tab:trans-inv-spin} System algebras $\mathfrak{t}_M(L)$ of translation-invariant 
systems with $1\leq L \leq 6$ spins and interaction lengths of less than $M$, where 
$\fk_{a}:=\su(4)+\sum_{i=1}^2\su(2)$,
$\fk_{b}:=\su(6)+\su(4)$,
$\fk_{c}:=\su(8)+\sum_{i=1}^4\su(6)$, and
$\fk_{d}:=\su(14)+\sum_{i=1}^2\su(11) + \su(10)+\sum_{i=1}^2 \su(9)$. Refer also to 
Theorem~\ref{translation_spins} for the structure of 
$\mathfrak{t}_{L}(L)$.}
\begin{tabular}{c@{\hspace{6mm}}rl@{\hspace{3mm}}l@{\hspace{3mm}}
l@{\hspace{3mm}}l@{\hspace{3mm}}l@{\hspace{3mm}}l}
\hline\hline\\[-2mm]
$M$& $L=$ & 1 & 2 & 3 & 4 & 5 & 6\\[1mm] 
\hline \\[-2mm]
1 &&  $\su(2)$ &  $\su(2)$ & $\su(2)$ & $\su(2)$ &$\su(2)$ & $\su(2)$  \\[1mm]
2 && -- &  $\su(3)+\uu(1)$ &  $\fk_{a}+\uu(1)$ & $\fk_{b} + \sum_{i=1}^2\su(2) + \uu(1)$& $\fk_{c} 
+ \uu(1)$ & $\fk_{d}+ \uu(1)$ 
\\[1mm]
3 && -- & -- &  $\fk_{a}+\sum_{i=1}^2 \uu(1)$ & $\fk_{b} + \sum_{i=1}^2\su(3) + \sum_{i=1}^3 \uu(1)$ & $\fk_{c} 
+ \sum_{i=1}^2 \uu(1)$ & $\fk_{d}+ \sum_{i=1}^3 \uu(1)$ \\[1mm]
4 && -- & -- & -- & $\fk_{b} + \sum_{i=1}^2\su(3) + \sum_{i=1}^3 \uu(1)$ & $\fk_{c} + 
\sum_{i=1}^4 \uu(1)$ & $\fk_{d}+ \sum_{i=1}^4 \uu(1)$ 
\\[1mm]
5 && -- & -- & -- & -- & $\fk_{c} + \sum_{i=1}^4 \uu(1)$ & $\fk_{d}+ \sum_{i=1}^5 \uu(1)$ \\[1mm]
6 && -- & -- & -- & -- & -- & $\fk_{d}+ \sum_{i=1}^5 \uu(1)$ \\[1mm] 
\hline\hline
\end{tabular}
\end{table*}

\subsection{Translation-Invariant Spin Chains\label{transl_spin}}

Consider a chain of $L$ qubits with Hilbert space
$\otimes_{i=1}^{L}\mathbb{C}^2$.
The 
\emph{translation unitary} $U_T$ 
is defined by its action on the canonical basis 
vectors as
\begin{equation}\label{eq:spin-trans-def} 
U_T\, |n_1, n_2, \ldots, n_L \rangle =
|n_L, n_1, \ldots, n_{L-1} \rangle
\end{equation}
where $n_i \in \{0,1 \}$.
We will determine the translation-invariant system algebra 
which is defined as the maximal Lie algebra
of skew-hermitian matrices commuting with the translation unitary
$U_T$.

\begin{lemma} \label{lemma:trans_inv_spinchain}
The translation unitary can be spectrally decomposed as
$
U_T=\sum_{\ell = 0}^{L-1} \exp(2\pi i \ell /L) P_\ell 
$, 
and the rank $\Tm_\ell$ of the spectral projection $P_{\ell}$ 
is given by the Fourier transform
\begin{equation} \label{eq:trans-mult} 
\Tm_\ell:=\frac{1}{L}\sum_{k=0}^{L-1} 2^{\gcd(L,k)}\, \exp(-2 \pi i k\ell/L),
\end{equation}
where $\gcd(L,k)$ denotes the greatest common divisor of
$L$ and $k$.
\end{lemma}
\begin{proof}
The eigenvalues of $U_T$ are limited to 
$\exp(2\pi i \ell/L)$ with $\ell \in \{0,\ldots,L{-}1 \}$ as the order of $U_T$ is
$L$, i.e.\ $U_T^L=\unity$. 
Hence, the corresponding spectral decomposition is given by
$U_T=\sum_{\ell = 0}^{L-1} \exp(2\pi i \ell /L) P_\ell$.
This induces a unitary 
representation $D_T$
of the cyclic group $\mathbb{Z}_L$ which maps
the $k$-th power of the generator $g\in\Z_L$ of degree $L$ to
$D_T(g^k)=U_T^k$.
Note that
$D_T$ 
splits up into 
a direct sum 
$D_T \cong \oplus_{\ell \in \{0,\ldots,L{-}1 \}}   (D_\ell)^{\oplus \dim (P_\ell)}$
containing
$\dim (P_\ell)$ copies of 
the one-dimensional representations satisfying 
$D_\ell (g^k)=\exp(2 \pi i k \ell /L)$. 
Therefore, we determine the rank of a projection
$P_\ell$ by computing the
multiplicity of $D_\ell$ using the character scalar product
\begin{align*}
\Tm_\ell &=\frac{1}{L}\sum_{k=0}^{L-1} \tr[D_T(g^k)]\; 
{\tr[D_{\ell}(g^k)]}^{*} \\ &= \frac{1}{L}
\sum_{k=0}^{L-1} \tr[D_T(g^k)]\; \exp(-2 \pi i k \ell /L).
\end{align*} 
The trace of $D_T(g^k)$ is equal to the number of
basis vectors left invariant since $D_T(g^k)$ is 
a permutation matrix in the canonical basis.
From elementary combinatorial theory we know that 
a bit string $(n_1, n_2, \ldots , n_L)$ is left invariant 
under a cyclic shift by $k$ positions if and only if it is of the form
\begin{gather*}
(n_1, n_2, \ldots, n_{\gcd(L,k)},\ldots, 
n_1, n_2, \ldots, n_{\gcd(L,k)}) \, .
\end{gather*}
It follows that the number of $U_T^k$-invariant basis vectors 
and---hence---the trace of $D_T(g^k)=U_T^k$ is equal to $2^{{\rm{gcd}}(L,k)}$.
Thus, the multiplicities of $D_\ell$ are given accordingly by
$
\Tm_\ell=\frac{1}{L}\sum_{k=0}^{L-1} 2^{\gcd(L,k)}\, \exp(-2 \pi i k\ell/L)
$.
\end{proof}

Note that a Hamiltonian commutes with $U_T$ iff it commutes with all 
spectral projections $P_\ell$ of $U_T$. Combining this fact with Theorem~\ref{thm:double-centralizer}
we obtain a characterization of the system algebra for translation-invariant spin systems:

\begin{theorem}\label{translation_spins}
The
translation-invariant Hamiltonians
acting on a $L$-qubit system generate the system algebra
$\mathfrak{t}(L):=\,  \mathfrak{s}[\oplus_{\ell=0}^{L-1}\, \uu (\Tm_\ell)]  
\cong [\sum_{\ell=0}^{L-1}\, \su (\Tm_\ell)] + [\sum_{i=1}^{L-1}\, \uu (1)]$, where the numbers 
$\Tm_\ell$ are defined in
Eq.~\eqref{eq:trans-mult}.
\end{theorem}

In complete analogy one can show that for a chain consisting of $L$ systems 
with $N$ levels, the 
system algebra is equal to $\mathfrak{s}[\oplus_{\ell=0}^{L-1}\, \uu (\Tm_{N,\ell})]$,
where $\Tm_{N,\ell}$ denotes the Fourier transform of the function $N^{\gcd(L,k)}$.

\subsection{Short-Ranged Spin-Chain Hamiltonians\label{shortrange}}

In many physical scenarios, we may only have direct control over 
translation-invariant Hamiltonians of limited interaction range. 
We will investigate in this section how the limitations on the interaction range constrain
the set of reachable operations.
In particular, we provide upper bounds for the system algebras with finite interaction range.

Let us denote the Lie algebra corresponding to Hamiltonians of 
interaction length less than $M$ by $\mathfrak{t}_{M}(L)$, or $\mathfrak{t}_{M}$ for short. 
In other words, $\mathfrak{t}_{M}(L)$ is the Lie 
subalgebra of $\mathfrak{t}(L)$  
generated by the skew-hermitian operators  
$$
i\sum_{q=0}^{L-1}
U^q_T \left[  \left( \bigotimes_{p=1}^{M} Q_p \right)  \otimes  \unity_2^{\otimes L-M}  \right]
U^{-q}_T
$$
for all combinations of $Q_p \in \{ \unity_2, \mathrm{X},\mathrm{Y},\mathrm{Z} \}$
apart from the case when $Q_1=\unity_2$. In this way, $\mathfrak{t}_1(L)$ corresponds to 
the translation-invariant on-site Hamiltonians, while $\mathfrak{t}_2(L)$ is generated by the 
on-site terms and the nearest-neighbor interactions, and so on. Finally, we have 
$\mathfrak{t}_L(L)=\mathfrak{t}_L$.

We computed all the 
algebras $\mathfrak{t}_M(L)$ for $ 1 \le L \le 6$ and $1 \le M \le L$ using the computer 
algebra system {\sc magma} \cite{MAGMA}. The results, shown in Table~\ref{tab:trans-inv-spin}, 
suggest that for certain restrictions on the interaction length (e.g.,
nearest-neighbor terms), there will be some translation-invariant interactions
that cannot be generated. This is in accordance with the result of Kraus et al. \cite{kraus-pra71}. 
Building partly on their work, we analyze the properties of
the algebras $\mathfrak{t}_M(L)$ for general $M$ and $L$ values, and then compare
our theorems with Table~\ref{tab:trans-inv-spin}.

We first mention a central proposition whose proof can be found in  Appendix~\ref{app:VIIB}:
\begin{proposition} \label{prop:trace-short-ranged}
Let $M < L$ denote a  divisor of $L$.  Given two elements
 $iH_{M} \in \mathfrak{t}_{M}$ and $iH_{M+1} \in \mathfrak{t}_{M+1}$, we obtain that
\begin{gather*}
\tr (U^{qM}_T H_{M} )= 0 \text{ and } 
\tr [(U^{qM}_T-U^{-qM}_T) H_{M+1}]=0
\end{gather*}
hold for any positive integer $q$.
\end{proposition}

Applying Proposition~\ref{prop:trace-short-ranged}, we can present upper bounds for the 
system algebras with restricted interaction length.

\begin{theorem} \label{Thm:SpinShortRange}
Let $M < L$ denote a divisor of the number of spins $L$, and define $R:=L/M$.
We obtain:\\
(a) The algebra $\mathfrak{t}_{M}$ is isomorphic to a 
Lie subalgebra of $[\sum_{\ell=0}^{L-1} \su(r_\ell)] + [\sum_{i=1}^{L-R}\uu(1)]$ and does not
generate $\mathfrak{t}_{L}$.\\
(b) The algebra $\mathfrak{t}_{M{+}1}$  is isomorphic to a 
Lie subalgebra of $[\sum_{\ell=0}^{L-1} \su(r_\ell)] + [\sum_{i=1}^{L-1-\lfloor R/2\rfloor }\uu(1)]$
and does not generate $\mathfrak{t}_{L}$.\\ 
(c) In addition, $\mathfrak{t}_{M} \ne \mathfrak{t}_{M{+}1} $.
\end{theorem}
\begin{proof}
(a) Since $M$ is a divisor of $L$, the equation
\begin{align*}
U_T^{qM}&=\sum_{\ell = 0}^{L-1} \exp(2\pi i qM\ell /L) P_\ell=
\sum_{\ell = 0}^{L-1} \exp(2\pi i q\ell/R) P_\ell\\
& =\sum_{\ell' = 0}^{R-1} \exp(2\pi i q\ell'/R) \left(\sum_{p=0}^{M-1} P_{pR+ \ell'}\right)
\end{align*}  
holds for any integer $q$.
One can invert the equation as
\begin{align*}
\left(\sum_{p=0}^{M-1} P_{pR+ \ell'}\right)=\frac{1}{R}\sum_{q = 0}^{R-1} \exp(-2\pi i q\ell'/R)\, U_T^{qM}.
\end{align*}
If $ih \in \mathfrak{t}_{M}$, we obtain by applying Proposition~\ref{prop:trace-short-ranged} that
\begin{align} \label{eq:trace-alg}
\tr \left( ih\sum_{p=0}^{M-1} P_{pR+ \ell'}\right)=0
\end{align}
holds for $\ell' \in \{0,1, \ldots, R-1\}$. It follows that $\mathfrak{t}_{M}$ is a subalgebra of the 
Lie algebra $\mathfrak{f}$ which consists of all skew-hermitian matrices satisfying the condition 
in Eq.~\eqref{eq:trace-alg}.  
Note that $\mathfrak{f}$  is isomorphic to
$\oplus_{\ell' = 0}^{R-1} (\mathfrak{s}[\oplus_ {p=0}^{M-1}u(r_{pR+ \ell'})])\cong
[\sum_{\ell=0}^{L-1} \su(r_\ell)] + [\sum_{i=1}^{L-R}\uu(1)]$,  and part (a) follows.\\
(b) For elements $ig \in \mathfrak{t}_{M{+}1}$, Proposition~\ref{prop:trace-short-ranged}
and Eq.~\eqref{eq:trace-alg} imply that
\begin{align} \label{eq:trace-alg-2}
\tr \left[ ig \sum_{p=0}^{M-1} \left(P_{pR+ \ell'}- 
P_{pR+ L-\ell'}\right) \right]=0.
\end{align}
The maximal Lie algebra consisting of skew-hermitian matrices which satisfy the
condition in Eq.~\eqref{eq:trace-alg-2} is isomorphic to 
$[\sum_{\ell=0}^{L-1} \su(r_\ell)] + [\sum_{i=1}^{L -1 -\lfloor R/2\rfloor }\uu(1)]$.\\ 
(c) Let 
$$
ih= i\sum_{q'=0}^{L-1}
U^{q'}_T \left[  X \otimes \unity_2^{\otimes M-1} \otimes X \otimes  \unity_2^{\otimes L-M-1}  
\right] U^{-q'}_T.
$$
Obviously,  $ih \in \mathfrak{t}_{M{+}1}$ holds. Using the formula for $F(1,M{+}1)$, we obtain
that $\tr(U^{qM}_T ih)=i2L$  holds for every integer $q$. Hence, 
$ih \notin \mathfrak{t}_{M}$. 
\end{proof}
In particular, this theorem implies that the algebra $\mathfrak{t}(L)=\mathfrak{t}_L(L)$
of all translation-invariant Hamiltonians cannot be generated from the subclass
of nearest-neighbor 
Hamiltonians, cp.\ also~\cite{kraus-pra71}. More precisely, one finds:
\begin{corollary}
If $L$ is even, $\mathfrak{t}_{2}(L)$ is isomorphic to a Lie subalgebra of
$[\sum_{\ell=0}^{L-1} \su(r_\ell)] + [\sum_{i=1}^{L/2}\uu(1)]$.
For odd $L \ge 3$, $\mathfrak{t}_{2}(L)$ is isomorphic to a Lie subalgebra of
$[\sum_{\ell=0}^{L-1} \su(r_\ell)] + [\sum_{i=1}^{(L-3)/2} \uu(1)]$.
\end{corollary}

Let us now compare our upper bounds with the results of Table~\ref{tab:trans-inv-spin}. 
Theorem~\ref{Thm:SpinShortRange} restricts the possibilities for the $M$-local  algebras $\mathfrak{t}_M(L)$ 
only by some central elements $\uu(1)$
when compared to the corresponding  full translation-invariant algebra $\mathfrak{t}(L)$. 
One can indeed identify in Table~\ref{tab:trans-inv-spin} some missing $\uu(1)$-parts
for $L\in\{3,\ldots,6\}$. In general, the dimensions of the $M$-local algebras $\mathfrak{t}_M(L)$
can be even smaller than predicted by the upper bounds of Theorem~\ref{Thm:SpinShortRange}
as can be seen in Table~\ref{tab:trans-inv-spin} for $L=4$.
Theorem~\ref{Thm:SpinShortRange} and 
Table~\ref{tab:trans-inv-spin} suggest that the prime decomposition of the chain 
length $L$ will have strong implications on the dimension of $\mathfrak{t}_M(L)$.

\subsection{Translation-Invariant Fermionic Systems\label{sec-transl-fermionic}}

To determine the system algebra generated by
all translation-invariant Hamiltonians of a fermionic chain, 
we can follow similar lines as in Sec.~\ref{transl_spin}.
Here, however, we additionally have to consider the parity superselection rule.
We define the fermionic translation-invariant system algebra as the maximal 
Lie subalgebra of $\su(2^{d-1})\oplus \su(2^{d-1})$ [see Theorem~\eqref{general}]
which contains only skew-hermitian matrices 
commuting with the fermionic translation unitary
$\TU$, which is defined below such that it commutes with the parity operator 
$P$ (see Eq.~\eqref{P_def}).
The standard orthonormal basis in the Fock space for a chain of $d$ fermionic modes 
is given by 
\begin{equation} \label{eq:fermi-basis}
| n_1, n_2, \ldots n_d \rangle := 
(f^{\dagger}_{1})^{n_1} (f^{\dagger}_{2})^{n_2} \cdots
(f^{\dagger}_{d})^{n_d}\; |0 \rangle 
\end{equation}
with $n_i \in \{0,1\}$.
Note that for the purpose of unambiguously defining this basis, we order
the operators $(f^{\dagger}_i)^{n_i}$ in Eq.~\eqref{eq:fermi-basis}
with respect to their site index $i$.
The fermionic translation unitary $\TU$ is defined by its action 
\begin{align}
&\TU\, | n_1, n_2, \ldots n_d \rangle = 
\TU\, (f^{\dagger}_{1})^{n_1} (f^{\dagger}_{2})^{n_2} \cdots
(f^{\dagger}_{d})^{n_d}\; |0 \rangle \nonumber \\
&=(f^{\dagger}_{2})^{n_1}  \cdots
(f^{\dagger}_{d})^{n_{d-1}} (f^{\dagger}_{1})^{n_{d}}\; |0 \rangle \nonumber \\
&=(-1)^{n_d(n_1+n_2 + \ldots + n_{d-1})} 
(f^{\dagger}_{1})^{n_d} (f^{\dagger}_{2})^{n_1} \cdots
(f^{\dagger}_{d})^{n_{d-1}}\; |0 \rangle 
\nonumber \\
&= (-1)^{n_d(n_1+n_2+ \ldots +n_{d-1})}\; | n_d, n_1, \ldots, 
n_{d-1} \rangle \label{Trans-ferm} 
\end{align}
on the standard basis.
The adjoint action of 
$\TU$ on the creation operators $f^\dagger_\ell$ is then given by
\begin{align*}
\TU\, f^\dagger_\ell\, \TU^{\dagger}=f^\dagger_{(\ell+1 \mod d)}.
\end{align*}
The superselection rule for fermions
splits the spectral decomposition of the translation unitary
into two blocks corresponding to the positive and negative parity subspace.
The translation unitary $\TU$ commutes with the parity operator $P$, and
hence $\TU=\TU_{+}+ \TU_{-}$ is block-diagonal in the eigenbasis of $P$
where $\TU_{+}:=P_+ \TU P_+$ and $\TU_{-}:=P_- \TU P_-$.
The following lemma gives the spectral decomposition of the operators $\TU_{\pm}$:

\begin{lemma}\label{ferm_spec}
The unitary operators $\TU_{\pm}$ can be spectrally decomposed as
$
\TU_{\pm}=\sum_{\ell = 0}^{d-1} e^{2\pi i \ell /d} P^{\pm}_\ell 
$, where the rank $\TM_{\ell}$ of the spectral projection $P^\pm_{\ell}$ is given
by the Fourier transform
\begin{equation} \label{eq:trans-mult-ferm}
\TM_{\ell}:=\frac{1}{d}\sum_{k=0}^{d-1} h(d,k)\, \exp(-2 \pi i k\ell/d)
\end{equation}
of $h(d,k)$ where $\ell\in\{0,\ldots,d-1\}$ and
\begin{align*}
h(d,k)
 :=
\begin{cases}
0 & \text{if } d/\gcd(d,k) \text{ is even,}\\
2^{\gcd(d,k)-1} & \text{if }  d/\gcd(d,k) \text{ is odd.}
\end{cases}
\end{align*}
\end{lemma}

\begin{proof}
We determine the spectral decomposition
of $\TU_+$ and $\TU_-$
along the lines
of Lemma~\ref{lemma:trans_inv_spinchain}.
Let $\mF_+(\C^d)$ denote
the subspace spanned
by those basis vectors of Eq.~\eqref{eq:fermi-basis} for
which $\bar{n}=\sum_{i=1}^d n_i$ is even. Likewise, $\mF_-(\C^d)$
corresponds to the case of odd $\bar{n}$.
As $(\TU_{\pm})^d=\unity_{\mF_{\pm}(\C^d)}$, the eigenvalues of $\TU_\pm$ 
are of the form $\exp(2\pi i \ell/d)$ with $\ell \in \{0,\ldots,d{-}1\}$. Hence,
the spectral decomposition is given by
$\TU_\pm=\sum_{\ell = 0}^{d-1} \exp(2\pi i \ell /d) P^{\pm}_\ell$.
We define representations $\DD_{\pm}$ of the cyclic group $\Z_d$ 
which map the $k$-th power of the generator $g\in\Z_d$ of degree $d$ to
$\DD_\pm(g^k):=\TU_{\pm}^k$.  
Note that
$\DD_{\pm}$ splits up into 
a direct sum 
$\DD_{\pm} \cong \oplus_{\ell \in \{0,\ldots,L{-}1 \}}   (D_\ell)^{\oplus \dim (P_\ell)}$
containing 
$\dim (P^{\pm}_\ell)$ copies of 
the one-dimensional representations satisfying
$D_\ell (g^k)=\exp(2 \pi i k \ell/d)$. 
The rank $\Tm^{\pm}_k$
of the projection $P^{\pm}_\ell$ is equal to  
the multiplicity of $D_\ell$ in the decomposition
of the reducible representation $\DD_{\pm}$.
This multiplicity can be computed as the  character scalar product
\begin{align*}
\Tm^{\pm}_k &=\frac{1}{d}\sum_{k=0}^{d-1}\, \tr[\DD_{\pm}(g^k)]\; 
{\tr[D_{\ell}(g^k)]}^{*} \\ &= \frac{1}{d}
\sum_{k=0}^{d-1} \tr[\DD_{\pm}(g^k)]\; \exp(-2 \pi i \ell k /d).
\end{align*} 
In the standard basis,
all matrix entries of $\DD_{\pm}(g^k)=\TU_{\pm}^k$ are elements of
the set $\{0,1,-1\}$.  It follows 
by repeated applications of Eq.~\eqref{Trans-ferm} that $\TU^k$
maps the basis vectors $| n_1, n_2, \ldots, n_d \rangle$ to 
$s | n_{\pi(1)}, n_{\pi(2)}, \ldots, n_{\pi(d)} \rangle$
where $\pi$ is a cyclic shift by $k$ positions and 
the sign $s$ is given by
\begin{equation}\label{transsign}
s:=
(-1)^{\left(\sum_{i=1}^{d-k}n_i\right)\left(\sum_{j=d-k+1}^{d}n_j\right)}.
\end{equation}
Recall from the proof of 
Lemma~\ref{lemma:trans_inv_spinchain} that a
bit string $(n_1, n_2, \ldots , n_N)$ is left invariant 
under a cyclic shift by $k$ positions
iff it is of the form
\begin{gather*}
(n_1, n_2, \ldots, n_{\gcd(d,k)},\ldots, 
n_1, n_2, \ldots, n_{\gcd(d,k)}) \, .
\end{gather*}
If $d/\gcd(d,k)$ is even, 
the sum $\bar{n}=\sum_{i=1}^d n_i$ is even
for all of the $2^{{\rm{gcd}}(d,k)}$
bit strings invariant under a cyclic shift by $k$ positions.
It follows that all the diagonal entries of $\TU_{-}^k$ are zero, while
$\TU_{+}^k$ has $2^{{\rm{gcd}}(d,k)}$ non-zero diagonal entries. 
The non-zero diagonal entries of $\TU_{+}^k$ are given by the number $s$ of Eq.~\eqref{transsign}.
Note that $s$ is $+1$ if $\sum_{j=1}^{d-k} n_j$ is even; and $-1$ otherwise.
Hence the frequencies of $+1$ and $-1$ in the set of diagonal entries are equal.
In summary,  $\tr(\TU_{\pm}^k)=0$
if $d/\gcd(d,k)$ is even.

Assume now that $d/\gcd(d,k)$ is odd. The sum $\bar{n}$ is odd 
for half of the $2^{{\rm{gcd}}(d,k)}$
bit strings and even for the other half. Applying again Eq.~\eqref{transsign}, we obtain always
a positive sign. Hence,
both traces $\tr(\TU_{\pm}^k)$ are
equal to $2^{{\rm{gcd}}(d,k)-1}$.
This completes the proof. 
\end{proof}

Lemma~\ref{ferm_spec} together with
Theorem~\ref{thm:double-centralizer} implies the following 
characterization of the system algebra for a translation-invariant fermionic system:

\begin{theorem}\label{translationinvariant}
Consider the translation-invariant Hamiltonians acting
on a  fermionic system with $d$ modes. The corresponding system algebra $\mathfrak{t}^f$
is given by
\begin{align*}
\mathfrak{t}^f& \cong \mathfrak{s}[\oplus_{\ell=0}^{d-1}\, \uu (\TM_\ell)]\oplus \mathfrak{s}[\oplus_{\ell=0}^{d-1}\, 
\uu (\TM_\ell)]\\
&\cong \left[\sum_{\ell=0}^{d-1} \su(\TM_\ell) + \su(\TM_\ell)\right] +  \sum_{\ell=1}^{2d-2} \uu(1),
\end{align*} 
where the numbers $\TM_\ell$ are defined in
Eq.~\eqref{eq:trans-mult-ferm}.
\end{theorem}

\begin{remark}
Note that $\TM_0 \ge \TM_\ell$ holds for any $\ell$ and that
$\sum_{\ell=0}^d \TM_{\ell}=2^{d-1}$. 
It follows that $\TM_0 \geq (2^{d-1}-1)/d$ and hence that the dimension of the system algebra in 
Theorem~\ref{translationinvariant} scales exponentially with $d$.
\end{remark}

\begin{remark}
Assuming that the number of modes is given by a prime number $p$, we can 
explicitly determine the numbers $\TM_\ell$ from Eq.~\eqref{eq:trans-mult-ferm}.
The corresponding system algebras are
\begin{equation}
\sum_{i=1}^{2} \su(F_p+1) + \sum_{i=1}^{2p-2}\su(F_p) + \sum_{i=1}^{2p-2}\uu(1),
\end{equation}
where $F_p=(2^{p-1}-1)/p$ is guaranteed to be an integer by Fermat's little theorem.
\end{remark}

\subsection{Fermionic Nearest-Neighbor Hamiltonians\label{ferm_nn}}

For  spin systems (see Section \ref{shortrange}) we verified
that the translation-invariant nearest-neighbor interactions together with the  on-site elements 
will never generate all 
translation-invariant operators, i.e.\ $\mathfrak{t}_L \ne \mathfrak{t}_2$ (if the number of 
spins $L$ is greater than two). This means that there exist certain
translation-invariant elements which cannot be generated by nearest-neighbor interactions and 
on-site elements, but we could not identify the explicit form of these 
translation-invariant elements for general $L$. In particular, it would be interesting
to know if $\mathfrak{t}_M \ne \mathfrak{t}_2$ holds for interaction lengths 
less than $M$ ($2<M<L$), where $M$ is independent of $L$.

In the case of fermionic systems, we can provide a result in this direction due to the 
restriction imposed by the parity superselection rule, which strongly 
limits the set of nearest-neighbor Hamiltonians.
 As we have discussed at the beginning of this section, the fermionic translation-invariant
Hamiltonians of nearest-neighbor type
are spanned by only six elements: $h_0$, $h_{\mathrm{rh}}$, $h_{\mathrm{ch}}$, $h_{\mathrm{rp}}$,
$h_{\mathrm{cp}}$, and $h_{\mathrm{int}}$ as defined in Eqs.~\eqref{Eq_lin_1}--\eqref{Eq_lin_6}.
We can show that  there exist next-nearest-neighbor or third-neighbor
interactions for odd $d\ge 5$ which
cannot be generated by these six Hamiltonians, while for  even $d\ge 6$
we provide a fourth-neighbor element.

Let $\mathfrak{t}^f_M$ denote the subalgebra of $\mathfrak{t}^f$ (see Theorem~\ref{translationinvariant}) 
which is generated by 
all elements of interaction length less than $M$. 
In particular, $\mathfrak{t}^f_2$ is generated by nearest-neighbor and on-site elements.
The result of this subsection is summarized in the following theorem:

\begin{theorem}\label{thm:main_nn}
Let us consider the Hamiltonian 
$$\hsodd := \sum_{n=1}^{d} i(f^{\dagger}_n f^{\phantom\dagger}_{n+3} - 
f^{\dagger}_{n+3} f^{\phantom\dagger}_{n}),$$ 
and fourth-neighbor Hamiltonian  $\hseven:=$
$$\sum_{n=1}^{d}(f^{\dagger}_n f^{\phantom\dagger}_n f^{\dagger}_{n+1} 
f^{\phantom\dagger}_{n+1}f^{\dagger}_{n+2} f^{\phantom\dagger}_{n+2}f^{\dagger}_{n+3} 
f^{\phantom\dagger}_{n+3}f^{\dagger}_{n+4} f^{\phantom\dagger}_{n+4}  - \tfrac{1}{32}\unity).$$
The generator $i\hsodd \in \mathfrak{t}^f_4$
is not contained in the system algebra $\mathfrak{t}^f_2$ generated by nearest-neighbor interactions and 
on-site elements if $d\geq 5$ is odd, while the element $i\hseven \in \mathfrak{t}^f_5$ is not contained 
in $\mathfrak{t}^f_2$ if $d\geq 6$ is even. Hence $\mathfrak{t}^f_2 \ne \mathfrak{t}^f_5 $ (when $d \ge 5$).
\end{theorem}

Note that the Hamiltonian $\hsodd$ of Theorem~\ref{thm:main_nn} is a third-neighbor 
Hamiltonian for $d\geq7$ and a next-nearest-neighbor Hamiltonian for $d=5$.
The proof of Theorem~\ref{thm:main_nn} is rather involved.
The proof for even $d$ is given in Appendix~\ref{app:main_nn_even},
while  Appendix~\ref{app:main_nn_odd} contains the proof for odd $d$.

\section{Quasifree Fermionic Systems Satisfying Translation-Invariance\label{TIQuasi}}

We continue the discussion of translation-invariant fermionic systems
from Sec.~\ref{sec:TI} by narrowing the scope to quadratic Hamiltonians.
In Sec.~\ref{transl_quadratic}, we derive the dynamic algebras 
for systems with and without (twisted) reflection symmetry.
Both of these cases are summarized for quasifree fermionic systems 
in Table~\ref{tab:trans-quad}: the system algebras  
were computed using the computer 
algebra system {\sc magma} \cite{MAGMA} for cases with low number of modes, while
the complete picture is provided by Theorem~\ref{thm:qf-trans-inv} and Corollary~\ref{Rsym}.
Sec.~\ref{orbits-q-trans} yields a classification of the orbit structure of pure translation-invariant 
quasifree states. This allows us to present an application to many-body physics in 
Sec.~\ref{applManyBody}, where we bound the scaling of the gap for a class of quadratic 
Hamiltonians.

\begin{table}[t]
\caption{\label{tab:trans-quad} System algebras of quasifree fermionic systems
 with $d$ modes satisfying translation-invariance.}
\begin{tabular}{l@{\hspace{2mm}}l@{\hspace{2mm}}l}
\hline\hline\\[-2mm]
$d$ & general case
& (twisted) reflection symm.\ \\ 
& (see Theorem~\ref{thm:qf-trans-inv}) & (see Eq.~\eqref{eq:Refl} and  Cor.~\ref{Rsym})\\[1mm] 
\hline \\[-2mm]
1 & - & - \\[1mm]
2 & $\uu(1)+\uu(1)$ & $\uu(1)+\uu(1)$  \\[1mm]
3 & $\uu(2)+\uu(1)$ & $\su(2)+\uu(1)$  \\[1mm]
4 &  $\uu(2)+\uu(1)+\uu(1)$ & $\su(2)+\uu(1)+\uu(1)$  \\[1mm]
5 &  $\uu(2)+\uu(2)+\uu(1)$ & $\su(2)+\su(2)+\uu(1)$ \\[1mm]
6 & $\uu(2)+\uu(2)+\uu(1)+\uu(1)$ & $\su(2)+\su(2)+\uu(1)+\uu(1)$
\\[1mm]
$\vdots$ &\quad$\vdots$ &\quad$\vdots$ \\[1mm] 
$2n{-}1$ &  $\sum_{i=1}^{n-1}\uu(2) + \uu(1)$ & $\sum_{i=1}^{n-1}\su(2) + \uu(1)$ \\[1mm]
$2n$ & $\sum_{i=1}^{n-1}\uu(2) + \uu(1) + \uu(1)$ & $\sum_{i=1}^{n-1}\su(2) + \uu(1) + \uu(1)$ 
\\[1mm]
\hline\hline
\end{tabular}
\end{table}

\subsection{Translation-Invariant Quadratic Hamiltonians\label{transl_quadratic}}

A quadratic Hamiltonian $H$ is translation-invariant (i.e.\ $[H,\TU]=0$)
iff the coefficient matrices $A$ and $B$ in Eq.~\eqref{Hqfree}
are cyclic (i.e.\ $A_{nm}-A_{n+1,m+1}=B_{nm}-B_{n+1,m+1}=0$). 
To study such Hamiltonians, it is useful to rewrite them in terms of the Fourier-transformed 
annihilation and creation operators
\begin{equation} \label{eq:F-trans_fermop}
\tilde{f}^{\phantom\dagger}_k:= \frac{1}{\sqrt{d}} \sum_{p=1}^d f_p e^{-2 \pi ipk/d} \text{ and }
\tilde{f}^\dagger_k:= \frac{1}{\sqrt{d}} \sum_{p=1}^d f^\dagger_p e^{2 \pi ipk/d},  
\end{equation}
with $k\in\{ 0,1,\ldots, d{-}1\}$, which satisfy again the
canonical anticommutation relations
\begin{equation} \label{eq:f-trans-car}
\{ \tilde{f}^{\dagger}_k, \tilde{f}^{\dagger}_{k'}\} = 
\{\tilde{f}^{\phantom\dagger}_k, \tilde{f}^{\phantom\dagger}_{k'}\} =0\; \text{ and }\;
\{\tilde{f}^{\dagger}_k, \tilde{f}^{\phantom\dagger}_{k'}\} = \delta_{kk'} \unity.
\end{equation}
A Hamiltonian from Eq.~\eqref{Hqfree} with cyclic $A$ and $B$ can now be rewritten as 
\begin{equation} \label{eq:quad.trans-inv}
H=\sum_{k=0}^{d-1} \tilde{A}_k   (\tilde{f}^{\dagger}_k \tilde{f}^{\phantom\dagger}_k- 
\tfrac{\unity}{2}) + \tfrac{1}{2}\tilde{B}_{k}\tilde{f}^{\dagger}_{k}\tilde{f}^{\dagger}_{d-k} - 
\tfrac{1}{2}\tilde{B}^*_{k}\tilde{f}^{\phantom\dagger}_{k}
\tilde{f}^{\phantom\dagger}_{d-k}
\end{equation}
applying
the definitions
$
\tilde{A}_k:=\sum_{p=1}^{d} A_{1p} \exp(-2 \pi ipk/d)
$
and 
$
\tilde{B}_{k}:= \sum_{p=1}^{d} B_{1p} \exp(-2 \pi ipk/d)
$, as well as the notation $\tilde{f}_d=\tilde{f}_0$.
The hermiticity of $A$ and the skew-symmetry of $B$
translates into the properties
$\tilde{A}_k=\tilde{A}_{d-k}^{*}$ and $\tilde{B}_k=-\tilde{B}_{d-k}$. This
allows us to decompose the Hamiltonian into a four-part sum
\begin{align} 
H=&\sum_{k=1}^{\lfloor (d-1)/2 \rfloor} \Im(\tilde{A}_k)\,  \oI\;
+ \sum_{k=1}^{\lfloor (d-1)/2 \rfloor} \Re(\tilde{B}_k)\,  \oX/2 \nonumber \\
& + \sum_{k=1}^{\lfloor (d-1)/2 \rfloor} \Im(\tilde{B}_k)\, \oY/2 \;
 +\sum_{k=0}^{\lfloor d/2 \rfloor} \Re(\tilde{A}_k)\, \oZ,
 \label{eq:quad.trans-inv2}
\end{align}
where one has the following definitions
\begin{gather}
\oI := i(\tilde{f}^{\dagger}_k \tilde{f}^{\phantom\dagger}_k {-} \tilde{f}^{\dagger}_{d-k} 
	\tilde{f}^{\phantom\dagger}_{d-k}),\,
\oX := (\tilde{f}^{\dagger}_k \tilde{f}^{\dagger}_{d-k} {+} \tilde{f}^{\phantom\dagger}_{d-k} 
	\tilde{f}^{\phantom\dagger}_{k}), \nonumber \\
\oY := i(\tilde{f}^{\dagger}_k f^{\dagger}_{d-k} {-} \tilde{f}^{\phantom\dagger}_{d-k} 
	\tilde{f}^{\phantom\dagger}_{k}),\,
\oZ := (\tilde{f}^{\dagger}_k \tilde{f}^{\phantom\dagger}_{k} {+} 
	\tilde{f}^{\dagger}_{d-k} \tilde{f}^{\phantom\dagger}_{d-k} {-} \unity) \label{ell-ops1}
\end{gather}
 with $k\in\{1,\ldots, \lfloor(d-1)/2 \rfloor \}$
 as well as
\begin{align}
\ell^{\mathrm{Z}}_{d/2} &:= (\tilde{f}^{\dagger}_{d/2} \tilde{f}^{\phantom\dagger}_{d/2}  - 
\unity/2) \; \text{ for $d$ even,}\nonumber \\
 \ell^{\mathrm{Z}}_0 &:= (\tilde{f}^{\dagger}_0 \tilde{f}^{\phantom\dagger}_{0} - \unity/2).
 \label{ell-ops2}
\end{align}

 Note that the operators $\ell^{\mathrm{Z}}_{d/2}$ (for $d$ even),
$\ell^{\mathrm{Z}}_0$, $\oZ$, $\oI$, $\oX$, and $\oY$ are linearly independent and
span the $(\lfloor d-1\rfloor +d)$-dimensional space of all
translation-invariant quadratic Hamiltonians.
For notational convenience we also introduce the dummy operators
$\ell^{\mathrm{Q}}_{d/2}:=0$ (assuming $d$ is even) and $\ell^{\mathrm{Q}}_{0}:=0$ for
$\mathrm{Q} \in \{\unity, \mathrm{X},  \mathrm{Y} \}$.

With these stipulations, we can characterize the system algebra:
\begin{theorem}\label{thm:qf-trans-inv}
Let $\fq_d$ denote the system algebra 
on a fermionic system with $d$  
modes which corresponds to
the set of Hamiltonians that are
translation-invariant and
quadratic. Then the Lie algebra $\fq_d$ is isomorphic to
$[\sum_{i=1}^{(d-1)/2} \uu(2)] + \uu(1)$
for odd $d$ and to
$[\sum_{i=1}^{(d-2)/2} \uu(2)] + \uu(1)+ \uu(1)$ for even $d$.
\end{theorem}
\begin{proof}
If $d=2m-1$ is odd,  
the generators $i\oI$, $i\oX$, $i\oY$, $i\oZ$, and $i\ell^{\mathrm{Z}}_{0}$
can be partitioned into $m$ 
pairwise-commuting sets, which each span linear subspaces as
\begin{equation*}
L_0:= \langle i\ell^{\mathrm{Z}}_0 \rangle_{\R}\; \text{ and }\;
L_k:= \langle i\oI, i\oX, i\oY, i\oZ \rangle_{\R}
\end{equation*}
with $k\in\{1,\ldots,m{-}1\}$. 
The commutation properties $[L_k,L_{k'}]=0$ (with $k \ne k'$)
follow from Eq.~\eqref{eq:f-trans-car}.
Moreover, $L_0$ is one-dimensional and forms a 
$\uu(1)$-algebra. 
Using Eq.~\eqref{eq:f-trans-car}, the relations 
$[\tilde{f}^{\dagger}_a \tilde{f}^{\phantom\dagger}_a, 
\tilde{f}^{\dagger}_a \tilde{f}^{\dagger}_b]=([\tilde{f}^{\dagger}_a \tilde{f}^{\phantom\dagger}_a,
\tilde{f}^{\phantom\dagger}_a \tilde{f}^{\phantom\dagger}_b ])^\dagger= \tilde{f}^{\dagger}_a 
\tilde{f}^{\dagger}_b$ 
and 
$[\tilde{f}^{\dagger}_a \tilde{f}^{\dagger}_b, \tilde{f}^{\phantom\dagger}_b 
\tilde{f}^{\phantom\dagger}_a]=\tilde{f}^{\dagger}_a \tilde{f}^{\phantom\dagger}_a +
\tilde{f}^{\dagger}_b \tilde{f}^{\phantom\dagger}_b - \unity$
can be deduced for $a \ne b$.
Substituting $k$ and $d-k$ into $a$ and $b$  in the previous formula, 
one can verify directly that the correspondence
\begin{equation*}
i \oI \mapsto i\unity,\;
i \oX \mapsto i\mathrm{X},\;
i \oY \mapsto i\mathrm{Y},\;
i \oZ \mapsto i\mathrm{Z}
\end{equation*}
provides an explicit Lie isomorphism between $L_k$ and $\uu(2)$.
If $d=2m$ is even, the system algebra consists of the 
above-described
generators supplemented with the element $i\ell_{d/2}^\mathrm{Z}$. This additional element 
commutes with all the other generators
and---therefore---provides an additional $\uu(1)$.
\end{proof}

The isomorphism between $L_k$ and $\uu(2)$ as given in the proof
leads to a compact formula
for the time evolution (in the Heisenberg picture) of the elements of $L_k$.
Since the operators $\oX$, $\oY$, $\oZ$, and $\oI$ (with $k\in\{1,\ldots, \lfloor 
(d-1)/2 \rfloor \}$) satisfy the same commutation relations as the Pauli matrices 
$\mathrm{X}$, $\mathrm{Y}$, $\mathrm{Z}$, and $\unity$, their time-evolution generated by
the Hamiltonian $H$ in Eq.~\eqref{eq:quad.trans-inv2} can be straightforwardly related to a 
qubit time-evolution
\begin{gather}
e^{iHt} i\left( a_\unity \oI + a_\mathrm{X} \oX + a_\mathrm{Y} \oY + a_\mathrm{Z} \oZ\right) 
e^{-iHt}= \nonumber \\
i a_\unity \oI + \sum_{\mathrm{Q} \in \{\mathrm{X},\mathrm{Y},\mathrm{Z}\}} i a_\mathrm{Q}\, 
\ell^\mathrm{Q}_k  \tr(e^{iH_s} \mathrm{Q}\, e^{-iH_s} \mathrm{Q}), \label{eq:te-qf-trans-inv}
\end{gather}
where $H_s=\Re(\tilde{A}_k) \mathrm{Z} + 
\Re(\tilde{B}_k) \mathrm{X}/2 + \Im(\tilde{B}_k) \mathrm{Y}/2$.

The \emph{twisted reflection symmetry} plays an important
role in translation-invariant quasifree fermionic systems.
It is defined by the unitary 
\begin{equation} \label{eq:Refl}
\eR  | n_1, n_2, \ldots n_d \rangle =  i^{(\sum_{\ell=1}^d n_\ell)^2} | n_d, n_{d-1}, \ldots n_1 \rangle . \,
\end{equation}
whose adjoint action 
on creation operators and their Fourier transforms are specified as
\begin{align} \label{eq:Refl_adj}
\eR\, f^\dagger_\ell\, \eR^{\dagger}=if^\dagger_{(d-\ell+1 \bmod d)}, \,
\eR\, \tilde{f}^\dagger_k \, \eR^{\dagger}=\tilde{f}^\dagger_{(-k \bmod d)}.
\end{align}
A given translation-invariant quasifree Hamiltonian is $\eR$-symmetric  (i.e.\
$[\eR,H]=0$) iff the coefficient matrix is restricted to be real. 
In our language, these Hamiltonians are exactly the ones for which
$\Im(\tilde{A}_k)=0$, i.e., the corresponding generators 
are spanned by the operators $i\ell_{d/2}^\mathrm{Z}$ (for $d$ even), 
$i\ell_{0}^\mathrm{Z}$, $i\oZ$,
$i\oX$, and $i\oY$.
From the proof of Theorem~\ref{thm:qf-trans-inv} one can immediately deduce
the corresponding system algebra:
\begin{corollary}\label{Rsym}
Consider a fermionic system with $d$ modes and
the set of quadratic Hamiltonians
which are translation-invariant and $\eR$-symmetric.
The corresponding system algebra $\fq_d^{\eR}$ 
  is isomorphic to $[\sum_{i=1}^{(d-1)/2} \su(2)] + \uu(1)$
for odd $d$ and to $[\sum_{i=1}^{(d-2)/2} \su(2)] + \uu(1)+ \uu(1)$ for even $d$.
\end{corollary}

Given the system algebras $\fq_d$ and 
$\fq_d^{\eR}$, we investigate the subalgebras generated by short-range Hamiltonians.
It will be useful to introduce for $p \in \{1, \ldots, \lfloor (d-1)/2 \rfloor \}$ the Hamiltonians 
\begin{subequations}\label{hZ}
\begin{align}
\hI_{p} &:=  \frac{1}{2}\sum_{\ell=1}^d i(f^{\dagger}_\ell f^{\phantom{\dagger}}_{\ell+p}{-}f^{\dagger}_{\ell+p} 
f^{\phantom{\dagger}}_{\ell})
= \sum_{k=0}^{\lfloor (d-1)/2 \rfloor} \sin(\tfrac{2 \pi kp}{d})\; \oI, \label{hZZ}\\
\hX_{p} &:=   \frac{1}{2}\sum_{\ell=1}^d i(f^{\dagger}_\ell f^{\dagger}_{\ell+p}{-}f^{\phantom{\dagger}}_{\ell+p} 
f^{\phantom{\dagger}}_{\ell})
= \sum_{k=0}^{\lfloor (d-1)/2 \rfloor} \sin(\tfrac{2 \pi kp}{d})\; \oX, \\ 
\hY_{p} &:=   \frac{1}{2}\sum_{\ell=1}^d (f^{\dagger}_\ell f^{\dagger}_{\ell+p}{+}f^{\phantom{\dagger}}_{\ell+p} 
f^{\phantom{\dagger}}_{\ell} )
= \sum_{k=0}^{\lfloor (d-1)/2 \rfloor} \sin(\tfrac{2 \pi kp}{d})\; \oY,\\
\hZ_{p} &:= \frac{1}{2}\sum_{\ell=1}^d (f^{\dagger}_\ell f^{\phantom{\dagger}}_{\ell+p}{+}f^{\dagger}_{\ell+p} 
f^{\phantom{\dagger}}_{\ell} )
=\sum_{k=0}^{\lfloor d/2 \rfloor}  \cos(\tfrac{2 \pi kp}{d})\; \oZ,
\end{align}
as well as the additional ones ($\hZ_{d/2}$ only for even $d$)
\begin{align}
\hZ_{0} &:= \frac{1}{2}\sum_{\ell=1}^d (f^{\dagger}_\ell f^{\phantom{\dagger}}_{\ell}{+}f^{\dagger}_{\ell} 
f^{\phantom{\dagger}}_{\ell} {-} \unity)
=\sum_{k=0}^{\lfloor d/2 \rfloor}  \oZ,\\
\hZ_{d/2} &:= \frac{1}{2}\sum_{\ell=1}^d (f^{\dagger}_\ell f^{\phantom{\dagger}}_{\ell+p}{+}f^{\dagger}_{\ell+p} 
f^{\phantom{\dagger}}_{\ell} )
=\sum_{k=0}^{d/2}  (-1)^k\; \oZ.
\end{align}
\end{subequations}
In these definition we used cyclic indices, e.g.\ $f_{d+a}=f_a$. 
The operators $\hZ_{d/2}$ (for $d$ even), $\hZ_{0}$, $\hZ_{p}$, $\hI_{p}$, $\hX_{p}$, and $\hY_{p}$
span $\fq_d$ linearly. Using the identities above, the commutation relations of the 
$\ell^{\mathrm{Q}}_k$ operators, and some trigonometric identities, we obtain 
\begin{subequations}\label{eq:trans-com}
\begin{align}
[i\hI_{a},i\hZ_{b}]=&[i\hI_{a},i\hX_{b}]=[i\hI_{a},i\hY_{b}]=0,  \\
[i\hX_{a},i\hY_{b}]=&-\tfrac{i}{2}(\hZ_{(a+b)\bmod \lfloor d/2 \rfloor}{-}\hZ_{(a-b)\bmod \lfloor d/2 \rfloor}), \\
[i\hY_{a},i\hZ_{b}]=&-\tfrac{i}{2}(\mathrm{sgn}(d{-}a{-}b)\,\hX_{(a+b)\bmod \lfloor d/2 \rfloor}
\nonumber\\& \phantom{-\tfrac{i}{2}(}-\mathrm{sgn}(a{-}b)\,\hX_{(a-b)\bmod \lfloor d/2 \rfloor}),\\
[i\hZ_{a},i\hX_{b}]=&-\tfrac{i}{2}(\mathrm{sgn}(d{-}a{-}b)\,\hY_{(a+b)\bmod \lfloor d/2 \rfloor}\nonumber\\
&\phantom{-\tfrac{i}{2}(}-\mathrm{sgn}(a{-}b)\,\hY_{(a-b)\bmod \lfloor d/2 \rfloor}) 
\end{align}
\end{subequations}
for $a,b\in \{0,\ldots, \lfloor d/2 \rfloor\}$.
In \cite{kraus-pra71} it was shown that already 
the nearest-neighbor Hamiltonians of $\fq_d^{\eR}$
generate the whole $\fq_d^{\eR}$. Now we are in the position to provide a more 
systematic proof of their result:

\begin{lemma} \label{thm:Kraus}  
The system algebra $\fq_d^{\eR}$ can be generated using the
one-site-local operator $i\hZ_{0}$ and a
nearest-neighbor element
$i(\alpha_1 \hZ_{1} + \alpha_2 \hX_{1} + \alpha_3 \hY_{1})$ 
with $\alpha_i\in\R$ assuming that
$\alpha_2\neq 0$ or $\alpha_3 \ne 0$ for odd $d$
and additionally requiring $\alpha_1 \ne 0$ for even $d$.
\end{lemma}
\begin{proof} 
(1) From Eq.~\eqref{eq:trans-com} we know that $i\hZ_{0}$, $i\hZ_{1}$,
$i\hX_{1}$, and $i\hY_{1}$ would generate the whole $\fq_d^{\eR}$. 
(2) Suppose that $\alpha_1 \ne 0$ and $\alpha_2^2 +\alpha_3^2 \ne 0$.
From  $2[i\hZ_{0}, i(\alpha_1 \hZ_{1} + \alpha_2 \hX_{1} + \alpha_3 \hY_{1})]=
\alpha_2 i\hY_{1} -  \alpha_3 i\hX_{1}$ and $2[i\hZ_{0},\alpha_2 i\hY_{1} -  \alpha_3 i\hX_{1}]
=-\alpha_2 i\hX_{1} -  \alpha_3 i\hY_{1}$ it follows that one can generate 
$i\hZ_{0}$, $i\hZ_{1}$, $i\hX_{1}$, and $i\hY_{1}$. Hence according 
to observation (1),  the 
whole $\fq_d^{\eR}$ is generated.
(3) Suppose now that $\alpha_1=0$, $d$ is odd, and $\alpha_2^2 +\alpha_3^2 \ne 0$.  
From $2[i\hZ_{0}, i(\alpha_2 \hX_{1} + \alpha_3 \hY_{1})]=
\alpha_2 i\hY_{1} -  \alpha_3 i\hX_{1}$ one can generate 
$i\hZ_{0}$,  $i\hX_{1}$, and $i\hY_{1}$. From Eq.~\eqref{eq:trans-com} it follows that 
these generators in turn generate all $i\hZ_{2p \mod d}$. Since $d$ is odd, 
$i\hZ_{1}$ is also generated. Hence we obtain 
$i\hZ_{0}$, $i\hZ_{1}$, $i\hX_{1}$, and $i\hY_{1}$, and according to (1), 
the algebra $\fq_d^{\eR}$ is generated.
\end{proof}
For the more general $\fq_d$, we obtain a slightly larger system algebra when we do \emph{not} 
assume  $\eR$-symmetry:
\begin{proposition}\label{fermionic_nn}
The elements of $\fq_d$ with interaction length 
less than $M$ (where $2 \le M \le \lceil d/2 \rceil$ and $d \ge 3$) generate a system algebra 
which is isomorphic to $[\sum_{i=1}^{(d-1)/2} \su(2)] + \sum_{i=1}^{M}\uu(1)$
for odd $d$ and to $[\sum_{i=1}^{(d-2)/2} \su(2)] + \sum_{i=1}^{M+1}\uu(1)$ for even $d$.
\end{proposition}
\begin{proof}
From Lemma~\ref{thm:Kraus} we know that the operators $h^{\mathrm{Q}}_{a}$ with
$\mathrm{Q} \in \{ \mathrm{X},\mathrm{Y},\mathrm{Z}\} $ already generate $\fq_d^{\eR} $ 
which is isomorphic to $[\sum_{i=1}^{(d-1)/2} \su(2)] + \uu(1)$
for odd $d$ and to $[\sum_{i=1}^{(d-2)/2} \su(2)] + \uu(1)+ \uu(1)$ for even $d$. 
We have $M{-}1$ additional operators $\hI_{q}$ with $q\in\{1, 
\ldots, M{-}1 \}$ which are linearly independent and commuting.  These generate the 
other parts corresponding $\sum_{i=1}^{M-1} \uu(1)$.
\end{proof}

We illustrate Lemma~\ref{thm:Kraus} and Proposition~\ref{fermionic_nn} with
a fermionic ring of $d=6$ modes.
Suppose that the drift Hamiltonian of this system
is the nearest-neighbor hopping Hamiltonian $i\hZ_{1} =   \tfrac{i}{2}\sum_{\ell=1}^6 
(f^{\phantom\dagger}_\ell f^{\dagger}_{\ell+1}{+}f^{\dagger}_{\ell+1} 
f^{\phantom{\dagger}}_{\ell})$, and that one can additionally control  
the on-site potential $i\hZ_{0} = \tfrac{i}{2}\sum_{\ell=1}^6 (f^{\dagger}_\ell 
f^{\phantom{\dagger}}_{\ell}{-} \tfrac{1}{2}\unity)$,
the pairing strength $i\hY_{1}=\frac{i}{2}\sum_{\ell=1}^6 (f^{\dagger}_\ell f^{\dagger}_{\ell+1}{+}
f^{\phantom{\dagger}}_{\ell+1} 
f^{\phantom{\dagger}}_{\ell} )$, and the magnetic flux $i\hI_{1} =  -\frac{1}{2}\sum_{\ell=1}^6 
(f^{\dagger}_\ell f^{\phantom{\dagger}}_{\ell+1}{-}f^{\dagger}_{\ell+1} 
f^{\phantom{\dagger}}_{\ell})$ in the ring. 
Lemma~\ref{thm:Kraus} implies that
the first three Hamiltonians generate the Lie algebra $\fq_6^{\eR}$ of all Hamiltonians
which are simultaneously
$\eR$-invariant, translation-invariant, and quadratic.
The magnetic flux term $i\hI_{1}$ commutes with all elements 
of $\fq_6^{\eR}$ and contributes only an additional $\uu(1)$ to the system
algebra. Thus, the system algebra generated by all nearest-neighbor quadratic Hamiltoinans
that are translation-invariant is given by $ \fq_6^{\eR} {+}\uu(1)\cong \su(2) {+}\su(2) {+} \uu(1) {+} \uu(1) {+} \uu(1)$.
In order to achieve full controllability for a translation-invariant and quasifree fermionic system 
(which corresponds to the Lie algebra $\fq_6 \cong \su(2) {+}\su(2) {+} \uu(1) {+} \uu(1) {+} \uu(1) {+}\uu(1)$), 
one has to add a next-nearest neighbor Hamiltonian as $i\hI_{2} =  -\frac{1}{2}\sum_{\ell=1}^6 
(f^{\dagger}_\ell f^{\phantom{\dagger}}_{\ell+2}{-}f^{\dagger}_{\ell+2} 
f^{\phantom{\dagger}}_{\ell})$.

\subsection{Orbits of Pure Translation-Invariant Quasifree States\label{orbits-q-trans}}

We characterize now the orbits of pure translation-invariant quasifree states 
under the action of translation-invariant quadratic Hamiltonians.
Since the operators $\oI=i(\tilde{f}^{\dagger}_k \tilde{f}^{\phantom\dagger}_k {-} 
\tilde{f}^{\dagger}_{d-k} \tilde{f}^{\phantom\dagger}_{d-k})$ commute with all the 
other translation-invariant quadratic Hamiltonians (as discussed in 
Sec.~\ref{transl_quadratic}), 
their expectation values stay invariant under the considered time evolutions. 
At the end of the section, we show that these invariant expectation values even form a separating set 
of invariants for the orbits of pure translation-invariant quasifree states. 

Let us recall that 
a quasifree state is fully characterized by its Majorana covariance matrix, 
defined in Eq.~\eqref{eq:cov-mat}.
The translation unitary $\TU$ acts on the Majorana operators by conjugation as
$\TU\, m_p\, \TU^\dagger=m_{(p+2 \mod 2d)}$. It follows that 
a quasifree state $\rho$ is translation-invariant (i.e.\ $[\rho,\TU]=0$)
iff its covariance matrix $G_{pq}$ 
is doubly-cyclic, i.e.\ $G_{pq}=G_{(p+2 \mod 2d),(q+2 \mod 2d)}$.
The double-cyclicity of $G$ implies that it can be expressed
as a block-Fourier transform of a block-diagonal matrix, i.e.
\begin{equation} \label{eq:block-fourier}
\tilde{G}=U_F\, G\, U_F^{\dagger} \, ,
\end{equation}
where 
$U_F:= 
\left(
\begin{smallmatrix}
1 & 0\\
0 & 1
\end{smallmatrix}
\right)
\otimes W$ with
$W_{pq}:=\exp(2 \pi i /d)^{q-p}$ and
$\tilde{G}=\oplus_{k=0}^{d-1} i \g(k)$
with $\g(k)$ being 
$2\times 2$-matrices.
The matrices $\g(k)$ can be calculated by the inverse 
block-Fourier transform

\begin{equation} \label{eq:g-def}
\g(k)=-i \sum_{\ell=1}^{d}  e^{2\pi k\ell i/d}
\begin{pmatrix}
G_{1,2\ell-1} & G_{1, 2\ell}\\
G_{2, 2\ell-1} & G_{2,2\ell}
\end{pmatrix}.
\end{equation}
The fact that $G$ is skew-symmetric and real implies
\begin{equation}\label{eq:g-transp}
\g(d{-}k)=-\g^T(k).
\end{equation}  
Moreover, due to Eq.~\eqref{eq:block-fourier} the set of eigenvalues of all the matrices 
$\g(k)$ equals the one of $-iG$ (including multiplicities).
Combining these observations with Proposition~\ref{prop:cov-sing-val} and 
Proposition~\ref{prop:qf-stab}, we obtain the following characterization of pure 
translation-invariant quasifree states:

\begin{lemma}
A set of  $2 \times 2$ matrices $\g(k)$ (with $k\in\{0, \ldots, d{-}1\}$)
defines a covariance matrix of a pure quasifree state through Eq.~\eqref{eq:block-fourier} 
iff they satisfy Eq.~\eqref{eq:g-transp} and their eigenvalues are in the set $\{1, -1\}$.
\end{lemma}

The entries of $\g(k)$ and the expectation values of the $\ell_k$ operators
defined in Eq.~\eqref{ell-ops1} 
can be related by 
\begin{equation}\label{eq:g-l-rel}
i\g(k)=
\unity_2 \langle \oI \rangle
+\mathrm{X} \langle \oX\rangle+
\mathrm{Y} \langle \oY \rangle
+ \mathrm{Z} \langle \oZ \rangle
\end{equation}
using Eq.~\eqref{eq:g-def} and
the definitions for  $\oI$,  $\oX$,  $\oY$, and  $\oZ$.
Now we can prove the main theorem of this subsection:

\begin{theorem}\label{orbit_pure_translation}
Two pure quasifree states $\rho_1$ and $\rho_2$ can be connected  
through the action of a translation-invariant quadratic Hamiltonian
if and only if 
$
\tr(\rho_1 \oI)=\tr(\rho_2 \oI)
$
holds for all $\oI$ with  $k \in \{0, \ldots, \lfloor (d{-}1)/2 \rfloor\} $.
\end{theorem}

\begin{proof}
First, we consider the `if'-case:
Let $H$ be a translation-invariant 
quadratic Hamiltonians for which $\rho_1=e^{-iHt} \rho_2 e^{iHt}$ holds.
Since the operators $\oI$  commute with any translation-invariant Hamiltonian, 
we have that
$
\tr(\rho_1 \oI)=\tr(e^{-iHt} \rho_2 e^{iHt} \oI)=
\tr(\rho_2 e^{iHt} \oI e^{-iHt})= \tr(\rho_2 \oI)
$.
Second, we treat the `only if'-case:
Let $\g_1(k)$ and 
$\g_2(k)$ denote the Fourier-transformed Majorana 
two-point functions (defined as in
Eq.~\eqref{eq:g-def})
of $\rho_1$ and $\rho_2$,
respectively. The action of a translation-invariant Hamiltonian,
 $\rho_a \mapsto e^{-iH} \rho_a e^{iH}$ 
is represented by the map
\begin{equation}\label{eq:g-transf}
\g_a(k) \mapsto U(k)\, \g_a(k)\, U(k)^{\dagger} \,
\end{equation}
where $U(k)$ is given by $\exp[-i\Re(\tilde{A}_k) \mathrm{Z} - 
i\Re(\tilde{B}_k) \mathrm{X}/2 - i \Im(\tilde{B}_k) \mathrm{Y}/2]$.
Using Eq.~\eqref{eq:g-l-rel}, we obtain
$\tr(\rho_a \oI)= i \tr[g_a(k)]$ for $a\in\{1,2\}$.
These expectation values have to be in the set $\{-2,0,2\}$, since
the eigenvalues of $\g_1(k)$ and $\g_2(k)$ are 
in the set $\{-1,1\}$.
Then, it follows from
$\tr(\rho_1 \oI)=\tr(\rho_2 \oI)$ 
that
the expectation values of $\g_1(k)$ and $\g_2(k)$
coincide. Thus, we obtain from Eq.~\eqref{eq:g-transf}
that $\rho_1$ and $\rho_2$ can be transformed into each other.
\end{proof}

Finally, we turn to the $\eR$-symmetric setting,
as introduced in Sec.~\ref{transl_quadratic},
and determine  the orbit structure of quasifree pure states
which are translation-invariant and  $\eR$-symmetric
under the action of operators in $\fq_d^{\eR}$.

\begin{proposition}\label{prop:qRd_transl}
The unitaries generated by the Lie algebra $\fq_d^{\eR}$
act transitively on the set of
quasifree pure states
which are translation-invariant and $\eR$-symmetric.
\end{proposition}
\begin{proof}
Since  $\eR\oI\eR^{-1}=-\oI$, the expectation value of
these operators in $\eR$-symmetric states must vanish as $\tr(\rho \oI)=-\tr(\rho \eR\oI\eR^{-1})=
-\tr(\eR^{-1}\rho\eR\oI)=- \tr(\rho \oI)$. 
Moreover, by Theorem~\ref{orbit_pure_translation} we know that 
two pure translation-invariant states are on the same $\fq_d$-orbit 
iff the expectation values of the $\oI$ operators coincide for 
all $k \in \{0, \ldots, \lfloor (d-1)/2 \rfloor\}$. Hence the 
translation-invariant $\eR-$symmetric states lie on the same 
$\fq_d$-orbit.
As Eq.~\eqref{eq:g-transf} implies that the 
$\fq_d$-orbits are equivalent to $\fq_d^{\eR}$-orbits, 
we have proved the proposition.
\end{proof}

\subsection{An Application to Many-Body Physics\label{applManyBody}}
In many-body physics, one of the important characteristics of
quantum criticality is the \emph{closing of the 
gap}. This means that the energy difference between the ground state and 
the first excited state  goes to zero in the thermodynamic limit, when the 
number of spins or fermionic modes goes to infinity. 
Quasifree fermionic models can display both gapped and gapless behavior. 
Using the techniques developed in the previous subsections, 
we will prove that the gap always disappears (i.e.\ closes) for 
translation-invariant quasifree models if
the coefficient matrix $A$ of Eq.~\eqref{Hqfree} is 
purely imaginary while $B$  is an arbitrary, complex skew-symmetric matrix.

To formalize this statement, let us consider a set
 $a_r$  of fixed (finite) real numbers with $r \in \{1, \ldots, M{-}1 \}$ and a 
set $b_r$  of fixed complex numbers (of finite modulus) with $r \in \{1, \ldots, M{-}1 \}$.
With these stipulations, we define for any $d\ge 2M$ the cyclic $d \times d$ 
matrices $A_d$ and $B_d$ (or $A$ and $B$ for short)
by specifying their entries
\begin{align}
A_{pq}
&:=
\begin{cases} 
ia_{q-p} & \text{if $q-p \in \{1, \ldots, M{-}1 \}$,}\\
 -ia_{p-q} & \text{if  $p-q \in \{1, \ldots, M{-}1 \}$,} \\
0 & \text{otherwise,}\\
\end{cases}  \label{eq:A-for-gapless} 
\intertext{and}
B_{pq}
&:=
\begin{cases}
b_{q-p} & \text{if $q-p \in \{1, \ldots, M{-}1 \}$,}\\
 -b_{p-q} & \text{if  $p-q \in \{1, \ldots, M{-}1 \}$,}  \\
0 & \text{otherwise.}
\end{cases} \label{eq:B-for-gapless}
\end{align}
By applying these definitions to Eq.~\eqref{Hqfree} we obtain:
\begin{theorem}\label{thm:gap_close}
Given the positive integers $d$ and $M$ with $d\ge 2M$, consider the 
corresponding translation-invariant quasifree Hamiltonian 
\begin{equation*}
H_d=\sum_{p,q=1}^d A_{pq} (f_p^{\dagger}f_q{-}\delta_{pq} \tfrac{\unity}{2}){+}
\tfrac{1}{2}B_{pq} f_{p}^{\dagger}f_{q}^{\dagger}{-}
\tfrac{1}{2}B_{pq}^{*} f_p f_{q},
\end{equation*}
where $A$ and $B$ are defined in Eqs.~\eqref{eq:A-for-gapless} and \eqref{eq:B-for-gapless}.
Assume that $H_d$ has a unique ground state. Then the gap $\Delta_d$ of $H_d$
is bounded by $\Delta_d \le \frac{8\pi (M-1)}{d}\sum_{p=1}^{M-1} (|a_p|+|b_p|)$, i.e.\ the 
gap closes algebraically in the thermodynamic limit of  $d$ 
going to infinity.
\end{theorem}

\begin{proof}
Since $H_d$ is translation-invariant and its coefficient matrix is 
imaginary, it can be decomposed in terms of the operators $\ell^{\mathrm{Q}}_{k}$ with
$\mathrm{Q} \in \{\unity, \mathrm{X},  \mathrm{Y} \}$ and 
$k \in \{1, \ldots, \lfloor (d{-}1)/2 \rfloor\} $ as
\begin{equation*}
H_d = \sum_{k=1}^{\lfloor (d{-}1)/2 \rfloor} \tilde{a}_k \oI  +\tfrac{1}{2}\tilde{b}^X_k \oX 
+\tfrac{1}{2} \tilde{b}^Y_k\oY,
\end{equation*}
using the definitions $\tilde{a}_k := - \sum_{p=1}^{M-1} a_{p} \sin(-2 \pi pk/d)$, 
$\tilde{b}^X_k := - \Re[\sum_{p=1}^{M-1} b_p \sin(-2 \pi pk/d)]$,
as well as
$\tilde{b}^Y_k := - \Im[\sum_{p=1}^{M-1} b_p \sin(-2 \pi pk/d)]$.
Let $\rho_d$ be a pure quasifree state, and let $\g_{d}(k)$ denote its
Fourier-transformed Majorana two-point functions (see Eq.~\eqref{eq:g-def}).
From Eq.~\eqref{eq:te-qf-trans-inv} we know that $\rho_d$
is an eigenstate of $H_d$ iff $[\tilde{b}^X_k \mathrm{X} + \tilde{b}^Y_k\mathrm{Y},
\g_{d}(k)]=0$. The eigenvalue of $H_d$ corresponding to this state is given by
\begin{align}
\tr(\rho_d H_d)=\sum_{k=1}^{\lfloor (d{-}1)/2 \rfloor} \tr[ig_{d}(k)(\tilde{a}_k\unity_2{+}
\tfrac{1}{2}\tilde{b}^X_k\mathrm{X}{+}\tfrac{1}{2}\tilde{b}^Y_k\mathrm{Y})]. \label{eq:eigen-energy}
\end{align}
Let us emphasize that the proof builds on the fact
that $M$ is fixed and finite, while $d$ goes to infinity in the thermodynamic limit.
Among the eigenstates of $H_d$, consider the (unique) ground state
$\rho^{d}_{gs}$, whose Fourier-transformed Majorana 
two-point functions (see Eq.~\eqref{eq:g-def}) will be denoted by $\g^d_{gs}(k)$. 
From this ground state let us construct another quasifree state $\rho^{d}_{e}$ which 
is defined through its Majorana 
two-point functions
\begin{align}
\g^d_{e}(1)
 :=
\begin{cases}
\phantom{-} \unity_2 & \text{if $\g^d_{gs}(1) \ne -\unity_2$,}\\
- \unity_2 & \text{otherwise,}
\end{cases} \nonumber
\end{align}
while for general $k \ne 1$ we assign $\g^d_{e}(k):=\g^d_{gs}(k)$.

The corresponding pure quasifree state $\rho^d_{e}$
is an eigenstate of $H_d$, since according to 
Eq.~\eqref{eq:g-transf} its Fourier-transformed Majorana 
two-point function stays invariant during the time-evolution generated by $H_d$.
Using Eq.~\eqref{eq:eigen-energy}, we can calculate the difference 
between the energies
corresponding to $\rho^{d}_{gs}$ and $\rho^{d}_{e}$ as
\begin{align*}
&\Delta_d := \tr [ (\rho^{d}_{e} - \rho^{d}_{gs}) H_d]\\[1mm]
&=\sum_{k=1}^{\lfloor (d{-}1)/2 \rfloor} \tr\left([\g^d_{e}(k)-\g^d_{gs}(k)](\tilde{a}_k\unity_2{+}
\tfrac{1}{2}\tilde{b}^X_k\mathrm{X}{+}\tfrac{1}{2}\tilde{b}^Y_k\mathrm{Y})\right) \\[1mm]
&=\tr\left([\g^d_{e}(1)-\g^d_{gs}(1)](\tilde{a}_1\unity_2{+}
\tfrac{1}{2}\tilde{b}^X_1\mathrm{X}{+}\tfrac{1}{2}\tilde{b}^Y_1\mathrm{Y})\right) \\[2mm]
&\le 2\| [\g^d_{e}(1)-\g^d_{gs}(1)](\tilde{a}_{1}\unity_2{+}
\tfrac{1}{2}\tilde{b}^X_{1}\mathrm{X}{+}\tfrac{1}{2}\tilde{b}^Y_{1}\mathrm{Y}) \|  \\[2.5mm]
&\le 4(|\tilde{a}_1| +\tfrac{1}{2} |\tilde{b}^X_1|+\tfrac{1}{2} |\tilde{b}^Y_1|  )  \\[1mm]
&\le 4 \left|\sum_{p=1}^{M-1} a_{p} \sin(2 \pi p/d)\right| +
4 \left|\sum_{p=1}^{M-1} b_{p} \sin(2 \pi p/d)\right| \\
&\le  \frac{8\pi (M{-}1)}{d}\sum_{p=1}^{M-1} (|a_p|+|b_p|)
\end{align*}
This completes the proof of the theorem.
\end{proof}

\section{Particle-Number Conserving Systems}\label{sec:PC}
Finally, we treat fermionic systems whose particle-number is conserved.   
The corresponding system algebras are given both in the general case  as well as in the quasifree case.
Furthermore, a necessary and sufficient condition for quasifree pure-state controllability
in this setting is provided.

\subsection{The System Algebra of Particle-Number Conserving Hamiltonians}

Let $P_n$ denote the orthogonal projection from the  Fock space 
$\mathcal{F}(\C^d)= \oplus_{n=0}^{d} \wedge^n \C^d$
onto the $n$-particle subspace $\wedge^n \C^d \subset \mathcal{F}(\C^d)$
of dimension $\binom{d}{n} $.
The \emph{particle-number operator} $\hat{n}$ of a fermionic system 
is defined as
$
\hat{n}:=\sum_{n=0}^d n P_n
$.
Note that
$
\sum_{p=1}^d f^{\dagger}_pf^{\phantom\dagger}_p \psi_n=n \psi_n
$ holds
for any $\psi_n \in \wedge^n\C^d$.
Hence, the particle number operator can also be expressed as 
\begin{equation*}
\hat{n}=\sum_{p=1}^d f^{\dagger}_pf^{\phantom\dagger}_p \, .
\end{equation*}
A fermionic Hamiltonian $H$ is called particle-number conserving if it commutes with
$\hat{n}$. Using the general Theorem~\ref{thm:double-centralizer} of 
Appendix~\ref{DoubleCentralizers}, one directly obtains the 
corresponding system algebra.
\begin{proposition} \label{prop:particle_number}
The system algebra of particle-number conserving fermionic  
interactions with $d$ modes is 
\begin{equation*}
\mathfrak{s}\left(\bigoplus_{n \, \text{even}}\; \uu\hspace{-0.75mm}\left[\tbinom{d}{n}\right]\right) 
\oplus\mathfrak{s}\left(\bigoplus_{n \, \text{odd}}\; \uu\hspace{-0.75mm}\left[\tbinom{d}{n}\right]\right) .
\end{equation*}
\end{proposition}

\subsection{Quadratic Hamiltonians} \label{subsec:part-numb-quad}

A quadratic Hamiltonian $H$ is particle-number conserving iff its
coefficient matrix $B$ of Eq.~\eqref{Hqfree} is zero, i.e., iff
$H=\sum_{p,q=1}^d A_{pq}(f^{\dagger}_ pf^{\phantom\dagger}_q-\delta_{pq}\tfrac{\unity}{2})$
where $A$ denotes any hermitian matrix. 
The corresponding system algebra is given by the following proposition:

\begin{proposition}\label{quasifree_particle}
The system algebra of the particle-number conserving quadratic $d$-mode  
Hamiltonians 
is isomorphic to $\uu(d)$.  
\end{proposition}
\begin{proof}
Let $\iota$ denote 
the $\mathbb{R}$-linear mapping  from 
the $d$-mode Hamiltonians 
which are quadratic and particle-number conserving
to
the $d \times d$ skew-hermitian matrices. We define $\iota$ using
$\iota (i(f^\dagger_p f^{\phantom\dagger}_q {+} f^\dagger_q f^{\phantom\dagger}_p {-}\tfrac{\unity}{2})) 
=i(e_{pq}{+}e_{qp})$ and
$\iota(f^\dagger_p f^{\phantom\dagger}_q {-} f^\dagger_q f_p)=e_{pq}{-}e_{qp}$,
where $e_{pq}$ denotes a matrix with entries 
$[e_{pq}]_{uv}:=\delta_{pu} \delta_{qv}$. Note that the canonical anticommutation relations imply that
$[f^{\dagger}_pf^{\phantom\dagger}_q,f^{\dagger}_rf^{\phantom\dagger}_s]=$
\begin{equation*} \label{eq:part_conv_comm}
\delta_{ps}f^{\phantom\dagger}_qf^{\dagger}_r {+}\delta_{qr}f^{\dagger}_pf^{\phantom\dagger}_s 
{-}\delta_{ps}\delta_{qr}\unity {-} \delta_{ps}\delta_{qr}\delta_{pq}(f^{\phantom\dagger}_qf^{\dagger}_r
{+}f^{\dagger}_pf^{\phantom\dagger}_s{-}\tfrac{\unity}{2}). 
\end{equation*} 
Thus, $\iota$ is a homomorphism since
\begin{align*}
&\iota([ \kappa_{\pm}(f^\dagger_p f^{\phantom\dagger}_q {\pm} f^\dagger_q f^{\phantom\dagger}_p 
{-} \delta_{pq}\tfrac{\unity}{2}), \kappa_{\pm}(f^\dagger_r f^{\phantom\dagger}_s {\pm} f^\dagger_r 
f^{\phantom\dagger}_s {-} \delta_{rs}\tfrac{\unity}{2})])\\
&=\kappa_{\pm}^2\iota([\delta_{qr}(f^\dagger_{p}f^{\phantom\dagger}_s{\mp}f^\dagger_{s}
f^{\phantom\dagger}_p) \pm \delta_{pr}(f^\dagger_{q}f^{\phantom\dagger}_s{\mp}f^\dagger_{s}
f^{\phantom\dagger}_q), \\
& \phantom{=\kappa_{\pm}^2\iota([\hspace{1mm}}  \delta_{ps}(f^\dagger_{q}
f^{\phantom\dagger}_r{\mp}f^\dagger_{r}f^{\phantom\dagger}_q)
\pm \delta_{qs}(f^\dagger_{p}f^{\phantom\dagger}_r{\mp}f^\dagger_{r}f^{\phantom\dagger}_p)])\\
&=[\kappa_{\pm}(e_{pq}{\pm}e_{qp}), \kappa_{\pm}(e_{rs}{\pm}e_{sr})]\\
&=[\iota( \kappa_{\pm}(f^\dagger_p f^{\phantom\dagger}_q {\pm} f^\dagger_q f^{\phantom\dagger}_p 
{-} \delta_{pq}\tfrac{\unity}{2})),\iota(\kappa_{\pm}(f^\dagger_r f^{\phantom\dagger}_s {\pm} 
f^\dagger_r f^{\phantom\dagger}_s {-} \delta_{rs}\tfrac{\unity}{2}))],
\end{align*}
where $\kappa_+=i$ and $\kappa_-=1$.
The map $\iota$ is even an isomorphism as its kernel
is trivial. The proposition follows as
the Lie algebra $\uu(d)$ is 
isomorphic to the Lie algebra of 
$d \times d$ skew-hermitian matrices.
\end{proof}
\begin{remark}\label{re:iota}
Obviously, the $\iota$ map from the previous 
proof establishes an isomorphism $ih^{(k)} \mapsto 
i\sum_{p,q=1}^d A^{(k)}_{pq}(f^{\dagger}_ pf^{\phantom\dagger}_q-\delta_{pq}\tfrac{\unity}{2})$ 
from $\langle ih^{(1)},  \ldots , ih^{(\ell)} \rangle_{\mathrm{Lie}}$
to $\langle iA^{(1)},  \ldots, iA^{(\ell)}
\rangle_{\mathrm{Lie}}$ for any set  $\{A^{(1)},\ldots, A^{(\ell)}\}$ of
$d{\times}d$ Hermitian matrices.
\end{remark}

\subsection{Quasifree Pure-State Controllability in the Particle-Number Conserving 
Setting \label{subsec:part-numb-pure-quasifree}}

We presented  in Section\ref{sec:qf_pure_cont} a necessary and
sufficient condition for quasifree pure-state 
controllability. Here, we provide an analogous result in the
particle-number conserving setting using a Lie-theoretic result of Ref.~\cite{KDH12}.

A quasifree state $\rho_F$ is called particle-number conserving if
$[\rho_F, P_n]=0$ holds for all $n\in\{0, \ldots , d\}$. 
As discussed in Section~\ref{subsec:quasifree_state}, quasifree states are uniquely characterized 
by the expectation values of the $m_x m_y$ operators. We obtain in the number-conserving case  that 
$\tr( \rho_F \, f^{\phantom\dagger}_q f^{\phantom\dagger}_p )=0$ 
as the condition $[\rho_F, P_n]=0$ 
implies $\sum_{n=0}^{d} P_n \rho_F P_n =\rho_F$ as well as $\tr(\rho_F  f_p f_q)
=
\sum_{n=0}^{d}\tr(P_n \rho_F P_n f_p f_q )
=\sum_{n=0}^d\tr(\rho_F P_nf_p f_q P_n)=0$. Similarly, one can prove 
$\tr( \rho_F \, f^{\dagger}_p f^{\dagger}_q )=0$.
It follows that $\rho_F$ is uniquely determined by the
$d{\times}d$ Hermitian matrix $M_{p,q}= \tr(\rho_F \, f^{\dagger}_p f^{\phantom\dagger}_q )$.
In the literature, this matrix is usually called the \emph{one-particle density matrix} of 
$\rho_F$  \footnote{Note that in some papers  
the one-particle density matrix is defined as $M/\tr(M)$.}.
Let us shortly summarize three well-known statements about one-particle 
density matrices of quasifree states (see \cite{Bach94, AL97}):
\begin{proposition}  \label{prop:pc_qf}
Consider a particle-number conserving quasifree state $\rho_F$ of a 
fermionic system, and let $M$ denote its
one-particle density matrix. The following statements hold:
(a) The eigenvalues of $M$ lie between $0$ and $1$. 
(b) $\rho_F$ is  pure iff $M$ is a projection.
(c) If $\rho_F$ is pure, then $\tr(M)=n$ is an integer, and
$\rho_F$ is supported on the $n$-particle subspace $\wedge^n \C^d$  of the Fock space, i.e.\
\begin{equation} \label{eq:fix_part_numb}
P_k \rho_F P_k = 
\begin{cases} 
\rho_F & \text{if }\; k =n,\\[2mm]
0 & \text{if }\;  k\ne n.
\end{cases} 
\end{equation}
\end{proposition}

The dynamics of particle-number conserving quasifree fermions can also be 
represented using the one-particle density matrices (see \cite{Bach94, AL97}):
\begin{proposition} \label{prop:dyn_numb_cons}
Consider a  particle-number conserving quasifree state $\rho_{a}$ corresponding to 
the one-particle density matrix $M_a$. Assume that the
quadratic Hamiltonian
\begin{equation*}
H=\sum_{p,q=1}^d A_{pq}(f^{\dagger}_ pf^{\phantom\dagger}_q-\delta_{pq}\tfrac{\unity}{2}),
\end{equation*}
which is defined by the Hermitian matrix $A$, generates the time-evolution of $\rho_{a}$.
The time-evolved state  (at unit time),
$\rho_{b}= e^{-iH} \rho_{a} e^{iH}$ is again a number-conserving quasifree state
with a one-particle density matrix
$M_b= U_AM_aU^\dagger_A$, where $U_A=e^{-iA} \in \U(d)$.
\end{proposition}

A particle-number conserving pure quasifree state $\rho_F$ with $\tr(M)=n$
is sometimes called an \emph{$n$-particle pure quasifree state}, since
according to Proposition~\ref{prop:pc_qf} its state is supported on the
$n$-particle subspace $\wedge^{n}\C^d$. We will denote the set of 
such quasifree pure states by $\QFn$.
A system of number-conserving quadratic Hamiltonians $\mathcal{S}=\{ih_1, \ldots , ih_\ell\}$ is said to 
provide quasifree pure-state controllability for a fixed particle number $n$
if  there exists an $iH \in \langle \mathcal{S} \rangle_{\mathrm{Lie}}$ for any $\rho_a, \rho_b \in \QFn$ such that
$\rho_b= e^{-iH} \rho_a e^{iH}$.
To find a necessary and sufficient conditions for this type of
controllability, let us invoke a Theorem~4.1 of Ref.~\cite{KDH12}:
\begin{theorem} \label{thm:symp}
Consider  the Lie algebra $s_\Sigma$ generated by the traceless 
$d{\times}d$ skew-Hermitian matrices $iB_1,  \ldots , iB_\ell$ and 
let  $\Pdn$ denote the set of all projections acting on $\C^d$ whose rank $n$ lies between
$1$ and $d-1$. The Lie group corresponding to $s_\Sigma$ acts naturally via the adjoint action 
on $\Pdn$ . This action is transitive if and only if  either\\
(a) $s_\Sigma$ is isomorphic to $\su(d)$
or (b) $d$ is even, $n\in \{1, d{-}1 \}$, and $s_\Sigma$ is isomorphic to  $\spp(d/2)$.
\end{theorem}

The theorem implies the following
 necessary and sufficient condition:
\begin{theorem}\label{thm:qf_pc_pure_cont}
Consider the set $\mathcal{S}=\{ih_1,  \ldots , ih_\ell\}$ corresponding to
number-conserving quadratic Hamiltonians of a fermionic system with $d \ge 2$ modes. 
The set $\mathcal{S}$ generates a particle-number conserving system giving rise to
full quasifree pure-state controllability on the $n$-particle subspace with $1 \le n \le d{-}1$,
 iff either
(a) $d$ is odd and $\langle \mathcal{S}\rangle_{\mathrm{Lie}}$ is isomorphic to $\uu(d)$ or $\su(d)$\
or (b) $d$ is even, $n \in \{1,  d{-}1 \}$ and  $\langle \mathcal{S} \rangle_{\mathrm{Lie}}$  is isomorphic to
$\uu(d)$, $\su(d)$, $\uu(1) + \spp(d/2) $, or $\spp(d/2)$.
\end{theorem}

\begin{proof}
We consider the set $\mathcal{A} = \{ iA^{(1)}, iA^{(2)}, \ldots , iA^{(\ell)}\}$ of skew-Hermitian 
matrices which correspond to
the generators in $\mathcal{S}$, i.e.\ $ih_{k}=i\sum_{p,q=1}^d A^{(k)}_{pq}(f^{\dagger}_ pf^{\phantom\dagger}_q
-\delta_{pq}\unity/2)$. We apply Remark~\ref{re:iota} and
obtain that $\langle \mathcal{S} \rangle_{\mathrm{Lie}}$ is isomorphic to 
$\langle  \mathcal{A} \rangle_{\mathrm{Lie}}$. We
combine this result with Propositions~\ref{prop:pc_qf} and  \ref{prop:dyn_numb_cons}:
There exists an $ih_{ab} \in \langle \mathcal{S} \rangle_{\mathrm{Lie}} $ 
for each pair $\rho_a, \rho_b \in \QFn$ such that $e^{-ih_{ab}} \rho_a e^{ih_{ab}}=\rho_b$, iff there exists 
an $iA_{ab} \in  \langle  \mathcal{A} \rangle_{\mathrm{Lie}}$  for each pair $M_a, M_b \in \Pdn$
such that $e^{-iA_{ab}} M_a e^{iA_{ab}}= M_b$.
Thus we have to find necessary and sufficient
conditions under which  $\langle \mathcal{A}\rangle_{\mathrm{Lie}}$ 
generates  a transitive action on $\Pdn$ for a given $d$ and $n$.
For any skew-Hermitian $iA$ and $M\in \Pdn$, we have that $\exp(-iA) M \exp(iA)= 
\exp[-i(A- \tr(A) \unity/d)] M \exp[i(A- \tr(A)\unity/d)]$. Hence we can infer that  
 $\langle \mathcal{A} \rangle_{\mathrm{Lie}}$  generates a transitive action iff the system algebra
generated by the set 
$\mathcal{A}':=\{ i(A^{(1)}-\tr(A^{(1)})\unity/d), \ldots, i(A^{(\ell)} - \tr(A^{(\ell)})\unity/d)\}$
 also gives rise to a transitive action. 
Since $\mathcal{A}'$ contains only traceless skew-Hermitian operators, we know from Theorem~\ref{thm:symp} 
that it can act transitively on $\Pdn$ if and only if either
$\langle \mathcal{A}' \rangle_{\mathrm{Lie}}$ is isomorphic to $\su(d)$, or $d$ is even,
$n \in \{1, d{-}1 \}$, and $\langle \mathcal{A}' \rangle_{\mathrm{Lie}}$ is isomorphic to 
$\spp(d/2)$.

On the other hand, if  $\langle\mathcal{A}' \rangle_{\mathrm{Lie}}=\su(d)$ or
$\langle\mathcal{A}' \rangle_{\mathrm{Lie}}=\spp(d/2)$ then $\langle\mathcal{A}' \rangle_{\mathrm{Lie}}$
is a simple irreducible Lie subalgebra 
of $\su(d)$. It follows that
$\langle\mathcal{A} \rangle_{\mathrm{Lie}}$ is either isomorphic to $\langle\mathcal{A}' \rangle_{\mathrm{Lie}}$
if $\tr(A^{(k)})=0$ for all $k\in\{1, \ldots, \ell\}$ or to $\uu(1) + \langle\mathcal{A}' \rangle_{\mathrm{Lie}}$
if there exists a $k$ such that $\tr(A^{(k)})\ne 0$. This proves the theorem.
\end{proof}

\section{Conclusion\label{conclusion}}

We have put dynamic systems theory of coherently controlled fermions into a Lie-algebraic frame 
in order to answer problems of
controllability, reachability, and simulability in a unified picture. 
As summarized in Tab.~\ref{tab:summary},
to this end we have determined the dynamic system Lie algebras
in a comprehensive number of cases, illustrated by examples, with and without confinement
to quadratic interactions (quasifree particles) as well as with and without symmetries
such as translation invariance, twisted reflection symmetry, or particle-number conservation.
Once having established the system algebras, the group orbits of a given (pure or mixed) 
initial quantum state determine the respective reachable sets of all states a system can be driven into 
by coherent control.  In this respect, different types
of pure-state reachability and 
its relation to coset spaces
has been treated with particular attention. 

There are illuminating analogies and differences
between spin and fermionic systems. 
For quasifree systems, this was discussed in Sec.~\ref{QUASI} 
and in  Appendix~\ref{appl_quasi}, while the translation-invariant case
is addressed in Sec.~\ref{sec:TI}.
In particular, translation-invariant Hamiltonians 
which cannot be generated from nearest-neighbor ones
appear both for spin systems (Sec.~\ref{shortrange}) and for fermionic systems (Sec.~\ref{ferm_nn}).
Moreover, for fermionic systems some of these Hamiltonians have bounded interaction length.
It is an open question if the same also holds for spin systems.

On a general scale, the system algebras determined serve 
as a dynamic fingerprint. 
Their application to quantum simulation has been elucidated in a plethora of paradigmatic
settings.
Hence we anticipate the comprehensive findings presented here will find a broad scope
of use.

\begin{table}[tb]
\caption{\label{tab:summary} System algebras for $d$-mode fermionic systems}
\begin{tabular}{p{1mm}l@{\hspace{-4mm}}r@{\hspace{2mm}}r}
\hline\hline\\[-2mm]
&Symmetries\footnote{besides parity superselection rule $\mathrm{P}$: $\mathrm{T}=$ 
translation-invariance,
\mbox{$\mathrm{R}=$ twisted reflection symmetry, $\mathrm{N}=$ particle-no.\ conservation}
\vspace{1.5mm}}
& System algebra  & Details\\[1mm] \hline \\[-2mm]
\multicolumn{2}{l}{general systems:}\hspace{-8mm}& $\su(2^{d-1})\oplus \su(2^{d-1})$ &  
Thm.~\phantom{0}\ref{general} \\[1mm]
&$\{ \mathrm{T} \}$ &$\mathfrak{s}[\oplus_{\ell=0}^{d-1}\, \uu (\TM_\ell)]\oplus 
\mathfrak{s}[\oplus_{\ell=0}^{d-1}\, \uu (\TM_\ell)]$
 & Thm.~\ref{translationinvariant}\\[1.5mm]
&$\{ \mathrm{N} \}$ &$\mathfrak{s}\left(\oplus_{n \, \text{even}}\; 
\uu\hspace{-0.75mm}\left[\tbinom{d}{n}\right]\right) \oplus\mathfrak{s}\left(\oplus_{n \, \text{odd}}\; 
\uu\hspace{-0.75mm}\left[\tbinom{d}{n}\right]\right)$ & 
Prop.~\ref{prop:particle_number}\\[2mm]\hline\\[-1mm]
\multicolumn{2}{l}{quasifree systems:}\hspace{-8mm}& $\so(2d)$\footnote{the orthogonal algebra is 
represented as direct sum of two equal copies given as irreducible blocks of dimension $2^{d-1}$; 
the system algebra $\so(2d)$ itself was 
determined already, e.g., in Ref.~\cite{SW86}.} & Prop.~\phantom{0}\ref{quasifree} \\[1mm]
&$\{  \mathrm{T}\}, d \text{ odd} $ & $[\sum_{i=1}^{(d-1)/2} \uu(2)] + \uu(1)$ & 
Thm.~\ref{thm:qf-trans-inv} \\[1mm]
&$\{  \mathrm{T}\}, d \text{ even}  $ & $[\sum_{i=1}^{(d-2)/2} \uu(2)] + \uu(1) + \uu(1)$ & 
Thm.~\ref{thm:qf-trans-inv} \\[1mm]
&$\{  \mathrm{T}, \mathrm{R}\}, d \text{ odd} $ & $[\sum_{i=1}^{(d-1)/2} \su(2)] + \uu(1)$ & 
Cor.~\ref{Rsym} \\[1mm]
&$\{  \mathrm{T}, \mathrm{R}\}, d \text{ even}  $ & $[\sum_{i=1}^{(d-2)/2} \su(2)] + \uu(1) + 
\uu(1)$ & Cor.~\ref{Rsym} \\[1mm]
&$\{  \mathrm{N} \}$ & $\uu(d)$ & Prop.~\ref{quasifree_particle}\\[1mm]
\hline\hline
\end{tabular}
\end{table}

\begin{acknowledgments}
This work was supported in part by the  {\sc eu}
through the programs {\sc coquit},
{\sc q-essence}, {\sc chist-era quasar}, {\sc siqs} and the {\sc erc} grant
{\sc gedentqopt},  by the Bavarian Excellence Network {\sc enb}
via the international doctorate programme of excellence
{\em Quantum Computing, Control, and Communication} ({\sc qccc}),
by {\em Deutsche Forschungsgemeinschaft} ({\sc dfg}) in the
collaborative research centre {\sc sfb}~631 as well as the international 
research group {\sc for} 1482 through the grant {\sc schu}~1374/2-1.
\end{acknowledgments}

\appendix

\section{Discussion of Double Centralizers\label{DoubleCentralizers}}

Motivated by Sec.~\ref{ExAndDisc}, in this appendix  we discuss 
how the form of the double centralizer of a Lie algebra $\fg\subset \su(k)$
limits the possibilities for $\fg$:
\begin{proposition}\label{Thm:DCT2}
Let $\fg$ denote a subalgebra of $\su(k)$. 
There exists a set $A \subset  \su(k)$ such that
$\fg= \cent_{\su(k)}(A)$, if and only if 
$\cent_{\su(k)}(\cent_{\su(k)}(\fg))=\fg$.
\end{proposition}
\begin{proof}
First, let us assume the existence of the set $A$.
As $\cent_{\su(k)}(\cent_{\su(k)}(\cent_{\su(k)}(A)))= \cent_{\su(k)}(A)$ 
holds for any set $A$, which can also be inferred from \cite[Proposition~6.1.3.1(iii)]{Procesi07},
we obtain $\cent_{\su(k)}(\cent_{\su(k)}(\fg))= \fg$.
Second, we assume that $\cent_{\su(k)}(\cent_{\su(k)}(\fg))=\fg$ holds.
We choose $A:=\cent_{\su(k)}(\fg)$ and verify its existence.  
\end{proof}

To further analyze the influence of symmetry properties
on the system algebra, we recall some elementary representation theory
(see, e.g., Theorem~1.5 of \cite{Ledermann}):

\begin{proposition}\label{reprtheory}
Consider a completely reducible complex matrix representation 
$\Phi(g)$
of a group
$G$,
where $k$ is the degree of $\Phi$.
Let $\comm(\Phi)=\Phi'$ denote
the commutant algebra of all complex $k \times k$-matrices 
simultaneously
commuting with $\Phi(g)$ for $g\in G$.
Then, 
$\Phi(g)$ is equivalent to
$$
\bigoplus_{j=1}^{w} \left[ \unity_{e_j} \otimes \phi_j(g) \right],
$$
where $\phi_j$ denote for $j\in\{1,\ldots,w\}$ distinct inequivalent irreducible complex 
matrix representations 
of $G$
with degree $k_j$, occurring with multiplicity $e_j$ in $\Phi$.
In particular,\\
(a) $\dim \comm(\Phi)=\sum_{j=1}^{w} e_j^2$,\\
(b) $\dim \centre(\comm(\Phi))=w$,\\
(c) $k= \sum_{j=1}^{w} k_j e_j$.
\end{proposition}

Obviously, the same is true for representations of a compact Lie group
or its Lie algebra. Given a subalgebra $\fg$ of $\su(k)$ (or respectively of $\uu(k)$) and a 
representation $\Phi$
of $\fg$ with degree $k$, we discuss the easiest case of Proposition~\ref{reprtheory} where
$w=1$ and $e_1=1$. Hence, $\Phi$ is irreducible and $\fg$ is an irreducible subalgebra
of $\su(k)$ (or respectively of $\uu(k)$). But $\fg$ is not necessarily equal to $\su(k)$ 
(or respectively to $\uu(k)$).
Irreducible simple subalgebras of $\su(k)$ were studied extensively in this regard in Ref.~\cite{ZS11}.
Note that the irreducible subalgebras of $\uu(k)$ are of the form
$\fg$ or $\fg+\uu(1)$ where $\fg$ denotes any irreducible subalgebra of $\su(k)$
(cf.\ pp.~27--28 and p.~321 of \cite{GG78}). --- A slight generalization is given by the case
of an abelian commutant algebra, i.e.\ $\dim \comm(\Phi)=\dim \centre(\comm(\Phi))$ and $e_j=1$ 
for all $j\in\{1,\ldots,w\}$. One may thus apply the spectral theorem 
(see, e.g., \cite{Hass99,Rom92,DeVito90})
simultaneously to all the elements of the commutant algebra:

\begin{theorem}\label{thm:double-centralizer}
Consider a Lie algebra $\fg \subseteq \su(k)$ and its
representation $\Phi$ of degree $k$.
Assume that the corresponding commutant algebra
$\mathcal{C}=\comm(\Phi)$ is abelian. One obtains that $\fg$ is a subalgebra
of $\mathfrak{s}[\oplus_{j=1}^{\dim \mathcal{C}} \uu(k_j)]$ and it is equivalent
to $\mathfrak{s}[\oplus_{j=1}^{\dim \mathcal{C}} \fg_j]$,
where $k=\sum_{j=1}^{\dim \mathcal{C}} k_j$ and $\fg_j$ are irreducible subalgebras of $ \uu(k_j)$.
Furthermore, one finds
$k_j=\dim(P_j)$, where $P_j$ are the orthogonal projection operators given by 
the joint spectral decomposition of $\mathcal{C}$ with $\sum_{j=1}^{\dim \mathcal{C}} P_j=\unity_k$ 
and $P_i P_j = 0$ for $i \ne j$.
If $\fg$ is the maximal Lie algebra with these properties, then
$\fg=\mathfrak{s}[\oplus_{j=1}^{\dim \mathcal{C}} \uu(k_j)]$.
\end{theorem}

Using Proposition~\ref{reprtheory} one can
directly characterize a maximal Lie algebra $\fg$ contained in $\su(k)$
which is defined by all its symmetries including cases where
the commutant to $\fg$  is not necessarily abelian.
Observe the notation of Remark~\ref{remark5} and the one of Proposition~\ref{reprtheory}.
\begin{theorem}\label{maxSymLie}
Consider a Lie algebra $\fg \subseteq \su(k)$ and its
representation $\Phi$ of degree $k$.
Let $\mathcal{C}=\comm(\Phi)$ denote the commutant of $\fg$.
If $\fg$ is the maximal Lie algebra with these properties, then
$\fg=\mathfrak{s}[\sum_{j=1}^{\omega} \uu(k_j)]$ where $\omega=\dim[\centre(\mathcal{C})]$
and $\sum_{j=1}^{\omega} k_j\leq k$.
\end{theorem}
\begin{proof}
Using Proposition~\ref{reprtheory} (and its notation) one obtains that $\fg$ is equivalent to 
$\oplus_{j=1}^{w} \left[ \unity_{e_j} \otimes \phi_j(g) \right]$.
Therefore, $\fg$ is a subalgebra
of $\mathfrak{s}[\sum_{j=1}^{\omega} \uu(k_j)]$ with $\sum_{j=1}^{\omega} k_j\leq k$.
The maximality of $\fg$ completes the proof.

In a dual approach, one could start from a set $S$ of symmetries
of $\fg$. Due to the maximality of $\fg$, the set $S$  has to comprise {\em all 
symmetries of} $\fg$. Next, one can apply Proposition~\ref{reprtheory} to the 
subalgebra of $\su(k)$ generated by the linear span intersected with $\su(k)$, 
i.e.~$\langle S\rangle \cap \su(k)$. 
The theorem then follows directly
using Schur's lemma and the maximality of $\fg$.
\end{proof}

The reader familiar with the double-commutant theorem in algebraic quantum mechanics will
wonder about the different power of symmetries for characterizing algebras of observables on the
one hand and Lie algebras on the other:
a von-Neumann algebra $\mathcal{A}$ is entirely determined by its commutant 
$\mathcal{A}'$, since $\mathcal{A}''=\mathcal{A}$ \cite{Dix81,Sak71}. 
In this sense, there is a duality between the algebra $\mathcal{A}$ and its
commutant $\mathcal{A'}$ encapsulating all \emph{symmetries}.
On the other hand,
consider the illustrative case of an irreducible Lie subalgebra $\fg$ of $\su(k)$
\footnote{Note that irreducible subalgebras of $\su(k)$ are semisimple (or even simple).},
where
the centralizer $\cent_{\su(k)}(\fg)$ is trivial, i.e.\ zero. This centralizer
is shared with \emph{all} irreducible Lie subalgebras of $\su(k)$.
So in turn, the double centralizer in $\su(k)$ to all these subalgebras is $\su(k)$ itself. 
We thus obtain the following
corollary to Proposition~\ref{Thm:DCT2} and Theorem~\ref{thm:double-centralizer},
where the double centralizer gives a maximality criterion ensuring
that an irreducible subalgebra $\fg$ of $\su(k)$ is in fact fulfilling $\fg=\su(k)$ 
\footnote{Note that the condition $\cent_{\su(k)}(\{0\})=\fg$  is not easily tested using
only a set of generators of $\fg$.}:

\begin{repcorollary}{cor_double}
Let $\fg$ denote an irreducible subalgebra of $\su(k)$, i.e.\ 
$\cent_{\su(k)}(\fg)=\{0\}$. Then one finds that
$\cent_{\su(k)}(\cent_{\su(k)}(\fg))=\fg$ if and only if $\fg=\su(k)$
\end{repcorollary}

Note that Corollary~\ref{cor_double} can be readily generalized: Let $\fg,\fh$ denote 
two irreducible subalgebras of $\su(k)$ with $\fg\subseteq\fh\subseteq\su(k)$ so that
$\cent_{\fh}(\fg)=\{0\}=\cent_{\fh}(\fh)$. Then one finds 
$\cent_{\fh}(\cent_{\fh}(\fg))=\fg$ if and only if $\fg=\fh$.

Summarizing the general case, the symmetry properties of a Lie algebra $\fg\subseteq \su(k)$,
as given by its commutant w.r.t.\ a representation of $\fg$,
do \emph{not} determine the Lie algebra $\fg$ uniquely. Yet the commutant allows us to infer a 
\emph{unique maximal Lie algebra}
contained in $\su(k)$, which is (up to an identity matrix) equal to the double commutant of $\fg$,
but in general not to $\fg$ itself. Although all representations of compact Lie algebras, such as $\su(k)$ 
and its semisimple subalgebras, 
are completely reducible, the situation for Lie algebras also
differs from the case of \emph{associative algebras}: here complete reducibility of a representation 
implies the double-commutant theorem (see Theorem~(3.5.D) of \cite{Weyl} or Theorem~4.1.13 of 
\cite{GoodmanWallach09}),
whereas the double-commutant theorem does not apply to Lie algebras as discussed above.

\section{Parameterizations of Quadratic Hamiltonians\label{ParaQuadratic}}

In this appendix, we discuss various parameterizations of quadratic Hamiltonians
related to the one of Eq.~\eqref{Hqfree} in Sec.~\ref{QUASI}.
We start with the parametrization
\begin{align*}
H := \sum_{p,q=1}^d C_{pq} f_p f_q + D_{pq} f_p f^{\dagger}_q 
+E_{pq} f^{\dagger}_p f_q + F_{pq} f^{\dagger}_p f^{\dagger}_q
\end{align*}
by complex $d\times d$-matrices $C$, $D$, $E$, and $F$.
Hermiticity of $H$ requires 
$C=F^{\dagger}$, $D=D^{\dagger}$, and  $E=E^{\dagger}$, while
the (anti-)commutator relations enforce
$C=-C^{t}$, $D=-E^{t}$, and $F=-F^{t}$. 
Setting $A:=2E$ and $B:=-2 C^{*}$, we recover the notation of
Eq.~\eqref{Hqfree} and obtain
\begin{align*}
&H = \frac{1}{2} \sum_{p,q=1}^d -B^{*}_{pq} f_p f_q - A^{*}_{pq} f_p f^{\dagger}_q 
+ A_{pq} f^{\dagger}_p f_q +  B_{pq} f^{\dagger}_p f^{\dagger}_q\\
&\phantom{H} = \frac{1}{2} \sum_{p,q=1}^d - B^{*}_{pq} f_p f_q + 2 A_{pq} (f^{\dagger}_p f_q 
{-} \delta_{pq} \tfrac{\unity}{2}) + B_{pq} f^{\dagger}_p f^{\dagger}_q\\
&=\frac{1}{2} \sum_{p,q=1}^d
\Re(B_{pq}) \left( f^{\dagger}_p f^{\dagger}_q {-} f_p f_q \right)
+\Re(A_{pq}) \left( f^{\dagger}_p f_q {-} f_p f^{\dagger}_q \right)\\
&\phantom{=}
+\Im(B_{pq})\, i\left( f^{\dagger}_p f^{\dagger}_q {+} f_p f_q \right)
+ \Im(A_{pq})\, i\left( f^{\dagger}_p f_q {+} f_p f^{\dagger}_q \right).
\end{align*}
Note
$\Re(A)=\Re(A)^t$, $\Im(A)=-\Im(A)^t$,
$\Re(B)=-\Re(B)^t$, and $\Im(B)=-\Im(B)^t$
which is a consequence of $A=A^{\dagger}$ and $B=-B^t$.
We rewrite the Hamiltonian using Majorana operators
such that
\begin{equation*}
-i H = -\frac{1}{2} \left[ \sum_{p=1}^d -\Re(A_{pp})\, m_{2p-1} m_{2p} + \sum_{p,q=1;p>q}^d V_{pq}\right],
\end{equation*}
where
\begin{align*}
V_{pq}=&
- 
\Re(B_{pq}) \left(m_{2p-1} m_{2q} - m_{2q-1} m_{2p} \right)\\
&
- 
\Re(A_{pq}) \left(m_{2p-1} m_{2q} + m_{2q-1} m_{2p} \right)\\
&
-
\Im(B_{pq}) \left(m_{2p-1} m_{2q-1} - m_{2p} m_{2q} \right)\\
&
-
\Im(A_{pq}) \left(m_{2p-1} m_{2q-1} + m_{2p} m_{2q} \right).
\end{align*}
By applying the Jordan-Wigner transformation we obtain the Hamiltonian for 
the corresponding spin system 
(for better readability, the tensor-product symbol is omitted, e.g.,
$\mathrm{I}\mathrm{X}\mathrm{Y}:=\mathrm{I}\otimes\mathrm{X}\otimes\mathrm{Y}$):
\begin{equation*}
-iH=-\frac{i}{2}\Big[\sum_{p=1}^d -\Re(A_{pp}) \underbrace{\mathrm{I}\cdot\cdot\,\mathrm{I}}_{p-1}
\mathrm{Z}\underbrace{\mathrm{I}\cdot\cdot\,\mathrm{I}}_{d-p}
+ \sum_{p,q=1;p>q}^d W_{pq}\Big]
\end{equation*}
with $W_{pq}=$
\begin{align*}
&+
\Re(B_{pq}) \big(
\underbrace{\mathrm{I}\cdot\cdot\,\mathrm{I}}_{q-1}
\mathrm{X} \underbrace{\mathrm{Z}\cdot\cdot\,\mathrm{Z}}_{p-q-1}
\mathrm{X} \underbrace{\mathrm{I}\cdot\cdot\,\mathrm{I}}_{d-p}
-
\underbrace{\mathrm{I}\cdot\cdot\,\mathrm{I}}_{q-1}
\mathrm{Y} \underbrace{\mathrm{Z}\cdot\cdot\,\mathrm{Z}}_{p-q-1}
\mathrm{Y} \underbrace{\mathrm{I}\cdot\cdot\,\mathrm{I}}_{d-p}
\big) \\
&+
\Re(A_{pq})
\big(
\underbrace{\mathrm{I}\cdot\cdot\,\mathrm{I}}_{q-1}
\mathrm{X} \underbrace{\mathrm{Z}\cdot\cdot\,\mathrm{Z}}_{p-q-1}
\mathrm{X} \underbrace{\mathrm{I}\cdot\cdot\,\mathrm{I}}_{d-p}
+
\underbrace{\mathrm{I}\cdot\cdot\,\mathrm{I}}_{q-1}
\mathrm{Y} \underbrace{\mathrm{Z}\cdot\cdot\,\mathrm{Z}}_{p-q-1}
\mathrm{Y} \underbrace{\mathrm{I}\cdot\cdot\,\mathrm{I}}_{d-p}
\big)
\\
&
-
\Im(B_{pq})
\big(
\underbrace{\mathrm{I}\cdot\cdot\,\mathrm{I}}_{q-1}
\mathrm{Y} \underbrace{\mathrm{Z}\cdot\cdot\,\mathrm{Z}}_{p-q-1}
\mathrm{X} \underbrace{\mathrm{I}\cdot\cdot\,\mathrm{I}}_{d-p}
+
\underbrace{\mathrm{I}\cdot\cdot\,\mathrm{I}}_{q-1}
\mathrm{X} \underbrace{\mathrm{Z}\cdot\cdot\,\mathrm{Z}}_{p-q-1}
\mathrm{Y} \underbrace{\mathrm{I}\cdot\cdot\,\mathrm{I}}_{d-p}
\big)\\
&
+
\Im(A_{pq})
\big(
\underbrace{\mathrm{I}\cdot\cdot\,\mathrm{I}}_{q-1}
\mathrm{Y} \underbrace{\mathrm{Z}\cdot\cdot\,\mathrm{Z}}_{p-q-1}
\mathrm{X} \underbrace{\mathrm{I}\cdot\cdot\,\mathrm{I}}_{d-p}
-
\underbrace{\mathrm{I}\cdot\cdot\,\mathrm{I}}_{q-1}
\mathrm{X} \underbrace{\mathrm{Z}\cdot\cdot\,\mathrm{Z}}_{p-q-1}
\mathrm{Y} \underbrace{\mathrm{I}\cdot\cdot\,\mathrm{I}}_{d-p}
\big).
\end{align*}

\section{Applications of Quasifree Fermions to Spin Systems\label{appl_quasi}}
Here we take new fermionic approaches to exhaustively prove and improve some 
results of Ref.~\cite{ZS11}, where some proofs were still sketchy---thereby also filling a 
desideratum voiced in \cite{WBS12}.

\subsection{A Spin System with System Algebra $\so(2n+1)$}

\begin{proposition}[see Proposition~27 in \cite{ZS11}] \label{prop:so2n+1}
Consider a Heisenberg-{\XX} chain with the drift Hamiltonian 
$$H_d=\mathrm{XX\cdot\cdot\,II}+\mathrm{YY\cdot\cdot\,II}+\cdots+\mathrm{II\cdot\cdot\,XX}
+\mathrm{II\cdot\cdot\,YY}$$
on $n$ spin-$\tfrac{1}{2}$ qubits with $n\geq 2$. Assume that one end qubit is individually 
locally controllable. 
The system algebra is given as the subalgebra $\so(2n+1)$
which is irreducibly embedded in $\su(2^n)$.
\end{proposition}

\begin{proof}
We use the fermionic picture where the number of modes $d$ equals the number
of spins $n$. The generators are given by $w_1=L(v_1)$ with
$v_1= \sum_{p=1}^{d-1} - m_{2p-1} m_{2p+2} + m_{2p} m_{2p+1}$, $L(m_1)$, and $L(m_2)$.
Obviously, the element $w_2=L(v_2)$ with $v_2 = m_1 m_2$ can also be generated. 
One can verify that exactly all Majorana operators of degree one or two can be obtained:
One line of reasoning uses Lemma~\ref{quad_ex_lem} together with
the commutation relations $[L(m_{2p-1}),(b_p)] = L(m_{2p+2})$ and
$[L(m_{2p}),L(b_p)] = -L(m_{2p+1})$ to show that all degree-one operators
can be generated. This immediately gives all quadratic operators as well, while operators
of higher degree are not attainable. Therefore, the dimension of the system algebra 
is $2d^2+d$. Note that the operators $L(m_{2p-1} m_{2p})$
form a maximal abelian subalgebra $\fa$ which proves that the system algebra has rank $d$.
In the spin picture, we can directly verify that $\fa=
\langle -i\mathrm{Z}_1/2,\ldots,-i\mathrm{Z}_n/2 \rangle_{\rm Lie}$
by computing the centralizer $\fc_\fa:=\langle -\tfrac{i}{2} \prod_{j\in S} \mathrm{Z}_j 
\,|\, \{\} \ne S \subset \{1,\ldots,n\}  \rangle_{\rm Lie}$ of $\fa$ in $\su(2^n)$.
Let us compute the centralizer $\fc_\fb$ of $\fb=\langle m_{2p} m_{2p+1}, -m_{2p-1} m_{2p+2} \,|\, p 
\in \{1,\ldots,d{-}1 \} \rangle_{\rm Lie}$ in $\su(2^n)$. Note that the generators of $\fb$ are given in the 
spin picture by
$-i\mathrm{X}_p \mathrm{X}_{p+1} /2$ and $-i\mathrm{Y}_p \mathrm{Y}_{p+1} /2$.
One can readily show by induction that $\fc_\fb=\langle -\tfrac{i}{2} \prod_{j=1}^n \mathrm{X}_j, 
 -\tfrac{i}{2} \prod_{j=1}^n \mathrm{Y}_j, -\tfrac{i}{2} \prod_{j=1}^n \mathrm{Z}_j \rangle_{\rm Lie}$.
It follows that the centralizer $\fc$ of the full system algebra in $\su(2^n)$ has to be contained in
$\fc_\fa \cap \fc_\fb=\langle -\tfrac{i}{2} \prod_{j=1}^n \mathrm{Z}_j \rangle_{\rm Lie}$.
One can now easily prove that the centralizer of the full system algebra
in $\su(2^n)$ is trivial and  that the 
system algebra is irreducibly embedded  in $\su(2^d)$.
As the coupling graph of the spin 
system is connected, we conclude with Theorem~6 of \cite{ZS11} that the system algebra is
simple. Listing all simple (and compact) Lie algebras with the correct dimension and rank,
we obtain (a) $\so(2d+1)$ for $d\geq 1$, (b) $\spp(d)$ for $d \geq 1$, (c) $\su(2)\cong \so(3)$
for $d=1$, and (d) $\fe_6$ for $d=6$. As the system algebra contains also all quadratic
operators, it has a subalgebra $\so(2d)$ which is of maximal rank.
This rules out the cases (b) and (d) (see p.~219 of \cite{BS49}
or Sec.~8.4 of \cite{GG78}) for $d\neq2$.
But the case (b) agrees with (a) for $d=2$.
For $d=1$, the cases (a) and (c) coincide. This completes the proof.
\end{proof}

Note that with our fermionic approach one can readily determine
the dimension and rank of the system algebra. Likewise, we establish that
all fermionic operators act irreducibly from which we can infer that the system algebra
is simple. The rest of the proof follows by an exhaustive enumeration.---In more general terms, 
as in Theorem~34 and Corollary~35 of \cite{ZS11}, we  
connect a spin system with a fictitious fermionic system:
\begin{corollary}
Consider a fictitious fermionic system with $d$ modes which consists of
all linear and quadratic operators and whose generators can, e.g., be chosen
as all Majorana operators of type $L(m_{2p-1})$ combined with the Hamiltonian
from Eq.~\eqref{Hqfree} where the control functions 
$A_{pq}$ and $B_{pq}$ can be assumed to be real.
This fictitious fermionic system and the spin system of Proposition~\ref{prop:so2n+1}
with $n=d$ spins can simulate each other. In particular, both can simulate
a general quasifree fermionic system with $d$ modes and system algebra $\so(2d)$
as presented in Proposition~\ref{quasifree} and Theorem~\ref{quad_ex_so}.
\end{corollary}

\subsection{A Spin System with System Algebra $\so(2n+2)$}

\begin{proposition}[see Proposition~28 in \cite{ZS11}]\label{prop:so2n+2}
Consider a Heisenberg-{\XX} chain with the drift Hamiltonian 
$$H_d=\mathrm{XX\cdot\cdot\,II}+\mathrm{YY\cdot\cdot\,II}+\cdots+\mathrm{II\cdot\cdot\,XX}
+\mathrm{II\cdot\cdot\,YY}$$ on $n$ spin-$\tfrac{1}{2}$ qubits with $n\geq 2$. Assume that 
each of the two end qubits is individually locally controllable. 
The system algebra is given as the subalgebra $\so(2n+2)$
which is irreducibly embedded in $\su(2^n)$.
\end{proposition}

\begin{proof}
We switch to a fermionic picture where the number of modes $d$ equals the number
of spins $n$. The generators are $w_1=L(v_1)$ with 
$v_1= \sum_{p=1}^{d-1} - m_{2p-1} m_{2p+2} + m_{2p} m_{2p+1}$, $L(m_1)$, 
$L(m_2)$, $L(m_{2d-1} \prod_{p=1}^{d-1} m_{2p-1}m_{2p})$, and 
$L(m_{2d} \prod_{p=1}^{d-1} m_{2p-1}m_{2p})$.
One can verify by explicit computations that exactly all Majorana operators
of degree one, two, $2d-1$, and $2d$ can be generated.
Therefore, the dimension of the system algebra 
is $2d^2+3d+1$. 
Using a similar argument as in the proof of Proposition~\ref{prop:so2n+1}, we conclude
that the operators $L(m_{2p-1} m_{2p})$ 
together with the operator $L(\prod_{p=1}^{d} m_{2p-1}m_p)$
form a maximal abelian subalgebra which proves that the system algebra has rank $d+1$.
One can also show 
that the system algebra
is irreducibly embedded in $\su(2^d)$. As the coupling graph of the spin system
is connected, we conclude with Theorem~6 of \cite{ZS11} that the system algebra is
simple. The proof is completed
by listing all simple (and compact) Lie algebras with the correct dimension and rank.
We obtain (a) $\so(2d+2)$ for $d\geq 1$ and (b) $\su(4)\cong \so(6)$ for $d=2$.
\end{proof}

\emph{Principle Remark:} Now we have established a setting that allows for
exploiting the powerful general results of \cite{MS43} 
on the structure of orthogonal groups that
provide a second avenue to Proposition~\ref{prop:so2n+1}
assuming we have already established Proposition~\ref{prop:so2n+2}:
Lemmata~3 and 4 of \cite{MS43} show
that for $k \geq 3$ any subalgebra of $\so(k)$ with dimension $(k-1)(k-2)/2$ is isomorphic
to $\so(k-1)$; moreover $\so(k-1)$ is a maximal subalgebra of $\so(k)$.
Thus, by proving that the system algebra has dimension $2d^2+d$ with
$d\geq 1$, it can be identified as the subalgebra $\so(2d+1)$ of $\so(2d+2)$.
We emphasize that this particular proof technique should be widely 
applicable in quantum systems theory.

Relying on the proof of Proposition~\ref{prop:so2n+2} and building on 
Theorem~32 as well as Corollary~34 of \cite{ZS11}, we obtain 
connections between a spin system, a quasifree fermionic system, and 
a fictitious fermionic system:
\begin{corollary}
The following control systems all have the system algebra $\so(2k+2)$ and
can simulate each other:
(a) the spin system of Proposition~\ref{prop:so2n+2} with $k$ spins,\\
(b) the quasifree fermionic system with $k+1$ modes 
as presented in Proposition~\ref{quasifree} and Theorem~\ref{quad_ex_so}, and\\
(c) a fictitious fermionic system with $k$ modes which contains all Majorana operators of
degree one, two, $2k{-}1$, and $2k$, and whose
generating Hamiltonian can be chosen from Eq.~\eqref{Hqfree} where the
control functions 
$A_{pq}$ and $B_{pq}$ can be assumed to be real.
\end{corollary}

\section{Proof of Theorem~\ref{quad_ex_u}\label{proof_u_d}}
The cases of $d\in\{2,3,4\}$ can be verified directly and we assume in the following that $d \geq 5$ holds.
We build on Lemma~\ref{quad_ex_lem}
and obtain a basis of $\fk_1$ consisting of 
$L(a_p)$ with $1\leq p \leq d$ as well as $L(b_p^{(i)})$
with
\begin{align*}
b_p^{(i)}&:=-m_{2p-1} m_{2p+2i} + m_{2p} m_{2p+2i-1}
\intertext{and $L(c_p^{(i)})$ with}
c_p^{(i)}&:=m_{2p-1} m_{2p+2i-1} + m_{2p} m_{2p+2i}
\end{align*}
where
$p,i\geq 1$ and $p+i\leq d$.  One can systematically 
enlarge the index $(i)$ starting from the elements
$L(b_p^{(1)})=L(b_p)\in \fk_1$ and $L(c_p^{(1)})=(c_p)\in \fk_1$
and generate all $L(b_p^{(i)})$ and $L(c_p^{(i)})$
by combining the commutator relations
$[L(c_p),L(b_{p+1}^{(i)})]=-L(b_p^{(i+1)})$ and
$[L(c_{p+i}),L(b_{p}^{(i)})]=L(b_p^{(i+1)})$ with the commutator relations 
$[L(a_p),L(b_p^{(i)})]=-L(c_p^{(i)})$ and $[L(a_p),L(c_p^{(i)})]=L(b_p^{(i)})$.
It is straightforward to check that no further elements are generated by
commutators starting from the elements $L(a_p)$, $L(b_p^{(i)})$, and $L(c_p^{(i)})$.
We obtain that $\dim(\fk_1)=d+(d-1)d=d^2$. Furthermore, the elements $L(a_p)$ form
a maximal abelian subalgebra of $\fk_1$ and the rank of $\fk_1$ is $d$.
It follows that $\fk_1$ is a subalgebra of maximal rank 
in $\so(2d)$.

We now show that the center of $\fk_1$ is one-dimensional and is
generated by $L(c)$ with $c:=\sum_{p=1}^d a_p$. Combining the commutator relations 
$[L(a_{p+i}),L(b_p^{(i)})]=L(c_p^{(i)})$
and $[L(a_{p+i}),L(c_p^{(i)})]=-L(b_p^{(i)})$ with the ones for $L(a_{p})$ mentioned above,
we conclude that $[L(a_p{+}a_{p+i}),L(b_p^{(i)})]=[L(a_p{+}a_{p+i}),L(c_p^{(i)})]=0$. 
In addition, we obtain
$[L(a_j),L(b_p^{(i)})]=[L(a_j),L(c_p^{(i)})]=0$ if $p\neq j \neq p+i$.
It follows that $[L(c),L(b_p^{(i)})]=[L(c),L(c_p^{(i)})]=0$ and that $L(c)$ 
commutes with all elements of $\fk_1$. We rule out the existence of further
elements in the center by explicitly computing the semisimple
part $\fs:=[\fk_1,\fk_1]$ of $\fk_1$. By applying $[L(b_p^{(i)}),L(c_p^{(i)})]/2=L(a_{p+i}-a_p)$
combined with previously mentioned commutator relations, we can fix a basis of $\fs$ 
consisting of the elements $L(b_p^{(i)})$, $L(c_p^{(i)})$, and $L(a_p-a_{p+1})$ where 
$1\leq p \leq d-1$ and $1\leq i \leq d-p$. 

We proceed to prove in the following that $\fs$ is actually 
simple by showing that $\fs$ is not abelian (which obviously holds) and that
any non-zero ideal $\fii$ of $\fs$ is equal to $\fs$. 
Starting from $(\adr^2(L(a_q)))\, L(b_p^{(i)})=-L(b_p^{(i)})$
and $(\adr^2(L(a_q)))\, L(c_p^{(i)})=-L(c_p^{(i)})$ for $q=p$ or $q=p+i$, we deduce
that $\adr^2(L(a_q))+\adr^4(L(a_q))=0$.
Likewise, $y_p^{(i)}:=[\adr^2(L(a_p{-}a_{p+i}))+\adr^4(L(a_p{-}a_{p+i}))]/12$
annihilates all basis elements of $\fs$ 
except for $L(b_p^{(i)})$ and $L(c_p^{(i)})$ 
which are left invariant. Using the definition
$x_p^{(i)}:=[\adr^2(L(b_p^{(i)}))-\adr^2(L(c_p^{(i)}))]/4$
and verifying $x_p^{(i)}\, L(a_p) = x_p^{(i)}\, L(a_{p+i}) =0$,
we can infer that $x_p^{(i)}\, L(a_q{-}a_{q+j}) = 0$ holds for all valid $q$ and $j$.
Furthermore, we have $x_p^{(i)}\, L(b_q^{(j)}) = x_p^{(i)}\, L(c_q^{(j)}) =0$
for all valid $q$ and $j$ unless
when both $q=p$ and $j=i$ hold. We obtain 
$x_p^{(i)}\, L(b_p^{(i)}) = L(b_p^{(i)})$
and $x_p^{(i)}\, L(c_p^{(i)}) = -L(c_p^{(i)})$ in this exceptional case.
As $\fs$ is semisimple, $\fii$ cannot be abelian and has to contain an element 
which is supported on some $L(b_p^{(i)})$ or $L(c_p^{(i)})$. 
Relying on the ideal property $[\fs,\fii] \subseteq \fii$
and the operators $x_p^{(i)}$ and $y_p^{(i)}$,
we conclude that $L(b_p^{(i)})\in \fii$ or $L(c_p^{(i)})\in \fii$.
Obviously, the conditions  $L(b_p^{(i)})\in \fii$, $L(c_p^{(i)})\in \fii$,
and $L(a_p{-}a_{p+i})\in \fii$ are equivalent. 
By applying previously mentioned commutator relations, we can verify
that $L(b_q^{(j)})\in \fii$ holds for all $q\leq p$ and $q+j\geq p+i$.
In particular, $L(b_1^{d-1})\in \fii$. Using the commutator relations
$[L(c_p),L(b_{p}^{(i)})]=L(b_{p+1}^{(i-1)})$ and 
$[L(c_{p+i-1}),L(b_{p}^{(i)})]=-L(b_{p}^{(i-1)})$ where $i>1$, we can reach
the conclusion that $L(b_q^{(j)})\in \fii$ for all valid $q$ and $j$.
Thus, we have shown that $\fii = \fs$ and $\fs$ has to be simple.

We summarize that $\fk_1$ has dimension $d^2$, has rank $d$, and
is a subalgebra of maximal rank 
in $\so(2d)$. In addition, it is a direct sum of a simple subalgebra and
a one-dimensional abelian subalgebra. 
We list all compact, simple Lie algebras $\fs$ of rank $k:=d-1\geq 4$:
$\su(k{+}1)$ has dimension $k^2+2k$, $\so(2k{+}1)$ has dimension $2k^2+k$,
$\spp(k)$ has dimension $2k^2+k$, $\so(2k)$ has dimension $2k^2-k$, as well as the
exceptional ones.  Note that the exceptional cases $\fg_2$, $\ff_4$, $\fe_6$, $\fe_7$, and $\fe_8$ are ruled out 
by  their respective ranks $2$, $4$, $6$, $7$, and $8$ as well as dimensions $14$, $52$, $78$, $133$, 
and $248$. We 
obtain $\fs\cong\su(d)$ and $\fk_1\cong\uu(d)$.

\section{Proof of Proposition~\ref{prop:trace-short-ranged}\label{app:VIIB}}
Here, a proof for the Proposition~\ref{prop:trace-short-ranged} of Section~\ref{shortrange} is provided.
We start in Subsection~\ref{app:VIIB_1}
by generalizing a key observation of Ref.~\cite{kraus-pra71}
(where the particular case of Proposition~\ref{prop:help-trace} when $K$ divides $L$ 
was already considered). 
This generalization is then applied in Subsection~\ref{app:VIIB_2}
where the details for the proof of  Proposition~\ref{prop:trace-short-ranged}
are given.

\subsection{Generalizing a Key Observation of Ref.~\cite{kraus-pra71}\label{app:VIIB_1}}
\begin{proposition} \label{prop:help-trace}
The trace of the product of  $U^{-K}_T$ with a tensor product of Pauli operators
$Q_i \in \{ \unity_2, \mathrm{X},\mathrm{Y},\mathrm{Z} \}$ can be computed as
\begin{equation}\label{eq:trace-cont}
\tr \left[U^{-K}_T \left(\bigotimes_{i=1}^L Q_i\right)\right]=
\prod_{p=1}^{c}\tr\left[\prod_{q=0}^{L/c-1} Q_{(qK+p) \bmod L}\right],
\end{equation}
where $c:=\gcd(K,L)$.
\end{proposition}
\begin{proof}
To simplify our calculations, let us introduce the notation $v(\ell)= (K+ \ell) \bmod L$, note
that 
$(v\circ v)(\ell)=v(v(\ell))=(2K+\ell) \bmod L$,  or more generally $v^{\circ p}(\ell)=(pK+\ell) \bmod L$.
We can now write  the action of $U_T^{-K}$ on an arbitrary standard basis vector as
\begin{equation*}
U_T^{-K}\, |n_1,\ldots, n_\ell,\ldots, n_L \rangle =
|n_{v(1)}, \ldots, n_{v(\ell)}, \ldots, n_{v(L)} \rangle .
\end{equation*}
Without loss of generality we can confine the discussion to the case where $K \le L$. 
We complete the proof, by evaluating the 
trace in Eq.~\eqref{eq:trace-cont} as
\begin{align*}
&\tr\left[U^{-K}_T \left(\bigotimes_{i=1}^L Q_i\right)\right]=
\tr\left[\left(\bigotimes_{i=1}^L Q_i\right)U^{-K}_T\right]\\
&=\sum_{\underline{n} \in \{ 0,1 \}^L}\langle n_1, \ldots, 
 n_L | \left(\bigotimes_{i=1}^L Q_i\right)U^{-K}_T
|n_1,\ldots, n_{L} \rangle\\
&=\sum_{\underline{n} \in \{ 0,1 \}^L}\langle n_1,\ldots, n_L | \left(\bigotimes_{i=1}^L Q_i\right)
|n_{v(1)}, \ldots, n_{v(L)} \rangle\\
&=\sum_{\underline{n} \in \{ 0,1 \}^L} \prod_{\ell=1}^L 
\langle n_\ell | Q_\ell
|n_{v(\ell)}\rangle\intertext{which can be further simplified to}
&\phantom{=}\hspace{1.5mm}\sum_{\underline{n} \in \{ 0,1 \}^L} \prod_{\ell=1}^L 
\langle n_\ell | Q_\ell
|n_{(K + \ell) \mod L}\rangle\\
&=\sum_{\underline{n} \in \{ 0,1 \}^L} \prod_{p=1}^c 
\prod_{q=1}^{L/c-1}  \langle n_{v^{\circ q}(p)} | Q_p | n_{v^{\circ(q+1)}(p)}\rangle\\
&=\prod_{p=1}^{c} 
\tr\left[\prod_{q=0}^{L/c-1} Q_{[(qL/c+p) \mod L]}\right]. \qedhere
\end{align*}
\end{proof}

\subsection{Details of the Proof of Proposition~\ref{prop:trace-short-ranged}\label{app:VIIB_2}}

The Lie algebras $\mathfrak{t}_{M}$ and $\mathfrak{t}_{M+1}$ 
are generated by elements of the form of
\begin{align*}
&i\sum_{q'=0}^{L-1}
U^{q'}_T \left[  \left( \bigotimes_{p=1}^{M} Q_p \right)  \otimes  \unity_2^{\otimes L-M}  
\right]U^{-q'}_T \; \;  \text{and} \\
&i\sum_{q'=0}^{L-1}
U^{q'}_T \left[  \left( \bigotimes_{p=1}^{M+1} Q_p \right)  \otimes  \unity_2^{\otimes 
L-M-1}  \right]U^{-q'}_T, 
\end{align*}
respectively. Here, we consider
all combinations of $Q_p \in \{ \unity_2, \mathrm{X},\mathrm{Y},\mathrm{Z} \}$
apart from the case when $Q_1=\unity_2$.
We introduce the notation $F(a,W):=$
\begin{align*}
&\tr\left(U^{aq M}_T i\sum_{q'=0}^{L-1}
U^{q'}_T \left[  \left( \bigotimes_{p=1}^{W} Q_p \right)  \otimes  \unity_2^{\otimes L-W}  
\right]U^{-q'}_T\right)\\
&=\tr\left( i\sum_{q'=0}^{L-1}
U^{q'}_T U^{aq M}_T  \left[  \left(\bigotimes_{p=1}^{W} Q_p \right)  \otimes  
\unity_2^{\otimes L-W}  
\right]U^{-q'}_T\right)\\
&=i\sum_{q'=0}^{L-1}\tr\left(U^{q'}_T U^{aq M}_T \left[  \left( \bigotimes_{p=1}^{W} Q_p \right)  
\otimes  \unity_2^{\otimes L-W}  \right]U^{-q'}_T\right)\\
&=i\sum_{q'=0}^{L-1}\tr\left(  U^{aq M}_T\left( \bigotimes_{p=1}^{W} Q_p \right)  \otimes  
\unity_2^{\otimes L-W} \right),
\end{align*}
where $a\in\{1,-1\}$ and $W\in\{M,M{+}1\}$.
Using Proposition~\ref{prop:help-trace}, we compute the formulas 
\begin{align*}
F(1,M)&=iL \prod_{p=1}^{M}\tr\left[Q_{p}\right] =0,\\
F(1,M{+}1)&=iL \tr\left[ Q_1 Q_{M{+}1} \right] \prod_{p=2}^{M}\tr\left[Q_{p}\right],\\
F(-1,M{+}1)&=iL \tr\left[ Q_{M{+}1} Q_1\right] \prod_{p=2}^{M}\tr\left[Q_{p}\right].
\end{align*}
It follows that the respective statements in the proposition hold for the 
generators of $\mathfrak{t}_{M}$ and $\mathfrak{t}_{M{+}1}$.
Now we prove this consequence also for any element in 
$\mathfrak{t}_{M}$ (or $\mathfrak{t}_{M{+}1}$). 
First, let us note that the elements generated must be contained in 
$[\mathfrak{t}_{M}, \mathfrak{t}_{M}]$ (or
$[\mathfrak{t}_{M{+}1}, \mathfrak{t}_{M{+}1}]$). Second, since 
all elements in $\mathfrak{t}_{M{+}1}$ (and hence in $\mathfrak{t}_M$) commute with 
$U^{qM}_T$, we have that $\tr(U^{qM}_Tih)=0$
holds for any element
$ih \in [\mathfrak{t}_{M{+}1}, \mathfrak{t}_{M{+}1}]$,
as
\begin{align*}
&\tr([ih^1_{M{+}1},ih^2_{M{+}1}]U^{qM}_T)\\
&=\tr(ih^1_{M{+}1} \, ih^2_{M{+}1} \, U^{qM}_T) -
\tr(ih^2_{M{+}1} \, ih^1_{M{+}1} \, U^{qM}_T)\\
&=\tr(ih^1_{M{+}1} \, ih^2_{M{+}1} \, U^{qM}_T) -
\tr(ih^2_{M{+}1} \, U^{qM}_T \, ih^1_{M{+}1})\\
&=\tr(ih^1_{M{+}1} \, ih^2_{M{+}1} \, U^{qM}_T) -
\tr( ih^1_{M{+}1} \, ih^2_{M{+}1} \, U^{qM}_T)=0.
\end{align*}
Thus Proposition~\ref{prop:trace-short-ranged} follows.

\section{Proof of Theorem~\ref{thm:main_nn} for $d$ Even\label{app:main_nn_even}}

Let us introduce the notation  $\NN_2$, which corresponds to 
the linear space spanned by the nearest-neighbor (and on-site) operators. 
Note that $\NN_2$ forms only a linear space and is in general not equal to 
the Lie algebra $\mathfrak{t}^f_2$ generated by its elements.
We first 
prove a fermionic generalization of Lemma~\ref{prop:trace-short-ranged}.

\begin{lemma} \label{lem:zero-trace-even}
Consider a fermionic system for which the number $d\geq 6$ of modes 
is even.
For any $ih \in \NN_2$ the 
condition $\tr(ih \, \TU^{-2})=0$ holds if $d \bmod 4 =2$, 
while $\tr(ih \, \TU^{-4})=0$ holds if $d \bmod 4 =0$.  
\end{lemma}

\begin{proof}
By definition, any element $ih  \in \NN_2$ can be written as 
$ih=\sum_{n=0}^{d-1} \TU^{n}\, ih_{12}\, \TU^{-n}$,
where $ih_{12}$ is a traceless skew-Hermitian operator acting only on 
the first two modes of the fermionic system. Therefore, $ih_{12}$ is
a linear combination of the elements $i \,m_1m_2m_3m_4$  and $m_a m_b$ where
$a,b \in \{1,2,3,4 \}$ and $a\ne b$.  We obtain that
$
\tr(ih\,  \TU^{-b})= \tr[ \sum_{n=0}^{d-1}(\TU^{n}\, ih_{12}\, \TU^{-n})\, \TU^{-b}] = 
 \sum_{n=0}^{d-1} \tr(\TU^{n}\, ih_{12}\, \TU^{-n}\, \TU^{-b})= d \, \tr(ih_{12}\, \TU^{-b})$.
If $d \bmod 4 =2$, 
we 
write out explicitly $\tr(ih\,  \TU^{-b})/d$  for $b=2$ by applying Eq.~\eqref{Trans-ferm}:
\begin{align*}
&\tr(ih\,  \TU^{-2})/d=  \tr(ih_{12}\, \TU^{-2})\\
&=\sum_{\underline{n} \in \{ 0,1 \}^d} \langle n_1, \ldots, 
  n_d | i h_{12}\, \TU^{-2}
|n_1, \ldots,  n_{d} \rangle\\
&=\sum_{\underline{n} \in \{ 0,1 \}^d }\kappa(\underline{n})\, \langle n_1, \ldots, 
  n_d | i h_{12}
|n_{3}, \ldots, n_{d}, n_{1}, n_{2} \rangle,
\end{align*}
where 
\begin{equation} \label{eq:kappa_sign}
\kappa(\underline{n})= (-1)^{(n_{1} + n_{2})(n_3 +n_4+ \cdots + n_d)}.
\end{equation}
In the sum given above, the basis vectors are orthogonal and thus most of the terms 
are zero.
The only terms with non-zero contributions can occur in the cases of 
$n_1=n_{2\ell-1}$ and $n_2=n_{2\ell}$
with $\ell \in\{1, \ldots, d/2\}$. In particular, we have 
$\kappa(n_1, n_2, n_1, n_2, \ldots, n_1, n_2)=1$ as $d/2$ is an odd  number if $d \bmod 4 =2$.
Hence we obtain that
\begin{align*}
&\tr(ih\,  \TU^{-2})/d= \tr(ih_{12}\,\TU^{-2}) \\
&=
\sum_{n_1,n_2\in\{0,1\}} \langle n_1, n_2, n_1, n_2, \ldots
  | i h_{12} 
|n_{1}, n_{2}, n_{1},n_{2}, \ldots \rangle\\
&= \sum_{n_1,n_2\in\{0,1\}} \langle n_1, n_2 | ih_{12} | n_1, n_2 \rangle=\tr(ih_{12})=0.
\end{align*}
If $d \bmod 4 =0$, 
 we can explicitly  write out the 
trace:
\begin{align*}
&\tr(ih\,  \TU^{-4})/d=  \tr(ih_{12}\,\TU^{-4})= \\
&=\sum_{\underline{n} \in \{ 0,1 \}^d} \langle n_1, \ldots, n_d | i h_{12}\, \TU^{-4}
|n_1, \ldots, n_{d} \rangle\\
&= \sum_{\underline{n} \in \{ 0,1 \}^d } \lambda(\underline{n})\, \langle n_1, \ldots, n_d | i h_{12}
|n_{5}, \ldots, n_d, n_1, \ldots, n_{4} \rangle,
\end{align*}
where 
\begin{equation} \label{eq:lambda_sign}
\lambda(\underline{n})= (-1)^{(n_{1} + n_{2}+n_3 + n_4)(n_5 +n_6+ \cdots+ n_d)}.
\end{equation}
The basis vectors in the sum are again orthogonal, and most of the terms are zero.
The only terms that can give non-zero contributions are for the cases of $n_1=n_{4\ell-3}$,
$n_2=n_{4\ell-2}$, $n_3=n_{4\ell-1}$, and $n_4=n_{4\ell}$ with $\ell \in\{1, \ldots, 
d/4\}$. 
It follows in these cases that
\begin{align}
\lambda(\underline{n}) &= (-1)^{(n_1+n_2+n_3+n_4)(d/4-1)} \nonumber\\
&=\begin{cases}
1 & \text{if  $d \bmod 8 =4$}, \\
(-1)^{(n_1+n_2+n_3+n_4)} & \text{if $d \bmod 8 =0$}. 
\end{cases}
\label{eq:sign4}
\end{align}
The notation $\chi:=n_1, n_2, n_3, n_4,n_1, n_2, n_3, n_4,\ldots,n_4$ is used, and we obtain
\begin{align}
&\tr(ih\,  \TU^{-4})/d=  \tr(ih_{12}\,\TU^{-4}) \nonumber \\
&=
\sum_{n_1,\ldots,n_4\in\{0,1\}} \lambda(\underline{n})\, \langle \chi
  | i h_{12} 
|\chi \rangle \nonumber \\
&=  \sum_{n_1,\ldots,n_4\in\{0,1\}} \lambda(\underline{n})\, 
\langle n_1, \ldots, n_4 | ih_{12} | n_1, \ldots, n_4 \rangle. \label{eq:sign4b}
\end{align}
We apply Eq.~\eqref{eq:sign4} and obtain that the Eq.~\eqref{eq:sign4b} is zero if
$d \bmod 8 =4$. We can also prove that Eq.~\eqref{eq:sign4b} is zero 
for $d \bmod 8 =0$ as Eq.~\eqref{eq:sign4b}
simplifies to
\begin{align*}
&\left[\sum_{n_1,n_2\in\{0,1\}} 
(-1)^{(n_1 + n_2)} \langle n_1 n_2| i h_{12} |n_1 n_2 \rangle\right]\\
&\times \left[\sum_{n_3,n_4\in\{0,1\}} (-1)^{(n_3 + n_4)}\right] =0.\qedhere
\end{align*}
\end{proof}

\begin{lemma} \label{lem:not_zero_trace_even}
Consider a fermionic system for which the number $d\geq 6$ of modes 
is even. 
The properties $\tr(i\hseven \, \TU^{-2}) \ne 0 $ and $\tr(i\hseven \, \TU^{-4}) \ne 0 $ 
hold for the operator $i\hseven$ of Theorem \ref{thm:main_nn}.
\end{lemma}

\begin{proof}
We proceed similarly  as in the proof of Lemma~\ref{lem:zero-trace-even}.
The operator $i\hseven$ can be written as 
$\sum_{n=0}^{d-1} \TU^{n}\, ih_{5}\, \TU^{-n}$,
where $ih_{5}:=(f^{\dagger}_1 f^{\phantom\dagger}_1 f^{\dagger}_{2} f^{\phantom\dagger}_{2}f^{\dagger}_{3} 
f^{\phantom\dagger}_{3}f^{\dagger}_{4} f^{\phantom\dagger}_{4}f^{\dagger}_{5} f^{\phantom\dagger}_{5}  - 
\unity/32)$.
Due to this particular structure of  $ih_{5}$, we can simplify the trace
$\tr(i\hseven\,  \TU^{-b})= \tr[ \sum_{n=0}^{d-1}(\TU^{n}\, ih_{5}\, \TU^{-n})\, \TU^{-b}] = 
 \sum_{n=0}^{d-1} \tr(\TU^{n}\, ih_{5}\, \TU^{-n}\, \TU^{-b})= d \, \tr(ih_{5}\, \TU^{-b})$.
Let us explicitly
write out  the trace for $b=2$ by applying Eq.~\eqref{Trans-ferm}:
\begin{align*}
&tr(i\hseven\,  \TU^{-2})/d= tr(ih_{5}\,\TU^{-2})= \\
&=\sum_{\underline{n} \in \{ 0,1 \}^d} \langle n_1, \ldots, n_d | i h_{5}\, \TU^{-2}
|n_1, \ldots, n_{d} \rangle\\
&=\sum_{\underline{n} \in \{ 0,1 \}^d } \kappa(\underline{n}) \, \langle n_1, \ldots, n_d |  ih_{5}
|n_{3}, \ldots, n_{d}, n_{1}, n_{2} \rangle \\
&=\sum_{\underline{n} \in \{ 0,1 \}^d } 
\theta(\underline{n})\,\langle n_1, \ldots, n_d |  n_{3}, \ldots, n_{d}, n_{1}, n_{2} \rangle.
\end{align*}
where $
\theta(\underline{n}):= (\delta_{n_3,1}\delta_{n_4,1}\delta_{n_5,1}\delta_{n_6,1} \delta_{n_7,1}- 1/32)\,
\kappa(\underline{n})$
and 
$\kappa(\underline{n})$ was defined in Eq.~\eqref{eq:kappa_sign}.
Most of the terms in the sum are zero as the basis vectors are orthogonal.
The only terms with non-zero contributions occur for
$n_{2\ell-1}=n_1$ and $n_{2\ell}=n_2$ with $\ell \in\{1, \ldots, d/2\}$.
If $d \bmod 4 = 2$, it follows that $\theta(\underline{n})=31/32$ for $n_1=n_2=1$, 
and $\theta(\underline{n})=-1/32$ otherwise. If $d \bmod 4 =0 $, we have
$\theta(\underline{n})=31/32$ for $n_1=n_2=1$, and $\theta(\underline{n})=1/32$ for $n_1+n_2=1$, 
and  $\theta(\underline{n})=-1/32$ for $n_1=n_2=0$.
We obtain
\begin{align*}
&\tr(i\hseven\, \TU^{-2})/d= \tr(ih_{5}\,\TU^{-2}) \\
&=
\sum_{n_1,n_2\in\{0,1\}} \theta(\underline{n})\,\langle n_1, n_2, n_1, n_2, \ldots
|n_{1}, n_{2}, n_1,n_{2}, \ldots \rangle\\
&=\begin{cases}
7/8 & \text{if  $d \bmod 4 =2$}, \\
1 & \text{if $d \bmod 4=0$}. 
\end{cases}
\end{align*}
Let us now consider the trace with $\TU^{-4}$:
\begin{align*}
&\tr(i\hseven\,  \TU^{-4})/d=  \tr(ih_{5}\,\TU^{-4})= \\
&=\sum_{\underline{n} \in \{ 0,1 \}^d} \langle n_1, \ldots, n_d | i h_{5}\, \TU^{-4}
|n_1, \ldots, n_{d} \rangle\\
&= \sum_{\underline{n} \in \{ 0,1 \}^d } \lambda(\underline{n})\,\langle n_1, \ldots, n_d | i h_{5}
|n_{5}, \ldots, n_{d}, n_{1}, \ldots, n_{4} \rangle \, \\
&=\sum_{\underline{n} \in \{ 0,1 \}^d } 
\mu(\underline{n})\,\langle n_1, \ldots, n_d | n_{5},  \ldots, n_{d}, n_{1}, \ldots, n_{4} \rangle
\end{align*}
where
$\mu(\underline{n}):= (\delta_{n_5,1}\delta_{n_6,1}\delta_{n_7,1}\delta_{n_8,1} \delta_{n_9,1}
- 1/32)\,\lambda(\underline{n})$
and  $\lambda(\underline{n})$ was defined in Eq.~\eqref{eq:lambda_sign}.
Again, most of the terms in the sum are zero as the basis vectors are orthogonal.
Provided that $d \bmod 4= 2$, the
only terms with non-zero contributions can occur in the case of
$n_{2\ell-1}=n_1$ and $n_{2\ell}=n_2$  where $\ell \in\{1, \ldots, d/2\}$.
In this case $\mu(\underline{n})=31/32$ for $n_1=n_2=1$, 
and $\mu(\underline{n})=-1/32$ otherwise. It follows that $\tr(i\hseven\,  \TU^{-4})/d=$
\begin{align*}
 \sum_{n_1,n_2\in\{0,1\}} \mu(\underline{n})\,\langle n_1, n_2, n_1, n_2, \ldots
|n_{1}, n_{2}, n_1,n_{2}, \ldots \rangle=\frac{7}{8}.
\end{align*}
If $d \bmod 4=0$, 
terms with non-zero contributions can occur for $n_{4\ell-3}=n_1$,
$n_{4\ell-2}=n_2$, $n_{4\ell-1}=n_3$, and $n_{4\ell}=n_4$ with $\ell \in\{1, \ldots, d/4\}$.
For these cases we obtain from Eq.~\eqref{eq:sign4} that
\begin{align*}
\mu(\underline{n})&=
 (\delta_{n_5,1}\delta_{n_6,1}\delta_{n_7,1}\delta_{n_8,1} \delta_{n_9,1}
- \tfrac{1}{32})\\
&\times
\begin{cases}
1 & \text{if  $d \bmod 8 =4$}, \\
(-1)^{(n_1+n_2+n_3+n_4)} & \text{if $d \bmod 8 =0$}. 
\end{cases}
\end{align*}
Using 
$\chi=n_1, n_2, n_3, n_4,n_1, n_2, n_3, n_4,\ldots,n_4$ we can simplify the trace to
\begin{align*}
&\tr(i\hseven\,  \TU^{-4})/d=  \tr(ih_{5}\,\TU^{-4}) =
\sum_{n_1,\ldots,n_4\in\{0,1\}} \mu(\underline{n})\, \langle \chi
|\chi \rangle  \\
&= \begin{cases}
1/2 & \text{if $d \bmod 8 = 4$}, \\
1 & \text{if $d \bmod 8 =0$}. \xqedhere{128.5pt}{\qed}
\end{cases} 
\end{align*}
\renewcommand{\qed}{}
\end{proof}

Now we can prove Theorem~\ref{thm:main_nn} for even $d$
as summarized in the following proposition:
\begin{proposition} \label{prop:even_case}
Consider a fermionic system for which the number $d \ge 6$ of modes
is even.
The fourth-neighbor element  $i\hseven\in \mathfrak{t}^f_5$
of Theorem \ref{thm:main_nn}
is not contained in the system algebra $\mathfrak{t}^f_2$ of nearest-neighbor interactions.
\end{proposition} 

\begin{proof}
We introduce the operator
\begin{equation*}
\CC=\begin{cases}
\TU^{-2} & \text{if $d \bmod 4=2$}, \\
\TU^{-4} & \text{if $d \bmod 4=0$}. 
\end{cases}
\end{equation*}
It follows from Lemma \ref{lem:zero-trace-even} that the equality $\tr(ih\,\CC)=0$ holds
 for any  $ih \in \NN_2$. 
Since  $\CC$ commutes with all elements of $\mathfrak{t}^f_2$ and
$\mathfrak{t}^f_2=\mathrm{span}(\NN_2, [\mathfrak{t}^f_2,\mathfrak{t}^f_2])$, we have
\begin{align*}
\tr( \CC\, [ih_1,ih_2])&=\tr(\CC\, ih_1\,ih_2) -\tr(\CC\, ih_2\,ih_1)\\
&=\tr(\CC\, ih_1\,ih_2) - \tr(ih_1\, \CC\, ih_2)\\
&=\tr(\CC\, ih_1\,ih_2)-\tr(\CC\, ih_1\,ih_2)=0.
\end{align*}
This means that 
$\tr(\CC \, ik)=0$ for all  $ik \in \mathfrak{t}^f_2$.
On the other hand, we know from Lemma~\ref{lem:not_zero_trace_even} that 
$\tr (i\hseven\, \CC) \ne 0$.
This means that $i\hseven \notin \mathfrak{t}^f_2$.
\end{proof}

\section{Proof of Theorem~\ref{thm:main_nn} for $d$ Odd\label{app:main_nn_odd}}

The proof of Theorem~\ref{thm:main_nn} for odd number of modes uses 
an expansion of the translation unitary $\TU$ by the
Fourier transformed Majorana operators, which are defined as
\begin{equation} \label{eq:F-major}
\widetilde{m}_{2k}:=i(\tilde{f}^{\phantom\dagger}_k - \tilde{f}^{\dagger}_{k}) \; \text{ and }\;
\widetilde{m}_{2k+1}:=\tilde{f}^{\phantom\dagger}_k + \tilde{f}^{\dagger}_{k} \, .
\end{equation}
Note that the operators $\tilde{f}^{\phantom\dagger}_k$ were defined in Eq.~\eqref{eq:F-trans_fermop}.
The self-adjoint operators $\widetilde{m}_x$ satisfy again the Majorana anticommutation
relations $\{ \widetilde{m}_x, \widetilde{m}_y\}= 2 \delta_{x,y} \unity$. Moreover, 
the trace of any $\widetilde{m}_x$-monomial is zero, since it is a linear combination of
Majorana monomials. The following lemma relates these operators to the translation unitary.
\begin{lemma} \label{lem:expand-u}
The  translation unitary $\TU$ can be written as 
\begin{align}
\TU&= (- i)^{d-1} \exp \left[-i \sum_{k=0}^{d-1} \tfrac{2 \pi k}{d}(\tilde{f}^{\dagger}_{k}
 \tilde{f}^{\phantom\dagger}_{k}-\tfrac{1}{2}\unity)\right] \label{eq:expand_1}\\
&= (-i)^{d-1}\exp \left[- \sum_{k=0}^{d-1} \tfrac{ \pi k}{d}\, \widetilde{m}_{2k+1} \widetilde{m}_{2k}\right]
\label{eq:expand_2}\\
&=(-i)^{d-1}\prod_{k=0}^{d-1}[\cos(\tfrac{\pi k}{d})\unity - \sin(\tfrac{\pi k}{d})\, \widetilde{m}_{2k+1}
\widetilde{m}_{2k}] \label{eq:expand_3}
\end{align}
using the Fourier-transformed operators
$\tilde{f}_k$ and $\tilde{f}^{\dagger}_k$ as well as $\widetilde{m}_{2k}$ and $\widetilde{m}_{2k+1}$.
\end{lemma}
\begin{proof}
Let us denote 
the right hand side of Eq.~\eqref{eq:expand_1} by
\begin{equation*}
\TUN:=  (-i)^{d-1}\exp \left[-i \sum_{k=0}^{d-1}\tfrac{2 \pi k}{d}(\tilde{f}^{\dagger}_{k} 
\tilde{f}^{\phantom\dagger}_{k}-\tfrac{1}{2}\unity)\right].
\end{equation*}
Applying the identity $\widetilde{m}_{2k+1}\widetilde{m}_{2k}= 
i(2\tilde{f}^{\dagger}_{k} \tilde{f}^{\phantom\dagger}_{k}- \unity) $, it follows that 
$\TUN=(-i)^{d-1} \exp(- \sum_{k=0}^{d-1}  \pi k\, \widetilde{m}_{2k+1} \widetilde{m}_{2k}/d)$.
Since the formula $[\widetilde{m}_{2k+1}\widetilde{m}_{2k}, \widetilde{m}_{2k'+1}\widetilde{m}_{2k'}]=0$ holds 
for  $k \ne k'$, we can split the exponential into the product 
$\TUN=(-i)^{d-1} \prod_{k=0}^{d-1}  \exp(- \pi k\, \widetilde{m}_{2k+1} \widetilde{m}_{2k}/d)$.
We employ $(\widetilde{m}_{2k+1} \widetilde{m}_{2k})^2=-\unity$ and obtain the 
formula 
\begin{align*}
&\exp(-\tfrac{ \pi k}{d}\, \widetilde{m}_{2k+1} \widetilde{m}_{2k})=\sum_{n=0}^{\infty} 
\tfrac{(- \pi k)^n}{n!\, d^n}(\widetilde{m}_{2k+1} \widetilde{m}_{2k})^n\\
&=\sum_{n=0}^{\infty} \tfrac{(-1)^n  (\pi k)^{2n}  } {(2n)!\, d^{2n}} \unity - 
\sum_{n=0}^{\infty} \tfrac{ (-1)^n (\pi k)^{2n+1}  } {(2n+1)!\, d^{2n+1}}\, \widetilde{m}_{2k+1} \widetilde{m}_{2k} \\
&= \cos(\tfrac{\pi k}{d})\unity - \sin(\tfrac{\pi k}{d})\, \widetilde{m}_{2k+1}\widetilde{m}_{2k}.
\end{align*}
Thus,  $\TUN$ is equal to the right hand side of Eq.~\eqref{eq:expand_3}.
Similarly, the adjoint of $\TUN$ can be written as
\begin{align*}
\TUN^{\dagger}&= i^{(d-1)} \prod_{k=0}^{d-1}  \exp( \tfrac{ \pi k}{d}\, \widetilde{m}_{2k+1} \widetilde{m}_{2k}) \\
&= i^{(d-1)}\prod_{k=0}^{d-1}\left[\cos(\tfrac{\pi k}{d})\unity + 
\sin(\tfrac{\pi k}{d})\, \widetilde{m}_{2k+1}\widetilde{m}_{2k} \right]. 
\end{align*}
The commutation relations of Eq.~\eqref{eq:f-trans-car} imply the formula
$\widetilde{m}_{2k+1} \widetilde{m}_{2k} \tilde{f}^{\dagger}_k = 
-  \tilde{f}^{\dagger}_k \widetilde{m}_{2k+1} \widetilde{m}_{2k} = 
i \tilde{f}^{\dagger}_k $. 
It follows that
\begin{align}
\TUN \tilde{f}^{\dagger}_k \TUN^{\dagger} &= 
[\cos(\tfrac{\pi k}{d})\unity - \sin(\tfrac{\pi k}{d})\, \widetilde{m}_{2k+1}\widetilde{m}_{2k}] \nonumber \\
&\times 
\tilde{f}^{\dagger}_k  [\cos(\tfrac{\pi k}{d})\unity + \sin(\tfrac{\pi k}{d}) \,
\widetilde{m}_{2k+1}\widetilde{m}_{2k}] \nonumber\\
&= e^{-2\pi i k/d} \tilde{f}^{\dagger}_k, \nonumber
\intertext{which implies that}
\TUN f^{\dagger}_n \TUN^{\dagger}&= \TUN \left(\frac{1}{\sqrt{d}} 
\sum_{k=1}^d \tilde{f}^\dagger_k e^{-2 \pi ink/d} \right) \TUN^{\dagger} \nonumber\\
& =\frac{1}{\sqrt{d}} \sum_{k=1}^d \tilde{f}^\dagger_k e^{-2 \pi i(n+1)k/d}  = f^{\dagger}_{n+1}. \label{eq:v-trans}
\end{align}
Applying the formulas
$\tilde{f}_k  | 0 \rangle =0$ and $[\tilde{f}^{\dagger}_{k} \tilde{f}^{\phantom\dagger}_{k},
\tilde{f}^{\dagger}_{k'} \tilde{f}^{\phantom\dagger}_{k'}]=0$, we conclude that 
$\exp[-i \sum_{k=0}^{d-1} 2 \pi k \tilde{f}^{\dagger}_{k} \tilde{f}^{\phantom\dagger}_{k}/d]  | 0 \rangle = |0 \rangle$.
This allows to investigate how $\TUN$ acts on the Fock vacuum $| 0 \rangle$:
\begin{align}
\TUN | 0 \rangle &= (-i)^{d-1} \exp \left[-i \sum_{k=0}^{d-1} \tfrac{2 \pi k}{d} (\tilde{f}^{\dagger}_{k} 
\tilde{f}^{\phantom\dagger}_{k}-\tfrac{1}{2}\unity)\right] | 0 \rangle \nonumber \\
&= (-i)^{d-1} e^{i \sum_{k=0}^{d-1} \tfrac{\pi k}{d}}\,
\exp \left[ -i \sum_{k=0}^{d-1} \tfrac{2 \pi k}{d} (\tilde{f}^{\dagger}_{k} \tilde{f}^{\phantom\dagger}_{k})\right] | 
0 \rangle \nonumber \\
&=  (-i)^{d-1} e^{i \tfrac{\pi}{2} (d-1)} | 0 \rangle = |0 \rangle. \label{eq:inv-vac}
\end{align}
It follows from Eqs.~\eqref{eq:v-trans} and \eqref{eq:inv-vac} that $\TUN$ satisfies 
Eq.~\eqref{Trans-ferm} if we substitute $\TUN$ for $\TU$. As
Eq.~\eqref{Trans-ferm} defines $\TU$ uniquely, $\TU=\TUN$ must hold.
\end{proof}
In the next step, we provide a polynomial of $\TU$ which multiplied by any
nearest-neighbor Hamiltonian gives an operator with zero trace (if the system is composed of 
an odd number of modes). One key observation
is that the action of the twisted reflection operator on the translation unitary is
\begin{equation} \label{eq:Refl_trans}
\eR\, \TU\, \eR^{\dagger} = \TU^{-1} \, ,
\end{equation}
which follows directly form the definition of $\eR$, see Eq.~\eqref{eq:Refl_adj}.  
Using this equation and Lemma \ref{lem:expand-u}, one can prove the following statement:
\begin{lemma} \label{lem:trace_odd_case}
Consider a fermionic system for which  the number $d \ge 5$ of modes is  odd
and introduce the operator 
\begin{equation} \label{eq:CCC}
\CCC= (-1)^{\lfloor d/4 \rfloor} (\TU^2 -\TU^{-2}) - (-1)^d  (\TU^4- \TU^{-4} ).
\end{equation}
The equality $\tr (ih \, \CCC)=0$ holds for any $ih \in \mathfrak{t}^f_2$. 
\end{lemma}

\begin{proof}
We will first prove that $\tr(v \, \CCC)=0$ holds for all $v \in \NN_2$, where
$\NN_2$ denotes the linear space spanned by the nearest-neighbor interactions
(as in Appendix \ref{app:main_nn_even}).
The equation $\eR\, \CCC\, \eR^{\dagger}=-\CCC$ follows from Eq.~\eqref{eq:Refl_trans}. On the other hand, 
Eq.~\eqref{eq:Refl_adj} implies 
that $\eR\, ih\, \eR^{\dagger}= ih$ holds for any 
$ih \in \{ ih_0, ih_{\mathrm{rh}}, ih_{\mathrm{rp}}, ih_{\mathrm{cp}}, ih_{\mathrm{int}}\}$, hence 
 $\tr(ih \, \CCC)= \tr( \eR\, ih\, \eR^{-1} \eR\, \CCC\, \eR^{-1})=-\tr(ih \, \CCC)=0$.

In order to calculate $\tr(ih_{\mathrm{ch}} \,  \CCC)$,  we first note that using Eq.~\eqref{hZZ} 
the operator $ih_{\mathrm{ch}}$ can be written as
\begin{equation}
ih_{\mathrm{ch}}= -\sum_{k=0}^{d-1} \sin(\tfrac{2 \pi k}{d})\; 
\widetilde{m}_{2k+1}\widetilde{m}_{2k} \, .
\end{equation}
Next, let us expand 
$\TU^2$  using Lemma \ref{lem:expand-u}:
\begin{align}
\TU^2 &= \prod_{k=0}^{d-1}\left[\cos(\tfrac{2\pi k}{d})\unity - \sin(\tfrac{2\pi k}{d})\, 
\widetilde{m}_{2k+1}\widetilde{m}_{2k} \right] \nonumber \\
&= \lambda_1 \unity - \lambda_1 \sum_{k=0}^{d-1}\tan(\tfrac{2\pi k}{d}) \, 
\widetilde{m}_{2k+1}\widetilde{m}_{2k}+ M_1, \label{eq:exp_for_odd}
\end{align}
where $M_1$ is a linear combination 
 of Majorana monomials of degree greater than two and
$\lambda_1:=\prod_{k=0}^{d-1}\cos\left(\tfrac{2\pi k}{d} \right)$.
Similarly, let us expand $\TU^4$:
\begin{align*}
\TU^4 &=\prod_{k=0}^{d-1}\left[\cos(\tfrac{4\pi k}{d})\unity - 
\sin(\tfrac{4\pi k}{d})\, \widetilde{m}_{2k+1}\widetilde{m}_{2k} \right]\ \\
&= \lambda_1 \unity - \lambda_1 \sum_{k=0}^{d-1}\tan(\tfrac{4\pi k}{d}) \, 
\widetilde{m}_{2k+1}\widetilde{m}_{2k}+ M_2,
\end{align*}
where  $M_2$ is a linear combination 
of Majorana monomials of degree greater than two. We employed that
$\prod_{k=0}^{d-1}\cos(\tfrac{4\pi k}{d}) = \prod_{k=0}^{d-1}\cos(\tfrac{2\pi k}{d})$ holds
for odd $d$.

We note that all monomials of Fourier-transformed Majorana operators
have zero trace and determine the traces $\tr (\TU^2  \, ih_{\mathrm{ch}})$  
and $\tr (\TU^4  \, ih_{\mathrm{ch}})$ by calculating the coefficient of $\unity$ in 
$\TU^2  \, ih_{\mathrm{ch}}$  and $\TU^4  \, ih_{\mathrm{ch}}$:
\begin{align*}
\tr(\TU^2 \, ih_{\mathrm{ch}})&= 2^d \lambda_1\sum_{k=0}^{d-1} \tan( \tfrac{2 \pi k}{d}) 
\sin( \tfrac{2 \pi k}{d}) = (-1)^d 2^d d\lambda_1, \\
 \tr(\TU^4 \, ih_{\mathrm{ch}})&=2^d \lambda_1 
\sum_{k=0}^{d-1}
\tan( \tfrac{4 \pi k}{d}) \sin( \tfrac{2 \pi k}{d})\\ &= (-1)^{\lfloor d/4 \rfloor} 2^d d \lambda_1.
\end{align*} 
Note that $\tr(ih_{\mathrm{ch}}\, \TU^{-\ell}) =\tr(\eR\, ih_{\mathrm{ch}}\, \eR^{\dagger}\eR\, \TU^{-\ell}\, 
\eR^{\dagger}) =- \tr(ih_{\mathrm{ch}}\, \TU^{\ell})$, which allows us to conclude 
$$\tr(i  h_{\mathrm{ch}} \, \CCC)=2 (-1)^{\lfloor d/4\rfloor} \tr(i h_{\mathrm{ch}}\, \TU^2)- 2(-1)^d 
\tr(i h_{\mathrm{ch}}\, \TU^4).$$ 
This implies that $\tr(\CCC \, ih_{\mathrm{ch}})=0$, and thus $\tr(v \, \CCC)=0$ holds for all $v \in \NN_2$. 
As $\CCC$ commutes with all elements of $\mathfrak{t}_2^f$, it also follows that  $\tr(ih \, \CCC)=0$ for 
any $ih \in \mathfrak{t}_2^f$.
\end{proof}

After these preparations we can prove Theorem~\ref{thm:main_nn} for odd $d$
as summarized in the following proposition:
\begin{proposition} \label{prop:odd_case}
Consider a fermionic system with  $d \ge 5$ odd modes and the
Hamiltonian $\hsodd$ of Theorem~\ref{thm:main_nn}.
The generator $i\hsodd \in \mathfrak{t}^f_4$
is not contained in the system algebra $\mathfrak{t}^f_2$ of nearest-neighbor interactions.
\end{proposition}

\begin{proof}
Using Eq.~\eqref{hZZ}, $i\hsodd$ can be written as 
\begin{equation}
i\hsodd = -\sum_{k=0}^{d-1} \sin(\tfrac{6 \pi k}{d})\; 
\widetilde{m}_{2k+1}\widetilde{m}_{2k}.
\end{equation}
Observe that $\tr(i\hsodd\, \TU^{-\ell}) =\tr(\eR\, i \hsodd\,\eR^{\dagger}\eR\, \TU^{-\ell}\, \eR^{\dagger}) =
- \tr(i \hsodd\, \TU^{\ell})$ and conclude that the formula $\tr(i \hsodd \, \CCC)=2 (-1)^{\lfloor d/4\rfloor} 
\tr(i \hsodd\, \TU^2)- 2(-1)^d \tr(i \hsodd \,\TU^4)$ holds. Now, the expansion of  $\TU$ given by 
Eq.~\eqref{eq:exp_for_odd} allows us to calculate the
trace of $i\hsodd \,  \CCC$:
\begin{align*}
\tr(i\hsodd \, \CCC)&= 2^{d+1}  (-1)^{\lfloor d/4 \rfloor}  \lambda_1 \sum_{k=0}^{d-1} \tan( \tfrac{2 \pi k}{d}) 
\sin( \tfrac{6 \pi k}{d})  \\
& - 2^{d+1} (-1)^d  \lambda_1 \sum_{k=0}^{d-1} \tan( \tfrac{4 \pi k}{d}) \sin( \tfrac{6 \pi k}{d})\\
&= 2^{d+1} (-1)^{\lfloor d/4\rfloor}  \lambda_1  (-1)^{d-1}  d\\
& - 2^{d+1} (-1)^d  \lambda_1 (-1)^{\lfloor d/4\rfloor} d\\
& = 2^{d+2}  (-1)^{\lfloor d/4 \rfloor} (-1)^{d-1} d \lambda_1 \ne 0.
\end{align*}
On the other hand, we know from Lemma~\ref{lem:trace_odd_case} that  the equality 
$\tr(\CCC \, ih)=0$ holds for any $ih \in \mathfrak{t}^f_2$. Therefore, 
$i\hsodd \notin \mathfrak{t}^f_2$.
\end{proof}


%

\end{document}